\documentclass[11pt]{article}
\usepackage{abstract}
\usepackage[utf8]{inputenc}
\usepackage{amsxtra}
\usepackage{amsthm}
\usepackage{amssymb}
\usepackage{amsfonts}
\usepackage{graphicx}
\usepackage{rotating}
\usepackage{color}
\usepackage{relsize}
\usepackage{epsfig}
\usepackage{subcaption}
\usepackage{verbatim}
\usepackage[printonlyused]{acronym}
\usepackage{setspace}
\usepackage{epstopdf}
\usepackage{authblk}
\usepackage{sectsty}
\usepackage[colorinlistoftodos]{todonotes}
\usepackage[colorlinks=true,allcolors=black]{hyperref}
\usepackage{textcomp,marvosym}
\usepackage{soul}
\usepackage{algorithm}
\usepackage{algorithmicx}
\usepackage{algpseudocode}
\usepackage{url}
\usepackage{hyperref}
\usetikzlibrary{shapes,arrows,automata,positioning}
\usepackage{mathtools}

\usepackage{booktabs}       
\usepackage{nicefrac}       
\usepackage{microtype}      
\usepackage{lipsum}		
\usepackage{graphicx}
\usepackage{caption}
\usepackage{subcaption}
\usepackage{dcolumn}
\usepackage{bm}
\usepackage{xcolor}
\usepackage{amsmath}
\usepackage{comment} 
\usepackage{amssymb}
\usepackage[T1]{fontenc}
\usepackage{algpseudocode}
\usepackage{algorithm}
\usepackage{authblk}
\usepackage[section]{placeins}
\usepackage{float}
\usepackage{cleveref}
\usepackage{natbib}

\usepackage[utf8]{inputenc}
\usepackage{graphicx}
\usepackage{newtxtext}
\usepackage{newtxmath}
\usepackage{natbib}
\usepackage{hyperref}
\usepackage{amssymb}
\usepackage{amsmath}
\usepackage{bm}
\usepackage{caption}
\usepackage{subcaption}
\usepackage{float}
\usepackage{enumitem}

\usepackage{tikz}
\usepackage{tikz-cd}
\usetikzlibrary{shapes,arrows,fit,backgrounds,positioning}
\usetikzlibrary{decorations.pathmorphing} 
\usetikzlibrary{matrix} 
\usetikzlibrary{arrows} 
\usetikzlibrary{calc} 
\usepackage{tcolorbox}
\usepackage{pgffor}

\usepackage[english]{babel}
\usepackage{cleveref}
\crefformat{figure}{#2Fig.~#1#3}
\Crefname{figure}{Fig.}{Figs.}
\newtheorem{theorem}{Theorem}

\hypersetup{
    colorlinks = true,
    urlcolor   = blue,
    citecolor  = black,
}

\newcommand{\RomanNumeralCaps}[1]
\linenumbers

\renewcommand{\l}{\left}
\renewcommand{\r}{\right}

\definecolor{Purpleee}{RGB}{91,40,224}
\newcommand{\bfv}[1]{{\bf #1}}

\newcommand{\bfvK}{\bfv{K}}
\newcommand{\bfvM}{\bfv{M}}
\newcommand{\bfvF}{\bfv{F}}

\newcommand{\bfvg}{\bfv{g}}

\newcommand{\Hil}{\mathcal{H}}
\newcommand{\Hilg}{\mathcal{H}_\bfvg}
\newcommand{\Hilgbar}{\mathcal{H}_{\bar{\bfvg}}}
\newcommand{\MM}{{\mathcal{M}}}
\newcommand{\MMbar}{{\bar{\mathcal{M}}}}

\floatstyle{boxed} 
\restylefloat{figure}

\sectionfont{\fontsize{12}{15}\selectfont}
\subsectionfont{\fontsize{10}{12}\selectfont}

\usepackage{geometry}
\title{Mori--Zwanzig mode decomposition: Comparison with time-delay embeddings}

\author{Michael Woodward$^{1}$\footnote{Email: mwoodward@lanl.gov}, Yen Ting Lin$^{1}$, Yifeng Tian$^{1}$, Christoph Hader$^{2}$,\\ Hermann Fasel$^{2}$, Daniel Livescu$^{1}$}

\affil{1. Computer, Computational and Statistical Sciences Division, Los Alamos National Laboratory, Los Alamos, NM 87544}
\affil{2. Department of Aerospace and Mechanical Engineering, University of Arizona, Tucson, AZ 85721, USA}

\begin{document}
\maketitle

\begin{abstract}

We introduce the Mori-Zwanzig Mode Decomposition (MZMD), a novel data-driven technique for efficient modal analysis of and reduced-order modeling of large-scale spatio-temporal dynamical systems. MZMD represents an extension of Dynamic Mode Decomposition (DMD) by providing an approximate closure term with MZ memory kernels accounting for how the unresolved modes of DMD interact with the resolved modes, thus addressing limitations when the state-space observables do not form a Koopman-invariant subspace. Leveraging the Mori-Zwanzig (MZ) formalism, MZMD identifies the modes and spectrum of the discrete-time Generalized Langevin Equation (GLE); an integro-differential equation that governs the dynamics of selected observables and their memory-dependent coupling with the unresolved degrees of freedom. This feature fundamentally distinguishes MZMD from time-delay embedding methods, such as Higher-Order DMD (HODMD). In this work, we derive and analyze MZMD and compare it with DMD and HODMD, using two exemplary Direct Numerical Simulation (DNS) datasets: a 2D flow over a cylinder (as validation) and laminar-turbulent boundary-layer transition over a flared cone at Mach 6.  We demonstrate that MZMD, via the addition of MZ memory terms, improves the resolution of spatio-temporal structures within the transitional/turbulent regime by the introduction of transient and periodic modes (not captured by DMD), which contain features that arise due to nonlinear mechanisms, such as the generation of the so-called hot streaks on the surface of the flared cone. Our results demonstrate that MZMD serves as an efficient generalization of DMD (reducing to DMD in the absence of memory), improves stability, and exhibits greater robustness and resistance to overfitting compared to HODMD.

\end{abstract}



\section{Introduction and motivation}

The analysis and understanding of fluid systems is crucial for many scientific fields and is ubiquitous in engineering design and control. Fluid systems exhibit complex dynamics which can be observed in a wide range of natural phenomena, such as transitional and turbulent flows \citep{pope_2011}. Nonlinearities in these systems are commonplace, posing a core challenge for analytic treatment due to strong coupling across a broad range of spatial and temporal scales. Analytical solutions, such as those to the Navier--Stokes equations, exist only in a handful of simplified cases \citep{batchelor_2000}. Therefore, in an attempt to understand fluid systems, one often relies on data collected from computational fluid dynamics and experimental measurements. As computing power continues to grow, so too does the volume of high-dimensional nonlinear dynamical systems data, increasing the demand for robust and scalable data-driven analysis techniques. 

Various techniques have been developed to analyze this vast amount of high-dimensional nonlinear dynamical systems data for uncovering and interpreting the underlying dynamics \citep{bonnet_review_2001}. Among these, modal decomposition methods have emerged as powerful tools for identifying and analyzing coherent flow features or structures present in the flow \citep{taira17_modal_analysis}. Extracting these dominant coherent structures enables a deeper understanding of the underlying physical mechanisms driving a phenomena of interest, such as laminar-turbulent transition \citep{pope_2011, lumley_book_2012}. These methods take advantage of the fact that common flow features emerge over a range of fluid flows \citep{taira_p2}. Furthermore, modal decomposition techniques can also serve as the basis for reduced-order models for prediction and control \citep{lumley_book_2012, dmd_book}.

Dynamic mode decomposition (DMD) is one such data-driven technique that has achieved wide ranging success \citep{schmid_2010}. DMD extracts large-scale spatio-temporal structures from data by computing an approximate eigendecomposition of the best fit linear operator that maps the discrete time evolution of state variables forward in time. It was later shown that DMD has connections to the Koopman operator \citep{rowley_2009}. The Koopman operator \citep{koopman_hamiltonian_1931, koopman_dynamical_1932} trades the nonlinearities of a finite dimensional dynamical system with an infinite-dimensional linear operator acting on a Hilbert space of observables. DMD attempts to approximate potentially infinitely many Koopman modes and eigenvalues by those of an approximate Koopman operator in a finite subspace spanned by a set of state variables \citep{rowley_2009}. Therefore, DMD has limitations in extracting or capturing nonlinear dynamics systems due to this observable selection \citep{taira_p2}, and may yield completely spurious results \citep{wu_dmd_issues_2021}. Consequently, DMD may miss critical nonlinear interactions, as well as transient and intermittent dynamics, yielding incomplete or misleading results in highly nonlinear regimes.

In addition to the challenges of using DMD for strongly nonlinear systems, underlying the approximate Koopman methods, is the assumption that a finite dimensional Koopman invariant subspace exists. In practice, however, finding a finite Koopman invariant subspace can be challenging, or even intractable. Therefore, a key issue with approximate Koopman learning techniques, such as DMD, is that of finding a finite set of observables that approximately closes the dynamics reasonably well. Thus, a key challenge of DMD is to resolve this closure problem, accounting for how the unresolved modes of the dynamics interact with the resolved modes. 

Variants of DMD have emerged to address the issues mentioned above by increasing the dimension of the functional basis that these linear models can represent. For example, Extended DMD (EDMD) \citep{williams_2015_edmd} includes more general observables, such as polynomial functions of the states alongside the states themselves. However, this approach becomes more computationally expensive for high-dimensional dynamical systems, can require \textit{a priori} knowledge of the system, and may lead to over-fitting. Attempts have been made to avoid this computational cost by using Kernel methods \citep{matt_williams_kernel_edmd}, that scale independently of the number higher order observable selections, however can be more prone to over-fitting.  Other methods have attempted to combine nonlinear auto-encoders with linear least squares to simultaneously identify invariant Koopman modes and operator \cite{lusch_deep_2018}. Nevertheless, in high-dimensional nonlinear dynamical systems, discovering a finite-dimensional invariant subspace can be exceedingly challenging \citep{Brunton_2016}. 

Other extensions have emerged that incorporate past history into the set of observables selection. Approaches in this direction are often inspired by Takens' time-delay embeddings \citep{takens81}, which when coupled with DMD results in methods such as the Hankel DMD \citep{hankel_dmd, kamb2020timedelay} and Higher-Order DMD (HODMD) \citep{hodmd_2017}, each shows promise in addressing the aforementioned weaknesses of DMD \citep{taira_p2}. HODMD, the extension of DMD most similar to MZMD, incorporates time delay embedding within the space of observables that are selected by using SVD (as a linear auto-encoder). In this work, HODMD is used for a comparison of MZ memory with time-delay embeddings. In the context of data-driven methods, using time-delay embeddings can result in very large concatenation of the snapshot matrices involved, such as with the Hankel DMD formulation when the delay dimension $d$ is large. While these methods do not necessitate \emph{a priori} selection of more general observables as in EDMD, they encounter similar issues for modal analysis techniques, such as significantly increased computational costs for high-dimensional dynamical systems and the potential for over-fitting \citep{cyoung_23, dhadphale_model_2024}.

The Mori--Zwanzig (MZ) formalism, originally developed in statistical mechanics \citep{mori1965transport, zwanzig1973nonlinear}, offers an alternative and complementary framework to the Koopman theory for constructing closures for reduced-order models \citep{lin2021datadriven_full}. The MZ framework provides a general procedure that includes interpretable memory kernels which can provide closures to the approximate Koopman techniques without relying on time-delay embeddings \citep{lin2021datadriven_full}. The MZ formalism provides a mathematical framework for constructing closed, non-Markovian reduced-order models of resolved variables from high-dimensional dynamical systems. This is achieved by utilizing a well defined projection operator to decompose any function of the full state space into a function of the resolved variables (observables) and its orthogonal complement (which depends on both the resolved and unresolved variables). The MZ procedure then derives the Generalized Langevin Equation (GLE) from the associated Liouville equation acting on the predefined observables. The GLE consists of a Markovian term, a memory term, and an orthogonal dynamics term. The GLE is a closed and exact system of evolutionary equations for the observables where the effects due to the unresolved dynamics are captured in the memory kernel and orthogonal dynamics. Therefore, the MZ framework offers a unique perspective on the interpretation of memory effects, which can provide closures to the approximate Koopman techniques \citep{lin2021datadriven_full}.

Several recent studies \citep{lin2021datadriven_full, maeyama2020extracting, meyer2021numerical, meyer2020non, mori2007dynamic}, established that with Mori's linear projection operator, it is possible to adopt a data-driven approach to learn the MZ operators. In addition, it was shown \citep{lin2021datadriven_full} that these methods provide higher-order and memory-dependent corrections to existing data-driven learning of the approximate Koopman operators. Notably, the recent developments by Lin et al. \citep{lin22_nn_mz, lin2021datadriven_full} have led to general data-driven techniques to approximate the MZ operators from data, with promising applications to turbulent flows \citep{tian_2021}, and transitional flows \citep{woodward_scitech, woodward_aviation}. In another line of related works, NARMAX (Nonlinear Auto-Regressive Moving Average with eXogenous input) is a data-driven method from system identification theory that includes memory, however, in general NARMAX is not equivalent to MZ, unless the Wiener projection operator is used \cite{Lin2021} (which can be used to derive a common variant of NARMAX). In particular, NARMAX methods typically assume that the noise term is Gaussian, which is not the case in general framework for MZ \citep{mori1965transport, Kubo_1966}. Other studies have considered data-driven learning of MZ operators as well, such as in the work by \citet{Curtis_2021}, which develops a Memory-Dependent Dynamic Mode Decomposition. This method derives a data-driven learning method that requires solving a nonlinear optimization problem, which is distinct from the convex optimization used in this manuscript. Additionally, data-driven MZ methods have been explored for Fluid dynamics applications, such as for discovering dynamic subgrid scale Large Eddy Simulation (LES) models using ‘dynamic-MZ-$\tau$’ model \citep{Parish2017}.

The method presented here, MZMD, represents a data-driven technique to perform modal analysis of the Generalized Langevin Equation (GLE) derived from the MZ formalism with Mori's linear projector. MZMD achieves this modal decomposition by efficiently approximating the modes and spectrum of the full state space GLE, where MZ memory is accounting for the fact that the state variables are not a Koopman invariant subspace. With the Mori's linear projection, the MZ framework offers an attractive formulation of reduced order models that can address the aforementioned issues of approximate Koopman techniques such as EDMD, Hankle DMD and HODMD, as shown later in this work. For example, since the MZ framework does not rely on time-delay embedding, the MZMD operators can be learned more efficiently. We demonstrate that MZMD efficiently generalizes DMD, where the memory terms can be seen as corrections to the DMD modes that can improve the ability to predict nonlinearities from the standard linear observable selections of DMD. Additionally, MZMD is less prone to the overfitting observed from HODMD.

The rest of the manuscript is organized as follows. For completeness and to fix notations, Sections \ref{sec:koopman_background} and \ref{sec:hodmd} provide background and general introductions to the Koopman operator, DMD, and HODMD. In Section \ref{sec:mz_background_theory} we summarize the Mori-Zwanzig formalism and data-driven MZ method developed in \citet{lin2021datadriven_full}. In Section \ref{sec:mzmd} present the MZMD method for performing modal analysis of the GLE, which can be interpreted as approximating a closure for DMD. Section \ref{sec:2d_cylinder} then provides numerical validation experiments comparing DMD, HODMD, and MZMD to the 2D flow over a cylinder at Reynolds number $Re=600$. Then in section \ref{sec:results} we compare the DMD, HODMD and MZMD approaches for wall pressure disturbance data obtained from a ``natural'' transition simulation for a flared cone at Mach 6 \citep{hader_2018}. Finally, in section \ref{sec:conclusion} we provide concluding remarks and a general discussion of comparing the differences between DMD, HODMD and MZMD.

\section{Background}
\label{sec:background}

For completeness and to fix notations, in this section we review the Koopman operator, DMD, Higher-Order DMD, and the MZ formalism. The Mori--Zwanzig mode decomposition method, the main contribution of this work, is developed in the next section.

\subsection{Koopman Operator and DMD}
\label{sec:koopman_background}
Two equivalent ways to describe autonomous dynamical systems are the state-(phase-)space description, which prescribes evolutionary equations of state-space variables $\bm x \in \mathbb{R}^N$ (or ``physical variables'')  and the Koopman description (see \citet{rowley_2009}, and \citet{schmid_2010} for a more detailed analysis). The state- or phase-space description can be stated as
\begin{equation}\label{eq:state_dynamics}
    \cfrac{{\rm d} \bm x(t)}{{\rm d} t} = \bm f (\bm x), \quad \bm x(0) = \bm x_0,
\end{equation}
where $\bm f$  is assumed to be a Lipschitz continuous function so that unique solutions $\bm x(t)$ of Equation (\ref{eq:state_dynamics}) exist. Autonomous dynamical systems of the form (\ref{eq:state_dynamics}) can describe a wide range of systems, from nonlinear systems of ODEs to spatial discretizations of PDEs, and typically result from physical conservation laws. The Koopman operator, introduced by \citet{koopman_hamiltonian_1931}, provides an alternative description of dynamical systems \citep{koopman_dynamical_1932}. The Koopman description constructs a linear evolution operator acting on the infinite dimensional Hilbert space  $\mathcal{H}$ of observables $g$. Observables $g \in \mathcal{H}$ (usually square-integrable functions $L^2$), are scalar-valued functions of the states; $g: \mathbb{R}^N \rightarrow \mathbb{R}$. The Koopman operator, $\mathcal{K}_t: \mathcal{H} \rightarrow \mathcal{H}$, is then defined using the composition operator $\circ$ as follows:
\begin{equation}\label{eq:koopman_dynamics}
    (\mathcal{K}_t g) (\bm x_0) \equiv g \circ \bm x(t, \bm x_0) = g(\bm x(t, \bm x_0)), \quad \forall g \in \mathcal{H}, \forall \bm x_0 \in \mathbb{R}^N.
\end{equation}

The linearity of the Koopman operator is inherited from the composition operator, however it acts on the infinite dimensional Hilbert space $\mathcal{H}$ of observables. Therefore, the Koopman approach trades the non-linearity of $\bm f$ acting on a finite dimensional space $\mathbb{R}^N$ with linearity of $\mathcal{K}_t$ acting on a infinite dimensional space $\mathcal{H}$. 
The linear Koopman operator $\mathcal{K}_t$ can be characterized by its eigenvalues and eigenfunction. A function $\psi:\mathbb{R}^N\rightarrow \mathbb{R}$ (or complex valued) is defined as a Koopman eigenfunction if it satisfies $\l(\mathcal{K}_t  \bm \psi\r)= e^{\lambda t} \bm \psi$. The infinitesimal generator of $\mathcal{K}_t$, $\lim_{t\downarrow 0}\l(\mathcal{K}_t-I\r)/t$ is the Liouville operator $\mathcal{L}:=\sum_{i=1}^N f_i \l(\bfv{\xi}\r) \partial_{\xi_i}$, with a dummy variable $\bfv{\xi} \in \mathbb{R}^N$, which is the Lie derivative with respect to the flow field $\bm{f}:\mathbb{R}^N \rightarrow \mathbb{R}^N$. Since $\mathcal{L} \bm \psi = \lambda \bm \psi$, the Koopman eigenfunctions are the eigenfunctions of the Liouville operator and are special ``modes'' of initial data following a coherent evolution $\psi(x,t)=\psi\l(x,0\r) \exp\l(\lambda t\r)$ \citep{rowley2009SpectralAnalysisNonlinear}.

In practice, the approximate Koopman techniques \citep{schmid_2010, williams_2015_edmd} are often performed in the discrete representation Eq. (\ref{eq:disc_dynamics}); where in this work we consider a fixed time step $\Delta t$. This can be connected to the continuous time dynamics through the continuously differentiable flow map $\mathbf{F}_t : \mathbb{R}^N \rightarrow \mathbb{R}^N$, which at a specified time $t$ maps the state $\bm x(t_0)$ to $\bm x(t + t_0)$, and can be found according to:
\begin{equation}\label{eq:flow_map}
    \mathbf{F}_t (\bm x(t_0)) \triangleq \bm x(t_0+t) = \mathbf{x}\left(t_0\right) + \int_{t_0}^{t_0 + t} \bm f(\bm x(\tau)) {\rm d} \tau.
\end{equation}
The flow map therefore induces the discrete-time dynamical system
\begin{equation}\label{eq:disc_dynamics}
    \bm x_{n+1} = \bm F_{\Delta t} (\bm x_n), \quad \bm x(0) = \bm x_0,
\end{equation}
where $\bm x_{n} = \bm x(n \Delta t)$ and $\bm F_{\Delta}$ is the discrete time flow map. This discrete-time setting is a natural framework when considering experimental or numerical data-sets. For example, solving the 3D compressible Navier Stokes Equations with Direct Numerical Simulation (DNS) will result in the collection of discrete state variables (velocity, pressure, density and internal energy) at at $N$ grid points over $K$ time steps. For this study, we only consider DNS sampled at a fixed time step $\Delta t$.  It is worth noting at this point that for practical turbulent flows involving turbulence, $N$ can become enormous $N \sim \mathcal{O}(10^{13})$ in order to resolve all necessary scales \citep{yeung_gpu-enabled_2024} (up to the integral scale and down to the Kolmogorov micro-scale).

The Koopman framework provides an attractive alternative to equation (\ref{eq:disc_dynamics}) as it offers a linear representation. Furthermore, it is valid far from fixed points and periodic orbits, as opposed to linearization techniques. However, although the Koopman operator is linear, it acts on an infinite-dimensional Hilbert space of observables $\mathcal{H}$. Therefore, in practice, and especially for high-dimensional nonlinear dynamical systems, this creates a significant challenge, namely identifying a finite and \textit{tractable} subspace of observables so that the finite rank Koopman operator is invariant to this subspace, i.e. closed.

A Koopman invariant subspace $\mathcal{S} \subseteq \mathcal{H}$, if it exits, is a subspace of $\mathcal{H}$ spanned by observables $\{g_1, g_2, ... g_q\}$ so that $\mathcal{K} g \in \mathcal{S}$ for any observable $g \in \mathcal{S}$. If this subspace exits, then there exists a finite rank Koopman operator $\bm K \in \mathbb{R}^{q \times q}$ restricted to $\mathcal{S} = {\rm span}(\{g_1, g_2, ... g_q\})$.  If such a matrix representation exists, it is possible to define a closed (and exact) linear system $\bm g_{n+1} = \bm K \bm g_n$, where $\bm g \in \mathcal{S}$. However, in practice, often one can only have an approximation with some nontrivial residual $\bm R$ so that $\bm g_{n+1} = \bm K \bm g_n + \bm R$. The Mori-Zwanzig offers an attractive framework for understanding this residual, as it provides a formal procedure using projection operators for closing these systems by accounting for non-local in time ``hysteresis'' effects as well as orthogonal dynamics. This will be described in more detail in the Section \ref{sec:mz_background}.

The goal of the approximate Koopman learning techniques, such as EDMD \citep{williams_2015_edmd}, is to identify a finite rank  matrix approximation $\bm K$ of the Koopman operator given a uniform sequence of $T$ snapshot measurements $\{\bm g_{n \Delta t}\left(x_0\right) = \bm g\left(\bm x\left(n\Delta t; x_0 \right)\right) \in \mathbb{R}^M:  n \in \{0, 1, ..., T-1\}\}$. This is done by obtaining the best fit linear map that evolves the snapshots forward in time to the following consecutive snapshots with a minimal squared residual error. Operationally, this is carried out by first constructing the snapshot matrices:
\begin{equation} \label{eq:snapshots}
\bm G_0 = \begin{bmatrix}
    \bm g(0) & \bm g(\Delta t) & ... & \bm g((T-2) \Delta t)
\end{bmatrix}, \hspace{1mm}
    \bm G_1 = \begin{bmatrix}
    \bm g(\Delta t) & \bm g(2\Delta t) & ... &  \bm g((T-1) \Delta t) \end{bmatrix},
\end{equation}
then, by minimizing the Frobenius norm squared residual error $\epsilon^2 = ||\bm G_1 - \bm K \bm G_0||_F^2$. The optimal matrix $\bm K$ that approximately evolves the snapshot matrices $\bm G_1 \approx \bm K \bm G_0$ is found to be $\bm K = \bm G_1 \bm G_0^{\dagger}$, where $\bm G_0^{\dagger}$ is the pseudo-inverse that is computed via Singular Value Decomposition (SVD).  The matrix $\bm K$ is a best-fit linear operator in the sense that it minimizes the Frobenius norm error. In practice, with EDMD, one tries to select observables spanning nonlinear functions of the state variables in an attempt to provide a finite Koopman invariant subspace. However, for high dimensional dynamical systems, such as DNS of 3D flow fields with high Reynolds numbers, this can drastically increase the computationally cost and requires a priori knowledge of the system.

Dynamic Mode Decomposition (DMD) is based on the linear observable function that selects the state space variables $g_i(\bm x) = x_i$. For high dimensional dynamical systems with $M\gg T$, this results in tall skinny matrices, and in order to avoid the memory intensive operation $\bm G_1 \bm G_0^{\dagger}$, the standard DMD algorithm first performs a low rank approximation of the snapshot matrix $\bm G_0$ by projecting onto the dominant POD modes (see algorithm \ref{alg:dmd}). Once the eigendecomposition of $\bm K$ is approximated with the DMD algorithm, which provides the eigenpair $(\lambda, \bm \psi)$, then the future state evolution can be described via:
\begin{equation}\label{eq:dmd_expansion}
    \bm x(t_n) \approx \sum_{m=1}^N a_m \bm \psi_m e^{(\delta_m + i \omega_m)n\Delta t}, 
\end{equation}
where $\delta_m + i \omega_m = ({\Delta t}^{-1}) \log(\lambda_m)$. However, since DMD is based on linear measurements ($g_i = x_i$) of the system, and models these observables with a linear map, it struggles to approximate the Koopman operator for strongly nonlinear systems \citep{williams_2015_edmd, wu_dmd_issues_2021}.

\subsection{Higher-Order DMD}
\label{sec:hodmd}

We now review the HODMD method, first introduced in the work of \citet{hodmd_2017}. HODMD is an extension of DMD seeking to learn a model of the state-space variables $\bm x$ of the form:
\begin{equation}\label{eq:hodmd_assumption}
    \bm x_{n+d} = \bm R_0 \bm x_{n} + \bm R_2 \bm x_{n+1} + ... + \bm R_d \bm x_{n+d}, 
\end{equation}
assuming ``higher order Koopman'' term $ \bm R_{i>0}$. 
This assumption is made in order to extend the DMD method to capture more general expansions of the form 
\begin{equation}\label{eq:hodmd_expansion}
    \bm x(t_n) = \sum_{m=1}^M a_m \bm \psi_m e^{(\delta_m + i \omega_m)n\Delta t}, 
\end{equation}
to dynamical systems in which the spatial complexity $N$ may be less than the spectral complexity $M$. The spatial complexity $N$ is defined as the dimension of the subspace generated by the $M$ DMD-modes, i.e. $N = dim(span\{\bm \psi_1, ..., \bm \psi_M\})$. The HODMD assumption allows for $M>N$, which occurs when the spatial dimension is less than the number of dynamic modes present ($N$ may also be inferred from the truncated singular value decomposition (SVD)). As established in \citep{hodmd_2017}, the standard DMD approach can only capture expansions with $M = N$, and for $M > N$ the general expansion of the form (\ref{eq:hodmd_expansion}) is equivalent to the more general assumption (\ref{eq:hodmd_assumption}) \citep{hodmd_2017}. For example, $N=3$ in the chaotic Lorenz '63 system, but due to its chaotic nature, the system involves infinite spectral complexity $M$. It is worth noting here that both EDMD and Hankel-DMD also increase the spatial complexity by increasing the number of observables and, therefore, by the above reasoning will also increase the spectral complexity over the standard DMD algorithm. 

HODMD increases the spectral complexity by relying on time-delay embeddings in the observable space (as seen in the Hankel matrix \ref{eq:hankle_matrix}). HODMD then uses the DMD algorithm with an additional SVD truncation of the Hankel matrix of time delayed observables, which increases the spectral complexity but also increases the computational cost. To reduce this cost, a tensor based decomposition \citep{li_hodmd_tensor2023}) method has been developed for applications to high-dimensional dynamical systems. 

Using the time delay coordinates $\bm z_k = [\bm x_k, \bm x_{k+1}, ..., \bm x_{k+d}]^T$, equation (\ref{eq:hodmd_assumption}) can be written as $\bm z_{k+1} = \tilde{\bm R} \bm z_k $, where
\begin{equation}\label{eq:hodmd_comp}
    \tilde{\bm R} = \begin{bmatrix}
    \bm 0 & \bm I & \bm 0 & \bm 0 & ... & \bm 0 \\
    \bm 0 & \bm 0 & \bm I & \bm 0 & ... & \bm 0 \\
      & & & \ddots & & \\
    \bm 0 & \bm 0 & \bm 0 & \bm 0 & ... & \bm I \\
    \bm R_0 & \bm R_1 & \bm R_2 & ... &  \bm R_{d-1} & \bm R_d \\
    \end{bmatrix}
\end{equation}
is the block companion matrix (which is also referred to as the modified Koopman matrix). Overall, the HODMD algorithm described in \citep{hodmd_2017} roughly applies the standard DMD algorithm to approximate the eigendecomposition of  $\tilde{\bm R}$ on a further SVD truncated Hankel matrix (see equation \ref{eq:hankle_matrix}). This Hankel matrix is defined by the time delay embedded observables to form an approximation $\hat{\bm R}$ of $\tilde{\bm R}$, then HODMD performs an eigendecomposition on $\hat{\bm R}$.  When applied to the time-delay data, this approach does not guarantee the preservation of the companion structure \eqref{eq:hodmd_comp} and does not attempt to identify individual terms $\bm R_i$. The additional SVD truncation of the Hankel matrix can break the original structure Eq.~\eqref{eq:hodmd_assumption}, as detailed below.

In what follows, we provide a general outline and description of the HODMD algorithm, as described in \citep{hodmd_2017}, in order to fix notations. Given the snapshot matrix $\bm X = [\bm x_0, ..., \bm x_{T}]$ with $T+1$ equally spaced state measurements, the first truncation is performed to compute the rank-$r_1$ approximation $\bm X \approx \bm U_{r_1} \bm \Sigma_{r_1} \bm V_{r_1}^*$ via the first SVD. From this, the low rank observable matrix $\bm G$ is formed by projecting $\bm X$ onto the left singular vectors (POD modes), which is  $\bm G = \bm U_{r_1}^* \bm X \approx \bm \Sigma_{r_1} \bm V_{r_1}^* = [\bm g_0, ..., \bm g_{T}]$. Using time delay embeddings, this is then stacked into the Hankel matrix 
\begin{equation}\label{eq:hankle_matrix}
    \tilde{\bm Z} = \begin{bmatrix}
    \bm g_0 & \bm g_1 & \bm g_2 &  ... & \bm g_{T-d} \\
    \bm g_1 & \bm g_2 & \bm g_3 &  ... & \bm g_{T-d+1} \\
      & & & \ddots & & \\
    \bm g_{d-1} & \bm g_d & \bm g_{d+1} &  ... & \bm g_{T-1} \\
    \bm g_d & \bm g_{d+1} & \bm g_{d+2} & ...  & \bm g_T \\
    \end{bmatrix}. 
\end{equation}
The Hankel matrix is then further truncated via the SVD by a rank-$r_2$ approximation $\tilde{\bm Z} \approx \bm U_{r_2} \bm \Sigma_{r_2} \bm V_{r_2}^*$. For each truncation the rank $r_j$  for  $j \in \{1, 2\}$, i.e. the number of retained modes, can be chosen so that the relative root mean squared error $E(r_j)$ is less than the selected threshold $\epsilon_j$
$$ E(r_j) = \cfrac{\sum_{i=r_j+1}^{R} \sigma_i}{\sum_{i=1}^R \sigma_i} \leq \epsilon_j,$$
where $\sigma_i$ is the $i^{\text{th}}$ singular value, $R=\min \{\dim(\tilde{\bm Z}) \}$. Now, the low rank approximation of the time delayed Hankel matrix $\tilde{\bm Z}$ is written as $\hat{\bm G} = \bm U_{r_2}^* \tilde{\bm Z} \approx \bm \Sigma_{r_2} \bm V_{r_2}^*$. From this, the standard DMD algorithm is used to solve the least square problem for $\hat{\bm R} := \hat{\bm G}_1 \hat{\bm G}_0^{\dagger}$, where the pseudo inverse $\dagger$ is computed via SVD and the subscripts coincide with those used in equation \eqref{eq:snapshots}. At this point, an eigenvector decomposition is performed in this twice reduced space $\hat{\bm R} = \bm W \bm \Lambda \bm W^{-1}$. To select the appropriately sized modes $\bm \psi_i \in \mathbb{R}^N$ used in Eq.~\ref{eq:hodmd_expansion}, one needs to consider the dynamics expressed via the block companion matrix. Since it is the last row of the block companion matrix that contributes to the dynamics of $\bm x$, the HODMD modes are selected as $\bm \Psi = \bm U_r \bm P_d \bm W$, where $\bm P_d = [\bm 0, ..., \bm I]$, where $\bm I$, $\bm 0$ are appropriately sized identity and zero matrices respectively.

\subsection{Mori--Zwanzig Formalism}
\label{sec:mz_background_theory}

Now we provide an overview of the key aspects of the Mori–Zwanzig formalism, highlighting the Generalized Langevin Equation (GLE), the primary result of this projector-based framework. We review the fundamental components of the MZ formalism, including the memory kernel, the projection operator, orthogonal dynamics, and the generalized fluctuation-dissipation (GFD) relation. We also outline the general procedure for learning the MZ operators in the discrete-time setting. For a more comprehensive discussion, including detailed operator algebra derivations, we refer the reader to \citep{zwanzig1973nonlinear, zwanzig_nonequilibrium_2001, mori1965transport, chorin_optimal_2002, li2017computing, Lin2021}. Additionally, \citet{lin2021datadriven_full} offers an alternative derivation of the GLE based on Koopman eigenfunctions, providing a geometric interpretation.

Conventionally, the goal of the Mori--Zwanzig procedure is to construct evolutionary equations for a subset of components of the phase-space variables $\hat{\bm x}:=\left\{x_i\right\}_{i=1}^M$, $M<N$, referred to as the resolved variables (observables). For example, these can be the observables which we can measure as the dynamics move forward in time. Despite this standard choice of using the components of the state ($g(\bm x) = x_i$) to determine the resolved and under-resolved observables, $g$ can be any $L^2-$integrable function of the state. Mori--Zwanzig formalism then proceeds with a specified projection operator $\mathcal{P}$, which maps a function of the full-space configuration, $g:\mathbb{R}^N \rightarrow \mathbb{R}$, to a function of only the resolved observables $\mathcal{P} g:\mathbb{R}^M \rightarrow \mathbb{R}$. An operator algebraic derivation results in the generalized Langevin equation (GLE) \citep{zwanzig1973nonlinear, chorin_optimal_2002}:
\begin{equation}
    \frac{d}{d t} \hat{x}_i(t, \bm x_0)
    =M_i(\hat{\bm x}(t,\bm x_0)) - \int_0^t K_i(\hat{\bm x}(s, \bm x_0), t-s) d s + F_i(t, \bm x_0),
\end{equation}
describing the evolution of the resolved components given an initial condition $\bm x_0$. The GLE is a closed and exact system for the resolved variables, although it is now nonlocal in time and contains an orthogonal dynamics term $F_i$ (often refereed to as noise term) that depends on time and the full state $\bm x_0$ at $t = 0$.

As developed in \citep{lin2021datadriven_full}, the Mori--Zwanzig framework can be derived using the Koopman representation of the dynamics. We provide an overview here as it gives a convenient interpretation of the MZ formalism, and later MZMD, in the context of approximate Koopman methods.  Given an initial condition $\bm x_0$, we would like to describe how the resolved variables evolve in time $g_i(t; \bm x_0)$. As opposed to decomposing a function into the Koopman eigenfunctions \citep{rowley_2009}, the Mori--Zwanzig formalism utilizes the inner product to decompose the Hilbert space $\mathcal{H}$ of solutions to Eq.~\eqref{eq:state_dynamics} into a subspace linearly spanned by the set of observables $\Hilg:=\text{Span}(\MM)$,  where $\mathcal{M}:=\{g_i\}_{i=1}^M$, and an orthogonal subspace $\Hilgbar=\l\{\bar{g} \in L^2: \l\langle \bar{g} , g_i\r\rangle=0, g_i \in \MM \r\}$. One can then construct a complete set of basis functions in $\Hil$, with a natural choice of using $\MM$ as the set of basis functions in $\Hilg$. An orthogonal space can be constructed from the Koopman eigenfunctions and is denoted by $\MMbar:=\l\{\bar{g}_i\r\}_{i=1}^\infty$. From this, an evolutionary equation can be derived for the observables $\bfvg\l(t\r)$ \citep{mori1965transport, lin2021datadriven_full}
\begin{equation}
\frac{d}{d t} \bfvg\l(t\r) = \bfvM \cdot \bfvg\l(t\r) - \int_{0}^t \bfvK\l(t-s\r)\cdot \bfvg\l(s\r)d s + \bfvF\l(t\r).
\label{eq:GLE}
\end{equation}

This Koopman representation of the MZ formalism highlights the core intuition underlying the approach. Since we are only resolving a subset of observables, $\bfvg(t) \equiv \bfvg_\MM(t)$, from the full system dynamics, the influence of the remaining observables, $\bfvg_\MMbar(t)$, cannot be directly observed. Instead, their effect is accounted for through the non-Markovian term with memory kernel $\bfvK$ and orthogonal dynamics $\bm \bfvF\l(t\r)$. Therefore, the memory and the orthogonal dynamics exist only because we have an incomplete observable set in $\Hil$. The Markovian term $\bfvM$ is the instantaneous configuration of the set of observables applied to the physical-space configuration at time $t$ and the non-Markovian temporal convolution is a delayed impact of the set of observables applied to the physical-space variables at an earlier time $s<t$. Both these terms depend only on the resolved observables $\bfvg_\MM$ at time $t$. However, the orthogonal dynamics is induced by the initial setting of the under-resolved observables, $\bfvg_\MMbar(0)$, which may not be known. Eq.~\eqref{eq:GLE} is exact if one knows both $\bfvg(0)$ and $\bfvg_\MMbar(0)$, in which case, the system is fully resolved. Unfortunately, we do not have direct access to $\bfvg_\MMbar(0)$ as they are under-resolved observables, and one has to postulate their configurations in practice.

Ultimately, Eq.~\eqref{eq:GLE} reveals that the evolution of the resolved observables $\bfvg(t)$ depends on three factors: (1) their instantaneous configuration, (2) their past history, and (3) an orthogonal contribution (external driving force) arising from the initial conditions of the orthogonal observables. Notably, both the memory and orthogonal contributions vanish if the dynamics are closed within the space $\MM$, which corresponds to having a complete set of observables to describe the full system (or a Koopman invariant subspace). Thus, memory effects and orthogonal dynamics arise only due to the incompleteness of the observable set in $\Hil$, and can therefore serve as a natural framework for developing closure terms for approximate Koopman methods. 

An essential aspect of the Mori-Zwanzig formalism, often overlooked in modeling papers, is that the memory kernel and orthogonal dynamics are not independent. With an appropriately chosen inner product (one that results in an anti-self-adjoint Liouville operator with respect to the chosen inner product), there is a relationship between the memory kernel $\bfvK$ and the two-time corelation of the orthogonal dynamics $\bfvF$, commonly known as the generalized fluctuation-dissipation (GFD) relationship \citep{Kubo_1966, lin2021datadriven_full}:

\begin{equation}
\label{eq:gfdc}
\bfvK(s) = \left< \bfvF(s), \bfvF^T(0) \right> \bm C^{-1}(0),
\end{equation}
where $\bm C (0) = \left< \bm g(0), \bm g^T(0) \right>$ is the expected auto-correlation of the observables with respect to an initial condition. 

The GFD provides a foundational tool in statistical mechanics for understanding irreversible processes and non-equilibrium phenomena. The GFD describes how a linear response of a given system to an external perturbation is expressed in terms of fluctuation properties of the system in thermal equilibrium. For example, in Brownian motion, random impacts of surrounding molecules generally cause two kinds of effects: (1) they act as a random driving force on the Brownian particle causing irregular motion, and (2) they give rise to the frictional force from a forced motion. In this context, GFD provides a relationship between the random force and the frictional force and shows that the frictional force is frequency-dependent (relating to noise with a delay in time), so that the random force cannot be white noise \citep{Kubo_1966}.

\subsection{Data-Driven Mori--Zwanzig}
\label{sec:mz_background}

We now summarize the recent data-driven Mori-Zwanzig method for learning the MZ operators from data. Following the work by \citet{lin2021datadriven_full},  we consider the discrete-time autonomous deterministic dynamical system where the states $\bm x(t) \in \mathbb{R}^D$ evolve according to \eqref{eq:disc_dynamics}. Next, we seek evolutionary equations for a set of observables $\bm g \in \mathbb{R}^r, r < D$, with $g_i : \mathbb{R}^D \rightarrow \mathbb{R}$, $i = 1,...,r$. The discrete-time GLE (see \cite{lin2021datadriven_full, Lin2021, darve_09, Gilani2021, She_2023} for detailed derivations), prescribes the exact and closed set of non-Markovian evolutionary equations for the observables given any initial condition of states $\bm x_0$ as:
\begin{equation}\label{eq:disc_gle}
    \bm g_{n+1}(\bm x_0) = \bm{\mathit{\Omega}}_0 (\bm g_{n}(\bm x_0)) + \sum_{l=1}^n \bm{\mathit{\Omega}}_l(\bm g_{n-l}(\bm x_0)) + \bm W_n(\bm x_0),
\end{equation}
where $\bm g_n:\mathbb{R}^D \rightarrow \mathbb{R}^r$ is the $r \times 1$ vector of functions of the initial state $\bm x_0$ so that $\bm g_n(\bm x_0)\equiv \bm g(\bm x(n \Delta t; \bm x_0)) \equiv \bm g(\bm F^n(\bm x_0))$. The discrete time GLE (Eq. \ref{eq:disc_gle}) states that the vector of observables at time $n+1$
evolves (and is decomposed) according to three parts: (1) a \textit{Markovian} operator: $\bm{\mathit{\Omega}}_0: \mathbb{R}^r \rightarrow \mathbb{R}^r$ which only depends on the observables at the previous time step ($n$), (2) the \textit{memory kernel}: the series of operators $\bm{\mathit{\Omega}}_l: \mathbb{R}^r \rightarrow \mathbb{R}^r$ depending on observables with a time lag $l$, and (3) the \textit{orthogonal dynamics}: $\bm W_n : \mathbb{R}^D \rightarrow \mathbb{R}^r$ depending on the full initial state $\bm x_0$. The above GLE is general for any projection operator (see \cite{lin22_nn_mz} for a detailed discussion). 

Using Mori's linear projection \citep{mori1965transport}, which employs the equipped inner product in the $L^2$ Hilbert space to define the functional projection, results in linear transformations for the Markovian term $\bm{\mathit{\Omega}}_0 (\bm g_n(\bm x_0)) = {\bm \Omega}_0 \bm g_n(\bm x_0)$, and memory kernel $\bm{\mathit{\Omega}}_l(\bm g_{n-l}(\bm x_0)) = {\bm \Omega}_l \bm g_{n-l}(\bm x_0) $, where ${\boldsymbol{\Omega}}_l$'s are $r\times r$ matrices \citep{lin2021datadriven_full}. In this manuscript, we assume that, in the process of learning the MZ operators with the convex optimization scheme derived in \cite{lin2021datadriven_full}, $\bm W_n$ is a small and negligible residual term. This is equivalent to considering the projected dynamics of Eq.~ \eqref{eq:disc_gle}, since the projection operator is orthogonal to $\bm W_n$. The algorithm for extracting the memory kernel, as well as the MZMD algorithm, is described in the Appendix for completeness (see Algorithm \ref{alg:mzmd}).

Given an inner product, the Mori projection operator $\mathcal{P}$ projects any function (of the initial condition $\bm x_0$) $f \in \Hil$, onto the subspace $\Hilg:=\text{span}\l(\MM\r)=\text{span}\l(\l\{g_i\r\}_{i=1}^M\r)$. Mori's projection operator is defined by 

\begin{equation}
\label{projection}
 \mathcal{P} f := \left< f, \bm g^T \right> \cdot \left<\bm g, \bm g^T \right> ^{-1} \cdot \bm g.
 \end{equation}

The GFD relation Eq.~\eqref{eq:gfdc} is a result of the MZ formalism, and is thus a necessary condition that must be satisfied for any data-driven MZ method. With the projection operator defined by Eq. (\ref{projection}), the discrete form of the GFD can be stated as
\begin{equation}\label{eq:gfd}
    \bm \Omega_n \bm g = \mathcal{P} \mathcal{K} \bm W_{n-1}, \quad \forall n\in \mathbb{N},
\end{equation}
where $\mathcal{K}$ is the discrete time Koopman operator. The operators $\bm \Omega_l$ 
and $\bm W_n$ depend on the choice of the projection operator $\mathcal{P}$, the choice
of the vectorized observable $\bm g$ , and the finite-time $(\Delta)$ Koopman operator $\mathcal{K}$. In the data-driven MZ methods \citep{lin2021datadriven_full, lin22_nn_mz} (summarized in Algorithm \ref{alg:mzmd}), the GFD relation is enforced by construction. In summary, there are three key elements that define the data-driven MZ procedure:
\begin{enumerate}[noitemsep,topsep=0pt]
    \item[] (i) learns the operators that satisfy the GLE,
    \item[] (ii)  uses a well defined projection operator,
    \item[] (iii)  satisfies the GFD.
\end{enumerate}
In this work, we use Mori's linear projector, which represents a direct generalization of the DMD approaches; however, these concepts can be extended to non-linear projection to learn non-linear MZ operators \citep{lin22_nn_mz}.


\section{Derivation and interpretation of the proposed MZMD method}
\label{sec:mzmd}

In this section, we present the MZMD method for performing modal analysis of the discrete-time GLE described above. Overall, MZMD can be interpreted as approximating a closure for DMD, where MZ memory accounts for the fact that the state variables selected as observables with DMD (most likely) do not form a Koopman invariant subspace. MZMD achieves this modal analysis by efficiently approximating the modes and spectrum of the full state-space GLE. MZMD is built upon the data-driven MZ approach described above, in which the Mori projector is used, and the linear operators of the GLE are learned so that the GFD is satisfied. Similarly to DMD, we first apply an SVD-based compression to the snapshot matrix $\bm X = [\bm x_0, ..., \bm x_T]$, a necessary step for high-dimensional dynamical systems. This provides a tractable method for extracting the modes and spectrum of the state space GLE from a reduced set of observables defined by projecting the high-dimensional state variables onto the POD modes. This SVD compression enables a low rank approximation of the snapshot matrices for tractable computations of the two-time covariances found in Algorithm (\ref{alg:mzmd}).

Rather than starting with the standard DMD assumption that the evolution of state variables is linear (i.e., $\bm x_{n+1} = \bm A \bm x_n$), with MZMD, we start from the derived discrete-time GLE (\ref{eq:disc_gle}) and assume that the orthogonal dynamics term can be made negligible by minimization \citep{lin2021datadriven_full}. This is formalized with the Mori's projector and $N$ state observables selected as $\pi_i \triangleq \pi_i\left(\bm x \right) = x_i$, similar to DMD. This truncated GLE, described in more detail below, provides a natural approach for approximating a closure for DMD using MZ memory terms, which appear when the observables do not form a Koopman invariant subspace (as described in more detail in the previous section). Otherwise, the matrix $\bm \Omega_0$ would be sufficient to completely describe the dynamics, as it would represent the finite-rank Koopman operator. 

Within the MZMD framework, memory terms $\mathbf{\Omega}_{i>0}$ are introduced to capture the influence of unresolved variables of DMD on the resolved ones. Assuming that only $k$ memory terms are sufficient so that the orthogonal dynamics are negligible (or minimized), equation \eqref{eq:disc_gle} becomes
\begin{equation} \label{eq_mz}
    \bm x_{n+1} = \mathbf{\Omega}_0^{(x)}  \bm x_n + ... + \mathbf{\Omega}_k^{(x)} \bm x_{n-k},
\end{equation}
where the linear MZ operators $\mathbf{\Omega}_i^{(x)}$ are acting in state space. Eq.~ \eqref{eq_mz} can also be considered as the projected dynamics of the GLE with finite memory, since $\mathcal{P} \bm W_n = 0$. The goal of this work is to approximate the modes and eigenvalues of the discrete-time GLE \eqref{eq_mz} to perform the analysis of high-dimensional nonlinear dynamical systems, and to improve the existing DMD method. 

Similarly to DMD, in order to tackle high-dimensional systems, we first perform a low-rank approximation of the snapshot matrix. This may also be interpreted as observable selection using an SVD-based compression of the snapshot matrix $\bm X \approx \bm U_r \bm \Sigma_r \bm V_r^*$, to obtain the observable matrices $\bm G = \bm U_r^* \bm X \approx \Sigma_r \bm V_r^*$, where $\bm U_r \in \mathbb{R}^{N \times r}$ is the matrix formed by the left singular vectors, $\Sigma_r \in \mathbb{R}^{r \times r}$ is the diagonal matrix of singular values arranged in descending order, and $\bm V_r \in \mathbb{R}^{(T+1) \times r}$ is the matrix formed by the right singular vectors (where $r_{max} = \min \{T+1, N\}$ is the maximal $r$ that can be used and $T+1 < N$ is common in the case of high-dimensional systems). The reduced observables are then $ \bm g(\bm x_n) \triangleq \bm U_r^{*} \bm x_n$, and represent the state variables projected onto the POD modes, or the left singular vectors. The method of snapshots can be used for an efficient computation of the POD modes \citep{method_of_snapshots}. Multiplying both sides of equation \eqref{eq_mz} by $\bm U_r^{*}$ and expressing $\bm x_i = \bm U_r \bm g_i$, we obtain
$$
\bm U_r^{*} \bm x_{n+1} = \bm U_r^{*} \mathbf{\Omega}_0^{(x)} \bm U_r \bm g_n + ... + \bm U_r^{*} \mathbf{\Omega}_k^{(x)} \bm U_r \bm g_{n-k}.
$$
This becomes
\begin{equation} \label{eq_mz_obs}
    \bm g_{n+1} = \bm \Omega_0^{(g)} \bm g_n + ... + \bm \Omega_k^{(g)}  \bm g_{n-k},
\end{equation}
where $\bm \Omega_i^{(g)} = \bm U_r^{*} \mathbf{\Omega}_i^{(x)} \bm U_r$ is the memory kernel $\mathbf{\Omega}_i^{(x)}$ projected onto the POD modes. In this work, we establish the relationship between the modes of \eqref{eq_mz} and \eqref{eq_mz_obs} in two parts; first a full SVD ($r=N$) is used to establish the equivalence, then a truncated SVD is used ($r<N$) resulting in approximations of the modes of the full state GLE, as shown below. 

To establish a method for performing modal analysis of equation \eqref{eq_mz}, we observe that the dynamics described by equations \eqref{eq_mz} and \eqref{eq_mz_obs} can be understood in terms of the associated block companion matrices
\begin{equation} \label{Cx_Cg}
    \bm C_x = \begin{bmatrix}
\bm \Omega_0^{(x)} & \bm \Omega_1^{(x)} & ... & \bm \Omega_k^{(x)} \\
\bm I & \bm 0 & ... & \bm 0 \\
  & \ddots & & \\
\bm 0 & ... &  \bm I & \bm 0 \\
\end{bmatrix} \quad \text{and}\quad 
\bm C_g = \begin{bmatrix}
\bm \Omega_0^{(g)} & \bm \Omega_1^{(g)} & ... & \bm \Omega_k^{(g)} \\
\bm I & \bm 0 & ... & \bm 0 \\
  & \ddots & & \\
\bm 0 & ... &  \bm I & \bm 0 \\
\end{bmatrix},
\end{equation}
where $\bm I$ and $\bm 0$ are the appropriately sized identity and zero matrices respectively.

The solutions to equation \eqref{eq_mz} are therefore given by 
\begin{equation}\label{eq:companion_sol}
    \bm x_{n} = \bm P_0 \bm C_x^{n} \bm z_0,
\end{equation}
where $\bm z_0 = [\bm x_0, ..., \bm x_{-k}]^T$ and $\bm P_0 = [\bm I \hspace{1mm} \bm 0 \hspace{1mm} ... \hspace{1mm} \bm 0]$, and the solution for equation \eqref{eq_mz_obs}, is $\bm g_{n} = \bm P_0 \bm C_g^{n} \tilde{\bm z}_0$ where $\tilde{\bm z}_0 = [\bm g_0, ..., \bm g_{-k}]^T$. It is worth noting that adding memory kernels introduces a multiplicative interaction between the $\bm \Omega_i$'s, as seen in the long time dynamics described by $\bm C_x^{n}$. Assuming $\bm C_x$ is diagonalizable with an eigendecomposition given by $\bm C_x = \bm \Psi \bm \Lambda \bm \Psi^{-1} $, then equation \eqref{eq:companion_sol} becomes 
\begin{equation}
    \bm x_n = \bm P_0 \bm \Psi \bm \Lambda ^n \bm \Psi^{-1} \bm z_0.
\end{equation}

Defining the amplitude vector as $\bm a := \bm \Psi^{-1} \bm z_0$ (i.e. $\bm z_0$ expressed in basis of the eigenvectors of the companion matrix), the temporal evolution of $\bm x$ can be described in the basis expansion about the eigenvectors (modes) $\bm \psi_i$; 
\begin{equation}\label{eq:sol_companion_modes}
    \bm x_n = \bm P_0 \sum_{i=1}^{rk} a_i \lambda_i^n \bm \psi_i = \sum_{i=1}^{rk} a_i \lambda_i^n \bm \psi^0_i, 
\end{equation} 
where $\bm \psi^0 \equiv \bm P_0 \bm \psi \in \mathbb{R}^N$. Thus, the solutions are fully characterized by the eigenpairs $(\lambda_i, \bm \psi_i^0)$ of the associated companion matrices. Furthermore, like DMD, the eigenvectors identify the large-scale coherent structures and their associated eigenvalues determine the temporal evolution. However, now the modes and eigenvalues contain the effects of the MZ memory terms, which serve to approximate a closure model for DMD in the case when the state-space observables do not form a Koopman invariant subspace.

Next, we make use of the reduced space description and derive a relationship between the eigenpair of the full system \eqref{eq_mz} with that of \eqref{eq_mz_obs}. A similar result is shown for the SVD-based DMD \citep{schmid_2010, dmd_book}, which is recovered in what follows when the number of memory terms $k$ is set to zero. 
\begin{theorem}\label{thm:1}
Let $(\lambda, \bm w)$ be an eigenpair of $\bm C_g$. Then $(\lambda, \bm \psi)$ is and eigenpair of $\bm C_x$ where $\bm \psi = [\bm \psi^0, (1/\lambda) \bm \psi^0, ..., (1/\lambda^k) \bm \psi^0]^T$, and $\bm \psi^0 = \bm U \bm w^0$.
\end{theorem}
\begin{proof}
Let $\bm X = \bm U \bm \Sigma \bm V^*$ be the full SVD of the snapshot matrix ($r=N$). Then, expanding $\bm C_g \bm w = \lambda \bm w$, we see that 
$$ 
\begin{bmatrix}
\bm \Omega_0^{(g)} & \bm \Omega_1^{(g)} & ... & \bm \Omega_k^{(g)} \\
\bm I & \bm 0 & ... & \bm 0 \\
  & \ddots & & \\
\bm 0 & ... &  \bm I & \bm 0 \\
\end{bmatrix} \begin{bmatrix}
\bm \omega^0 \\
\bm \omega^1 \\
  \vdots \\
\bm \omega^k \\
\end{bmatrix} = \lambda \begin{bmatrix}
\bm \omega^0 \\
\bm \omega^1 \\
  \vdots \\
\bm \omega^k \\
\end{bmatrix},
$$
so that after substitution the associated nonlinear eigenvalue problem (\cite{ch9_nonlinear_eig}) of $\bm C_g$ satisfies
$$ \bm \Omega_k^{(g)} \bm w^0 + \lambda \bm \Omega_{k-1}^{(g)} \bm w^0 +...+ \lambda^{k} \bm \Omega_0^{(g)} \bm w^0 = \lambda^{k+1} \bm w^0.  $$
Substituting $\bm \Omega_i^{(g)} = \bm U^{*} \mathbf{\Omega}_i^{(x)} \bm U$ then left multiplication by $\bm U$ results in
$$ \bm \Omega_k^{(x)} (\bm U \bm w^0) + \lambda \bm \Omega_{k-1}^{(x)} (\bm U \bm w^0) +...+ \lambda^{k} \bm \Omega_0^{(x)} (\bm U \bm w^0) = \lambda^{k+1} (\bm U \bm w^0).  $$
This is the associated nonlinear eigenvalue problem of $\bm C_x$, and since $\bm \omega^{i>0} = \bm \omega^{i-1} / {\lambda^i}$,  $(\lambda, \bm \psi)$ is an eigenpair of $\bm C_x$. 
\end{proof}

This simple argument therefore directly extends what is done in DMD, providing a model analysis that contains an approximate closure model for DMD with MZ memory. In practice, like DMD, a truncated POD basis $\bm U_r$ (the left singular vectors) is computed ($r < N$) to obtain a suitable low rank approximation. In this case, the above relation is only approximate, in which $ \bm U_r \bm U_r^*$ forms an orthogonal projection operator projecting onto the first $r$ POD modes of the snapshot matrix. This fact together with Theorem \ref{thm:1} allows us to approximate the modes and spectrum of the state space GLE \eqref{eq_mz} from the eigenpair of the reduced observable GLE \eqref{eq_mz_obs}, which is otherwise intractable for large systems. Figure \ref{fig:diagram} illustrates the relationships between observables, eigenpairs, and operators in these two spaces. The MZMD modes are consequently $\tilde{\bm \Psi} \equiv [\bm \psi_1^0,..., \bm \psi_{rk}^0]$ and coefficients are then $\bm a = \bm \Psi^{\dagger} \bm z_0 \approx (\bm W)^{-1} \tilde{\bm z}_0$, where $\bm \Psi \equiv [\bm \psi_1,..., \bm \psi_{rk}]$, and $\bm W $ the is eigenvector matrix of $\bm C_g$.

This provides a scalable modal analysis technique that improves upon DMD by approaching a closure term using hysteresis effects with the Mori-Zwanzig formalism. It is useful to point out that when there is no memory present in MZMD, the formulation is equivalent to DMD. Indeed, it has been established that MZ without memory terms is equivalent to EDMD \citep{lin2021datadriven_full}. Functionally, MZMD minimizes the residual from DMD by adding MZ memory terms, and as an efficient generalization of DMD, improves the ability to capture nonlinearities.

\begin{figure} 
	\centering
	\begin{tikzpicture}
	\matrix (m) [matrix of math nodes,row sep=1.15em,column sep=10em,minimum width=3em]
	{
		\begin{array}{c}
		\text{\underline{States (GLE)}} \\  \bm x_{n+1} = \sum_{i=1}^{k+1} \bm \Omega_{i-1}^{(x)} \bm x_{n-i+1}
		\end{array} & \begin{array}{c}
		\text{\underline{Observables (GLE)}} \\  \bm g_{n+1} = \sum_{i=1}^{k+1} \bm \Omega_{i-1}^{(g)} \bm g_{n-i+1}
		\end{array}  \\
		\begin{array}{c}
		\text{Eigenpair}\\ ( \lambda_i, \bm \psi_i)
		\end{array} & \begin{array}{c}
		\text{Eigenpair}\\ ( \lambda_i, \bm w_i)
		\end{array} \\
        \begin{array}{c}
		\text{MZ Operators}\\ \bm \Omega_i^{(x)}
		\end{array} & \begin{array}{c}
		\text{MZ Operators}\\ \bm \Omega_i^{(g)}
		\end{array} \\};
	\path[-stealth]
    (m-1-1) edge node [above] {$\bm x_n \approx \bm U_r \bm g_n$} (m-1-2)
    (m-1-2) edge node [below] {$\bm g_n = \bm U_r^* \bm x_n$} (m-1-1)
    (m-2-2) edge node [above] {$\bm \psi_i^0 \approx \bm U_r \bm w_i^0$} (m-2-1)
    (m-3-1) edge node [below] {$ \bm \Omega_i^{(g)} =  \bm U_r^* \bm \Omega_i^{(x)} \bm U_r$} (m-3-2)
    (m-3-2) edge node [above] {$\bm \Omega_i^{(x)} \approx \bm U_r \bm \Omega_i^{(g)} \bm U_r^* $} (m-3-1);
	\end{tikzpicture}
	\caption{Diagram of the MZMD concept} 
    \label{fig:diagram}
\end{figure}


\section{Numerical results: hypersonic laminar–turbulent boundary layer transition}
\label{sec:results}

In this section we present our main results concerning MZMD as applied to the hypersonic laminar–turbulent boundary-layer transition. Validation tests and comparisons against HODMD and DMD are presented in Appendix \ref{sec:2d_cylinder} using a two-dimensional flow over a circular cylinder. All three approaches successfully capture the dominant shedding frequency identified in DNS \citep{colonius2008} and the dominant modes associated with this frequency, as well as the higher harmonics, which is consistent with previous studies \citep{tu_dmd, rowley_2009}. Although the modal decompositions of  DMD, HODMD, and MZMD are broadly consistent, subtle differences are observed. For the simple flow over a circular cylinder, one key finding was that the introduction of memory (either via time-delay embedding or MZ memory) enabled the recovery of modes that were missing in a highly truncated DMD, specifically those associated with the first higher harmonic of the dominant shedding frequency. A comparison of generalization errors in this case shows that MZMD and HODMD perform similarly (both improving over DMD), but MZMD achieves these results with substantially lower computational cost (avoiding HODMD’s extra SVD) and faster convergence, requiring fewer snapshots to reach equivalent accuracy.

We now turn to the primary application of the new method: the transition from laminar to turbulent flow in a high-speed boundary layer. Transitional and turbulent flows arise in many applications and exhibit a hierarchy of coherent structures essential for understanding the underlying physics \citep{lumley_book_2012, pope_2011}. Such insights are particularly valuable for the advancement of engineering design, especially in hypersonic applications involving complex geometries \citep{hader_fasel_2019, meersman_2021}. Moreover, the development of flow control strategies, whether to delay or accelerate the transition, requires reduced-order modeling and a deeper understanding of the primary mechanisms at play. 

The high-speed laminar-turbulent boundary layer transition is a complex dynamical phenomenon and remains an active research area. The transition to turbulence results in dramatic increases in skin friction (drag) and heat transfer, often far exceeding laminar values. For example, studies on the laminar-turbulent transition of a flared cone at Mach 6, including wind tunnel experiments \citep{chynoweth_2019} and DNS \citep{hader_fasel_2019}, have revealed the formation of localized hot streaks in the nonlinear breakdown, where heat fluxes can significantly exceed both the laminar and turbulent heat transfer values. Accurately predicting these hot spots is therefore critical for thermal protection systems and vehicle integrity. 

Given the hypersonic boundary-layer transition’s sensitivity to past disturbances and its highly nonlinear nature, we augment linear reduced-order models with explicit memory terms and assess whether this enhances their predictive and diagnostic capabilities. In the remainder of this section, we present a detailed analysis comparing the performance of DMD, HODMD, and MZMD applied to a hypersonic laminar-turbulent boundary layer transition. This flow is significantly more complex than the 2D cylinder case, posing distinct challenges and highlighting key differences in each method’s capability to capture relevant flow features. We assess model accuracy and generalization by evaluating each on an ensemble of initial conditions that are independently and identically distributed (\textit{i.i.d.}) test samples. Evaluating the generalization error provides insights into which spectral representation best captures the underlying flow physics. We then investigate the spatial regions within the flow where the MZMD memory terms most significantly enhance predictive capability. Finally, we analyze the dominant flow modes and quantify how MZ memory effects influence the overall dynamics.

\subsection{Description of the data-set}

We base our analysis on a high-fidelity DNS dataset of hypersonic boundary-layer transition over a Mach 6 flared cone. Transition was triggered by introducing random perturbations at the computational inflow boundary, commonly referred to as natural transition (see schematic in Figure~\ref{fig:geometry}(a) and details in \citet{hader_2018}).  The natural transition data set considered here includes the relevant transition stages defined by \citet{morkovin_1994} from the primary instability to breakdown to turbulence (see Figure \ref{fig:geometry} (b)). 

Our reduced-order modeling focuses on the statistically stationary, normalized 2D pressure fluctuations at the cone surface, derived from the full 3D DNS performed by \citet{hader_fasel_2019}. The fluctuating normalized pressure field $p'/(\rho_{\infty} U_{\infty})$, where $p' = p - \langle p \rangle$, is considered at the wall with a spatial resolution of $129 \times 4600$. This reduced dataset remains representative of the critical dynamics responsible for the nonlinear generation of hot streaks shown in Figure~\ref{fig:geometry}(c). For data-driven modeling, a subset of DNS data, comprising 3,000 uniformly sampled time snapshots, is used for fitting DMD, HODMD, and MZMD. This is the number of snapshots required to achieve convergence of the operators (see Figure \ref{fig:convergence_ops} (a)). The DNS data are uniformly sampled at a frequency of $30 \Delta t$, where the DNS time step $\Delta t \approx 3.33 \times 10^{-9} s$ satisfies the CFL condition (see \citep{hader_fasel_2019} for more details). This sampling rate can resolve frequencies up to approximately $2600KHz$, which is more than adequate to capture the first several higher harmonics of the fundamental instability wave (as seen in Figures \ref{fig:convergence_ops} (b) and \ref{fig:evals_spectrum}). Additionally, this temporal window corresponds to a timescale in which the flow advects approximately $\mathcal{O}(5L)$ where $L = 0.36 m$ is the length of the computational domain (where $0.51 m$ is the length of the cone). A held-out test set of 1200 snapshots (where flow advects $\mathcal{O}(2L)$) that is identically distributed with the training set are used for performing analysis of each  model. The test set is used in measuring short time generalization errors over an ensemble of initial conditions sampled \textit{i.i.d}.

\begin{figure}[!htb]
\centering
\begin{subfigure}[]{1\textwidth}
\centering
\includegraphics[width=1\textwidth]{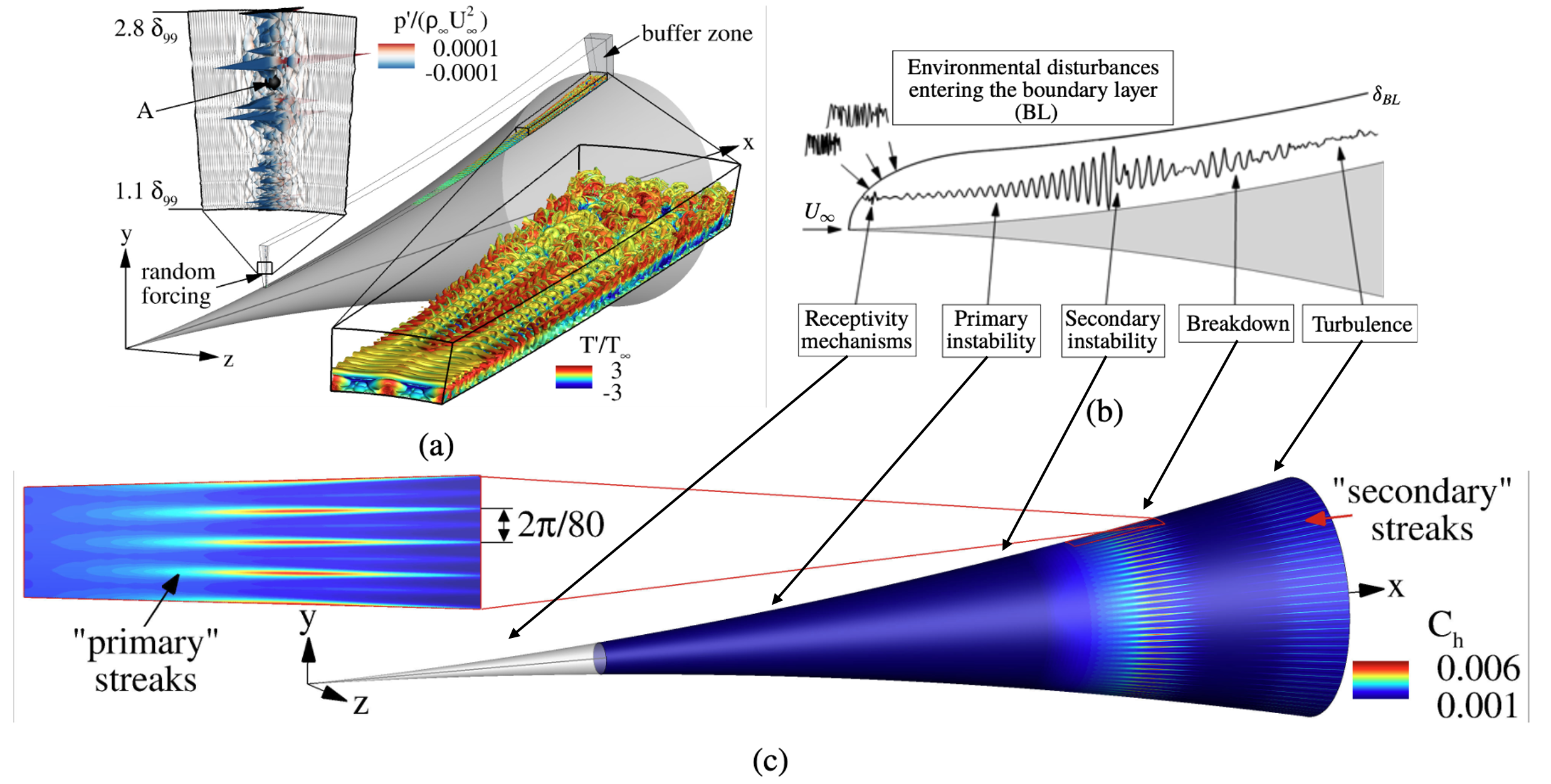}
\end{subfigure}
\caption{(a) Computational setup for the natural transition DNS dataset\citep{hader_2018},  (b) schematic of the transition stages \citet{morkovin_1994}, and (c) time-averaged Stanton number contours on the surface of the cone obtained from DNS.}
\label{fig:geometry}
\end{figure}

\begin{figure}[!htb]
\centering
\begin{subfigure}[]{0.5\textwidth}
\centering
\includegraphics[width=1\textwidth]{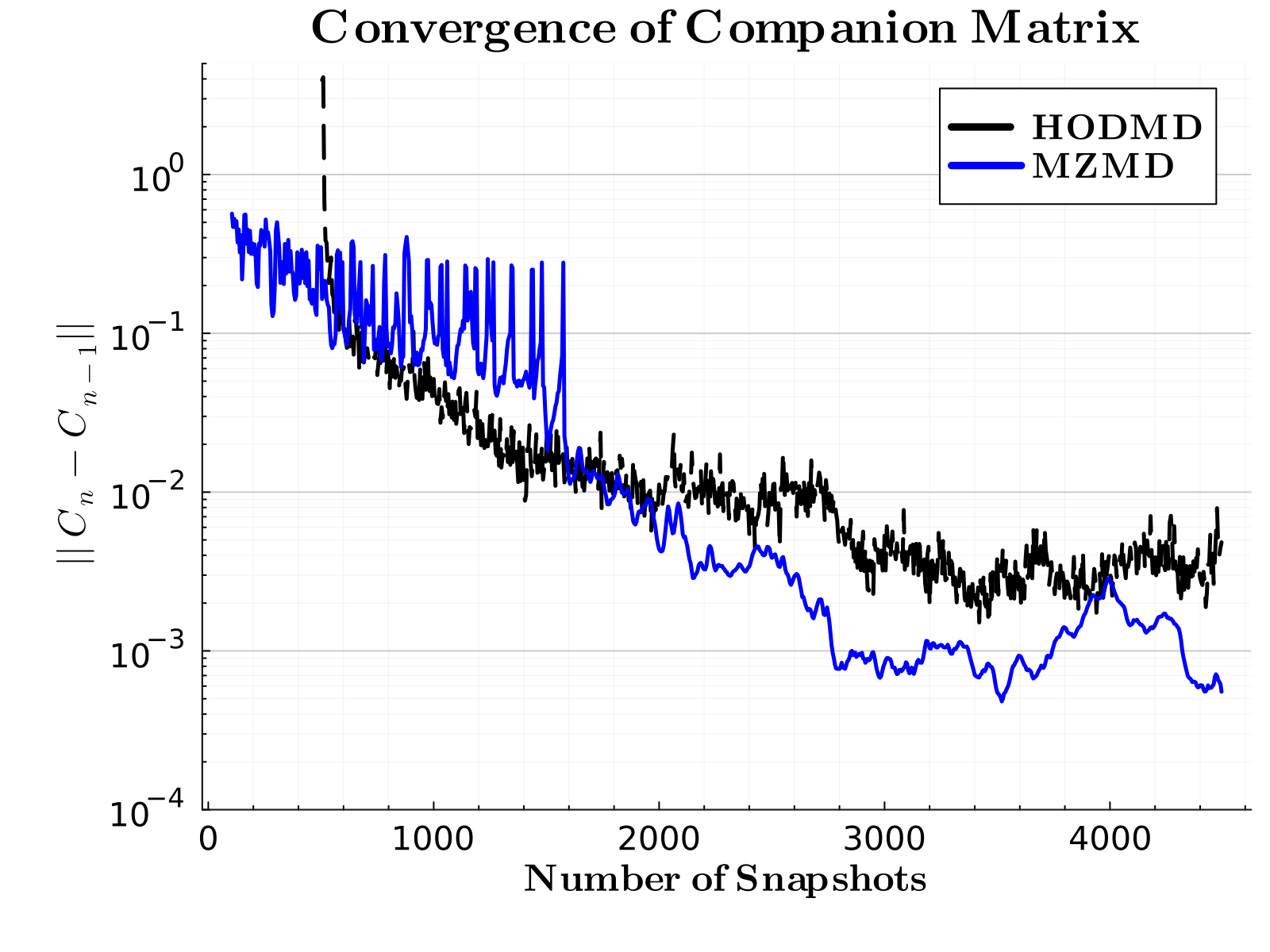}
\caption{}
\label{fig:conv_operators}
\end{subfigure}
\centering
\begin{subfigure}[]{0.44\textwidth}
\centering
\includegraphics[width=1\textwidth]{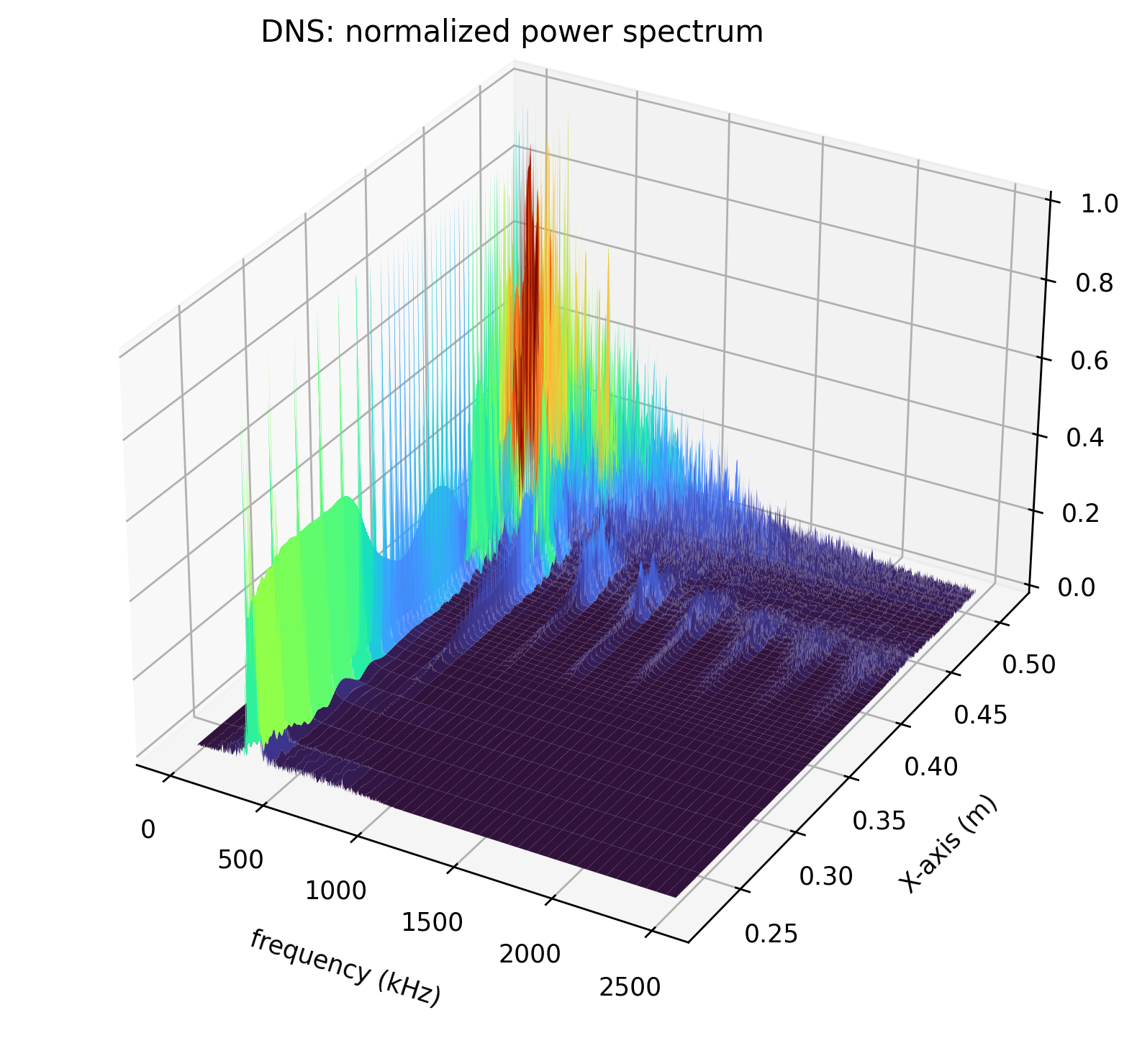}
\caption{}
\label{fig:pow_spec}
\end{subfigure}
\caption{(\subref{fig:conv_operators}) convergence of HODMD and MZMD companion matrices with respect to the amount of training data used for the hypersonic dataset, and (\subref{fig:pow_spec}) the power spectrum of DNS as a function of the downstream direction $x$.} 
\label{fig:convergence_ops}
\end{figure}

\subsection{Analysis of future state predictions}
\label{subsec:analysis_preds}

In order to determine which modal analysis technique obtains the most accurate description of the dynamics, we perform a detailed comparison of the predictive performance of each model. With MZMD, like HODMD and DMD, there is a choice to be made with the selected level of low-rank approximation $r$ performed. Figure \ref{fig:r_k_selection} (a) illustrates that the fluctuating pressure data ($p'(t,x)$) exhibits significant low-rank structure. Specifically, low-rank approximations that reconstruct 95\% and 98\%  of the original data correspond to r values of $r=100$ and $r=270$, respectively. More precisely, $r$ is chosen so that $||X - \bm U_r \bm \Sigma_r \bm V_r||_F^2 / ||X||_F^2 <\epsilon_0$, where $\epsilon_0 = \{0.05, 0.02\}$ and $|| \cdot ||_F$ is the Frobenius norm. In the remainder of the analysis, we select these two (somewhat arbitrary) thresholds for $r$. 

Another key parameter for both HODMD and MZMD is the choice of the number of time-delay embeddings ($d$) and Mori–Zwanzig (MZ) memory terms ($k$). Optimal values for these parameters can be selected based on the minimal ensemble-averaged $L_2$ prediction error. This ensemble consists of trajectories evolved from 20 independently and identically distributed (\textit{i.i.d.}) initial conditions sampled from the test dataset. Figures~\ref{fig:errors_mzmd_hodmd_l2_linf}(a,b) indicate that the optimal memory length satisfying this criterion is $k = 14$ for MZMD and $d = 4$ for HODMD. The predictions are evaluated over a duration of 300 time steps, which corresponds approximately to the advection of fluid structures by half the domain length ($\mathcal{O}(1/2L)$), equivalent to roughly 100 times the Kolmogorov timescale. For MZMD, as shown in Figure \ref{fig:r_k_selection} (b), the number of memory terms could alternatively be determined using a relative convergence criterion based on the decay of memory terms with respect to the Frobenius norm. This decay of memory contributions is expected for complex physical systems like the one under consideration, as the two-time correlation decays over time. In contrast, HODMD lacks this type of memory decay. 

Figure \ref{fig:errors_mzmd_hodmd_l2_linf} shows that HODMD performs better at reconstruction than MZMD; however, for future state prediction, MZMD outperforms HODMD. Thus, HODMD suffers from overfitting, whereas adding MZ memory in MZMD improves the generalization errors for short time predictions. Results presented are for $r=100$; at $r=270$, in which we see that HODMD becomes unstable for several choices of embedding dimension $d$ (see Figure \ref{fig:errors_oper_t}(b)). Predictive performance, rather than reconstruction, better represents model accuracy, reinforcing that HODMD's improved reconstruction is due to overfitting (a known limitation of time-delay embeddings \citep{cyoung_23, dhadphale_model_2024}). Conversely, MZMD shows consistent improvement in both reconstruction and predictive accuracy over standard DMD as the number of memory terms increases. Interestingly, this overfitting behavior observed in HODMD does not appear in the simpler 2D cylinder flow case (see Appendix, Figure \ref{sec:2d_cylinder}). For the prediction horizon considered, MZMD with $k=14$ reduces relative prediction error by approximately $3\%$ compared to standard DMD. Finally, Figure \ref{fig:errors_oper_t} shows error growth over time, highlighting MZMD's slowest error growth rate and HODMD’s instability at $r=270$.

\begin{figure}[!htb]
\centering
\begin{subfigure}[]{0.48\textwidth}
\centering
\includegraphics[width=1\textwidth]{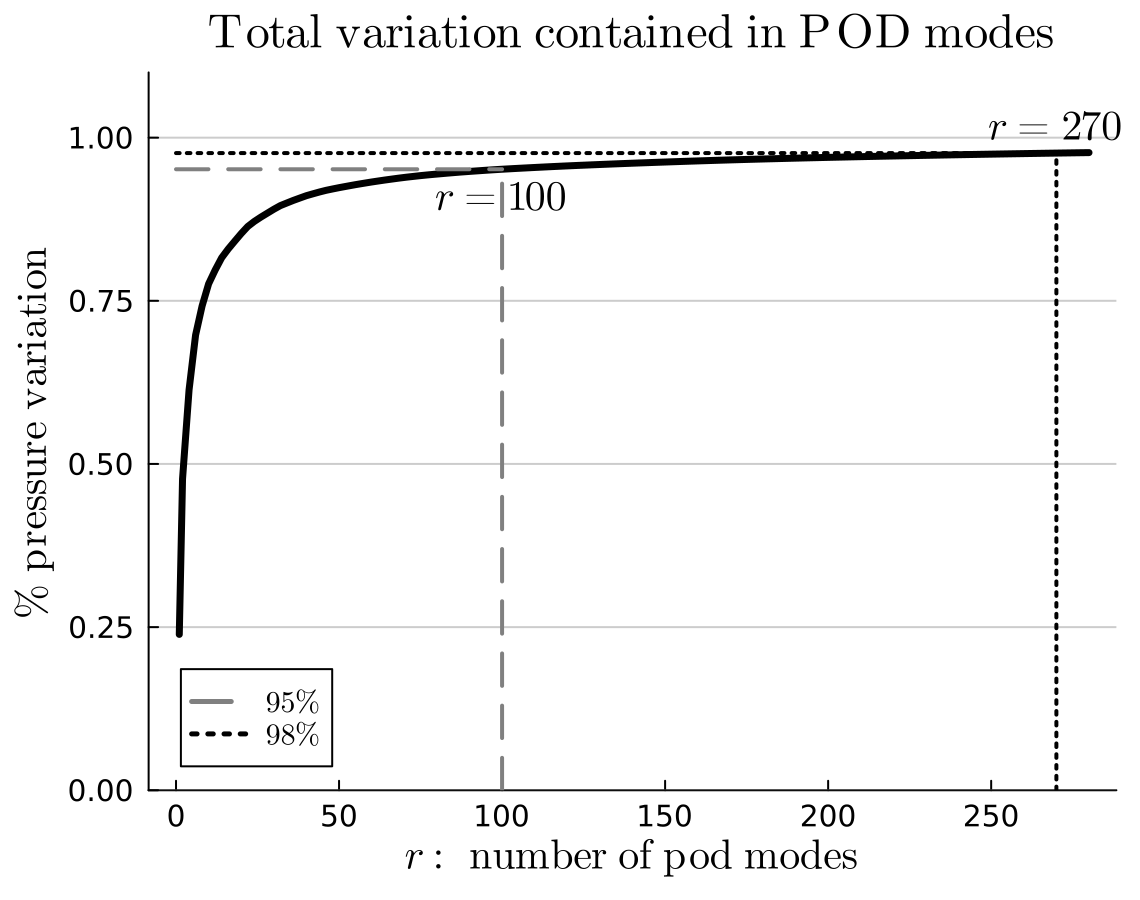}
\caption{}
\label{fig:pod_ener}
\end{subfigure}
\centering
\begin{subfigure}[]{0.48\textwidth}
\centering
\includegraphics[width=1\textwidth]{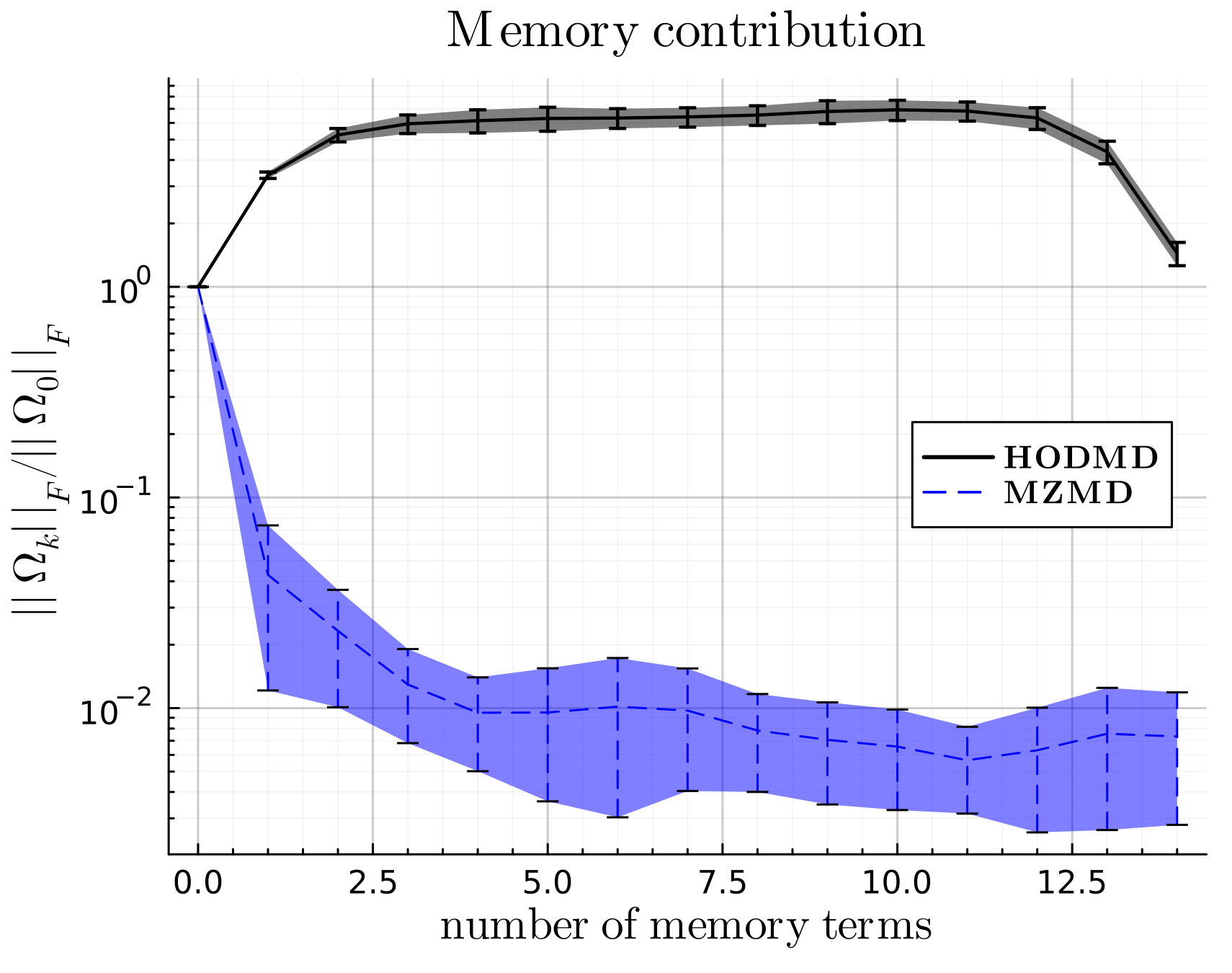}
\caption{}
\label{fig:mem_cont}
\end{subfigure}
\caption{(\subref{fig:pod_ener}) Low rank structure present in the pressure field. This figure presents the total variance contained in the POD modes of the pressure field and thresholds used for obtaining low rank approximations; $\%95$ and $\%98$ variation of the pressure fields are contained within the first $r=100$ and $r=270$ modes, respectively. (\subref{fig:mem_cont}) Relative memory contributions normalized with respect to the Markovian term. This figure demonstrates decaying memory contribution with MZMD and relatively similar  contributions of all memory terms for HODMD.}
\label{fig:r_k_selection}
\end{figure}

\begin{figure}[!htb]
\centering
\begin{subfigure}[]{0.48\textwidth}
\includegraphics[width=1\textwidth]{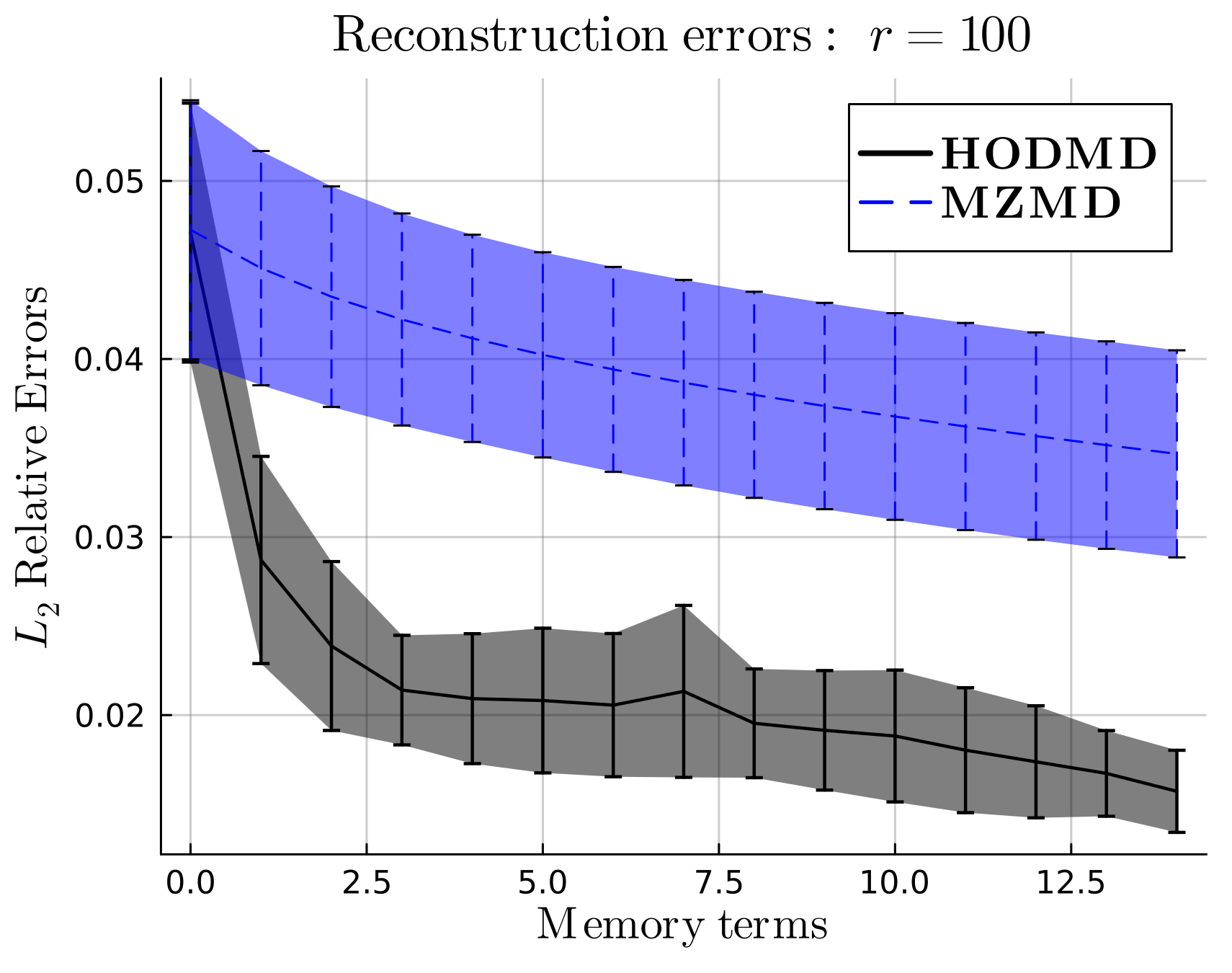}
\caption{}
\label{fig:rec_err}
\end{subfigure}
\centering
\begin{subfigure}[]{0.48\textwidth}
\includegraphics[width=1\textwidth]{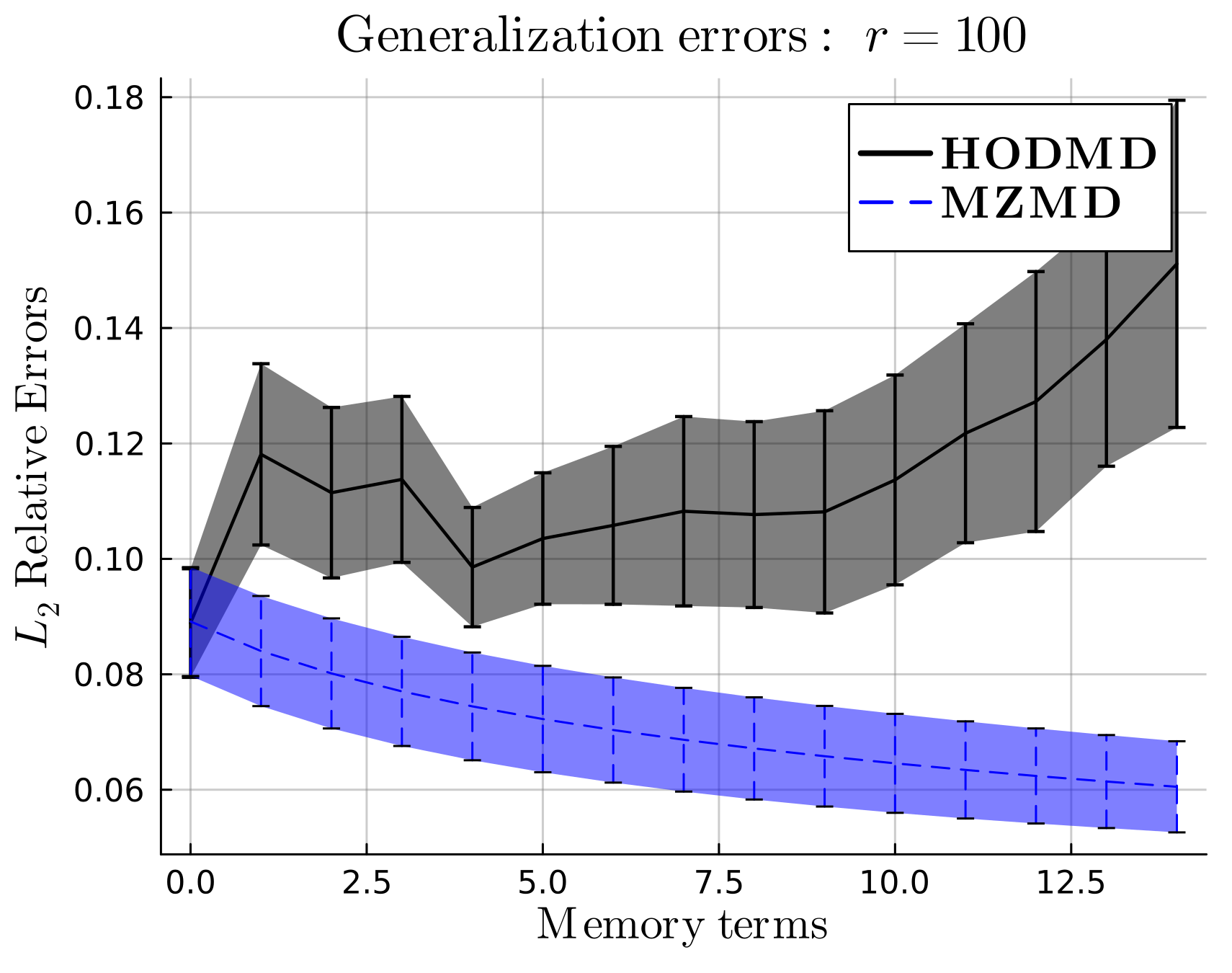}
\caption{}
\label{fig:gen_err}
\end{subfigure}
\centering
\caption{(\subref{fig:rec_err}) Relative $L_{2}$ reconstruction errors (over an ensemble of initial conditions uniformly sampled on the training data) and (\subref{fig:gen_err}) future state prediction (or generalization) over $300$ time steps, displaying error as a function of the number of memory terms (over an ensemble of initial conditions uniformly sampled on the test set).} 
\label{fig:errors_mzmd_hodmd_l2_linf}
\end{figure}

\begin{figure}[!htb]
\centering
\begin{subfigure}[]{0.48\textwidth}
\includegraphics[width=1\textwidth]{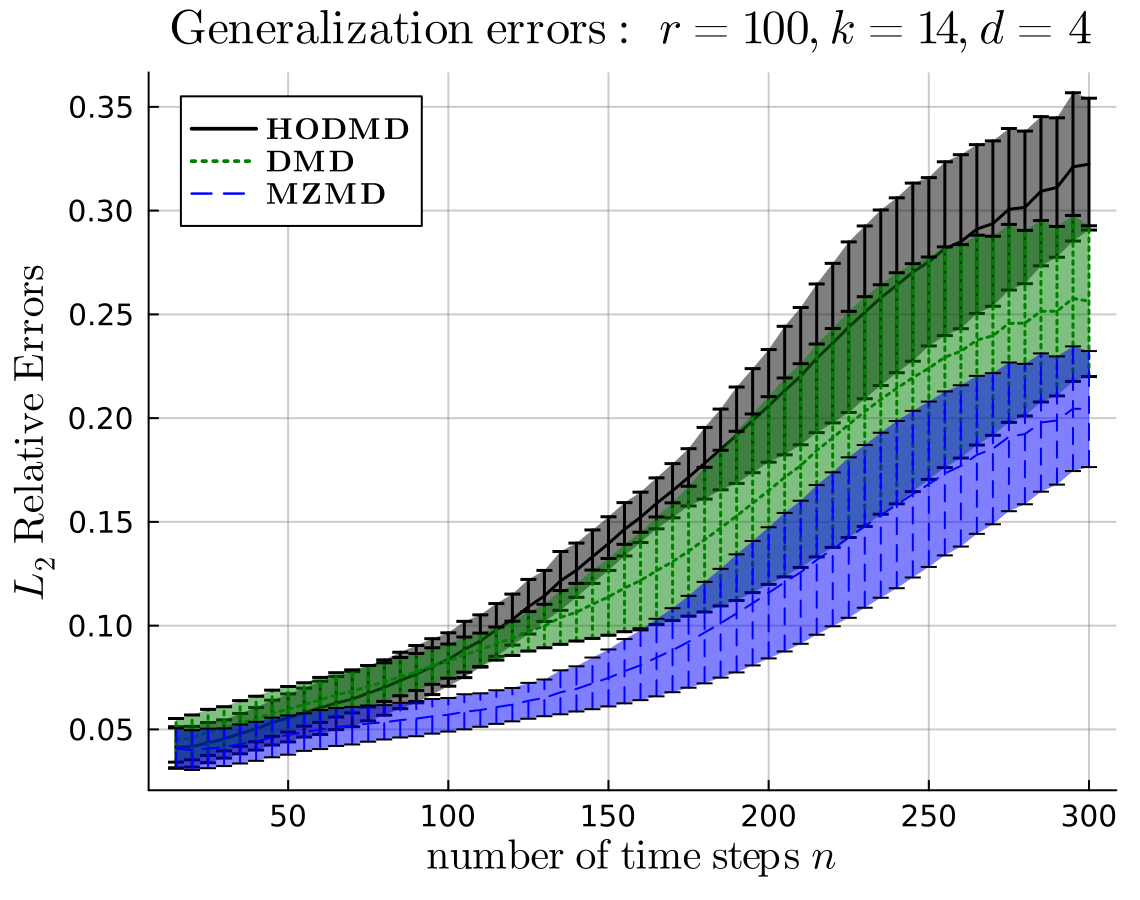}
\caption{}
\label{fig:gen_err_t_r100}
\end{subfigure}
\centering
\begin{subfigure}[]{0.48\textwidth}
\includegraphics[width=1\textwidth]{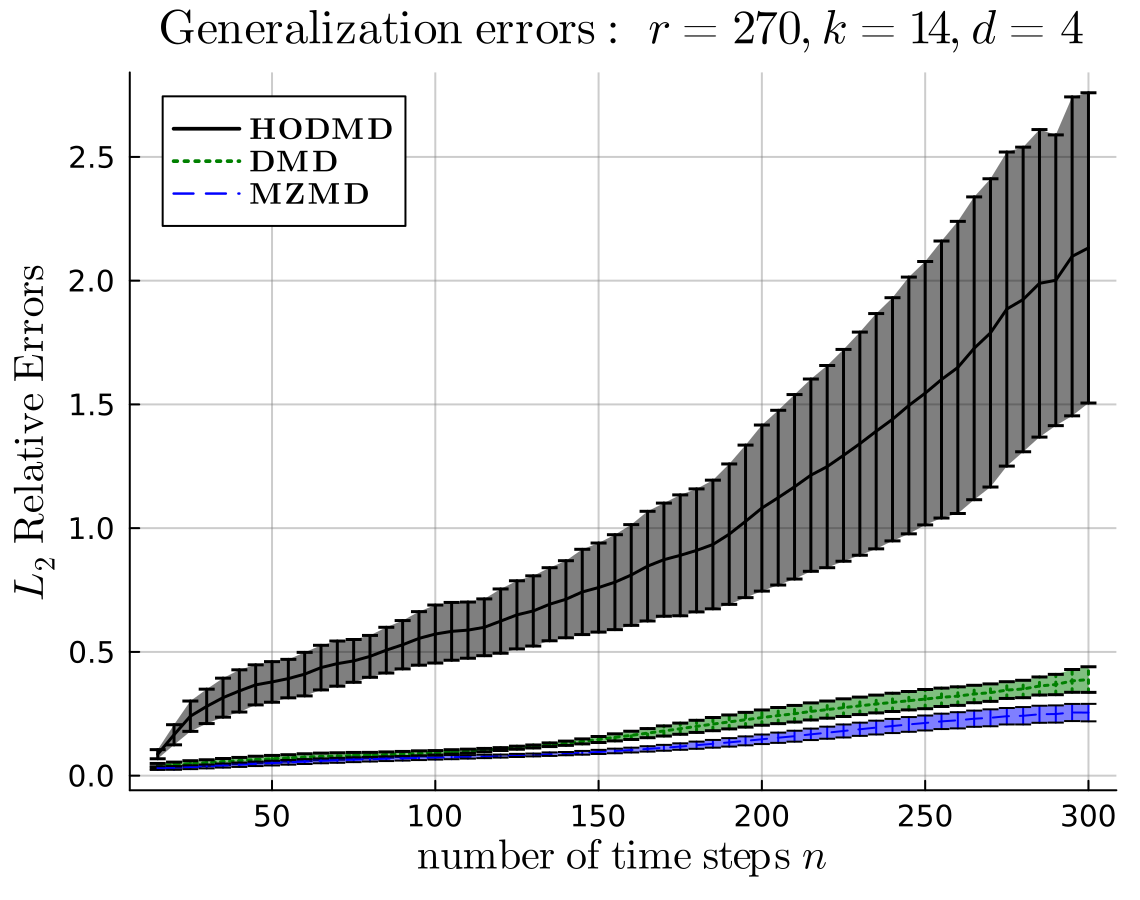}
\caption{}
\label{fig:gen_err_t_r270}
\end{subfigure}
\centering
\caption{Comparing future state prediction errors over time of DMD, HODMD, and MZMD. (\subref{fig:gen_err_t_r100}) $r = 100$ (which contains $\%95$ of total pressure variation), (\subref{fig:gen_err_t_r270}) $r=270$. For both $r$ values, MZMD uses $k=14$, while HODMD uses $d=4$ memory terms. We include the case of $r=270$ to show that in some cases HODMD is found to blow up (it introduces unstable modes, i.e. eigenvalues that lie outside the unit circle), but in the remainder of the analysis we perform all comparisons with $r=100$.} 
\label{fig:errors_oper_t}
\end{figure}

So far we have investigated the aggregated statistics and integrated metrics to compare the prediction accuracies for each method. Next, we examine where in the flow each model is more or less accurate by measuring the point-wise Mean Squared Error (MSE) time-averaged over the future state predictions. This is reported in Figure \ref{fig:pointwise_mse_models}, which shows that the largest errors occur immediately after the onset of the strongly nonlinear portion of the transition region. This is the region of the flow where the hot streaks are generated. Compared to HODMD and DMD, MZMD provides a significant improvement to the prediction of the hot streaks. 

\begin{figure}[!htb]
\centering
\begin{subfigure}[]{1\textwidth}
\includegraphics[width=1\textwidth]{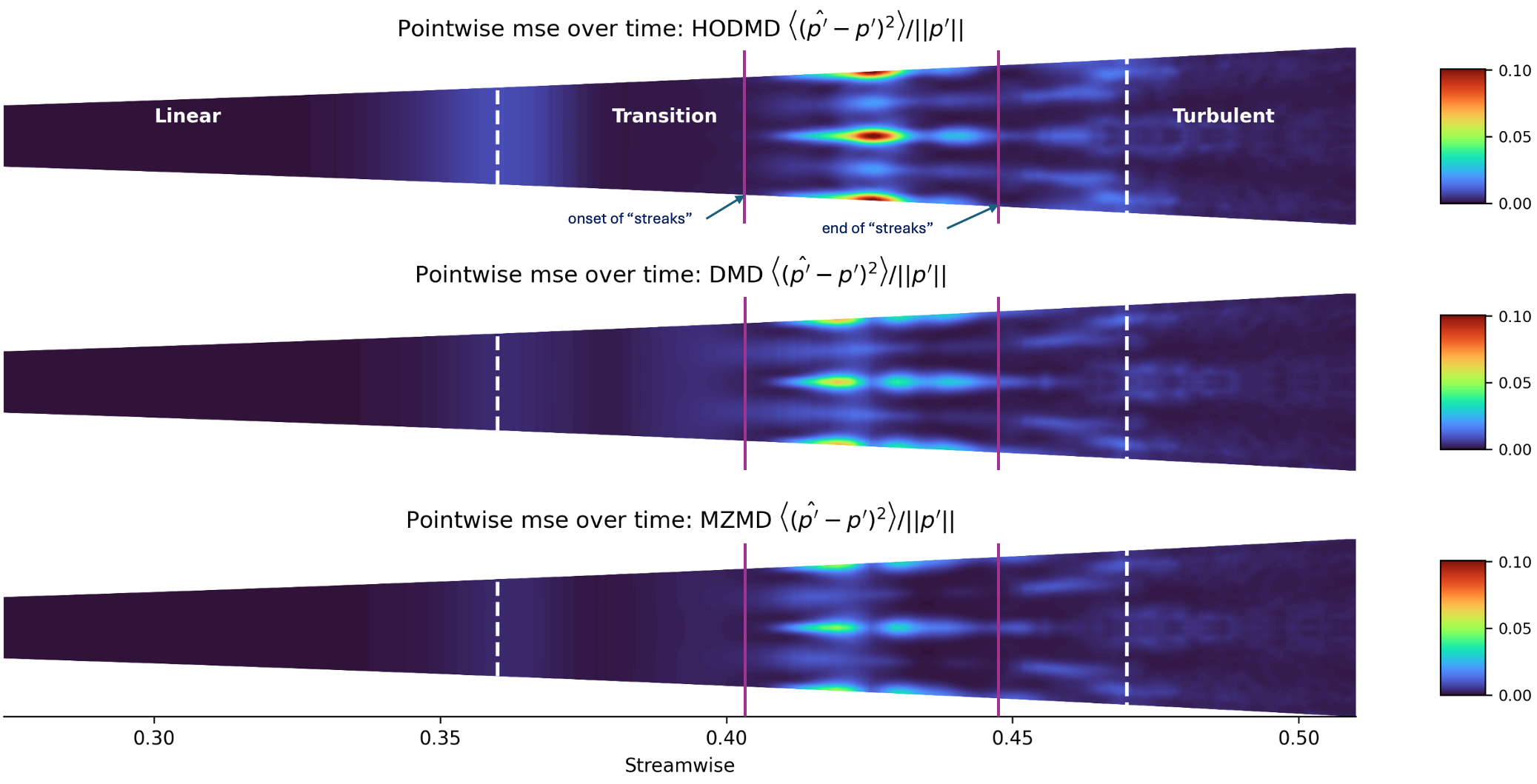}
\end{subfigure}
\caption{Pointwise, time-averaged mean squared error (MSE) for predictions on the test dataset, visualized by mapping the conical geometry onto a two-dimensional plane for clarity. Notably, MZMD demonstrates a significant improvement in forecasting future flow states, achieving substantially lower prediction errors compared to alternative methods, particularly within the highly nonlinear flow regime associated with hot-streak formation. The primary instability exhibits exponential growth within the linear region, subsequently saturating and triggering secondary instabilities in the transitional region. Ultimately, this cascade culminates in a breakdown into turbulence within the turbulent region.} 
\label{fig:pointwise_mse_models}
\end{figure}

Next, we seek to understand what information the MZMD memory terms contain and specifically, what regions of the flow memory terms contribute the most. For the prediction horizon considered here, Figure \ref{fig:mzmd_mem_improvement} shows that each memory term contributes primarily to the transition region, specifically concentrated on the hot streak structures. In the breakdown region (after streamwise position $\sim 0.42$) each successive memory term provides a monotonically decreasing contribution (as expected from memory decay), with the corresponding flow structures becoming increasingly diffuse around the hot streaks. Right before the full breakdown over the first part of the hot streaks (streamwise position $\sim 0.42$), the memory contribution increases first until the fourth memory term, after which the contributions of higher-order terms monotonically decrease. This indicates longer memory in this region, where the full breakdown has not started yet, so that higher-order MZ memory terms can account for the longer time the dominant and secondary instabilities take to develop. Afterwards, the memory effects shorten. MZ memory is therefore driving the improvement of the prediction accuracy by concentrating its impact immediately after the onset of the strongly nonlinear transition regions of the flow.

\begin{figure}[!htb]
\centering
\begin{subfigure}[]{1\textwidth}
\centering
\includegraphics[width=1\textwidth]{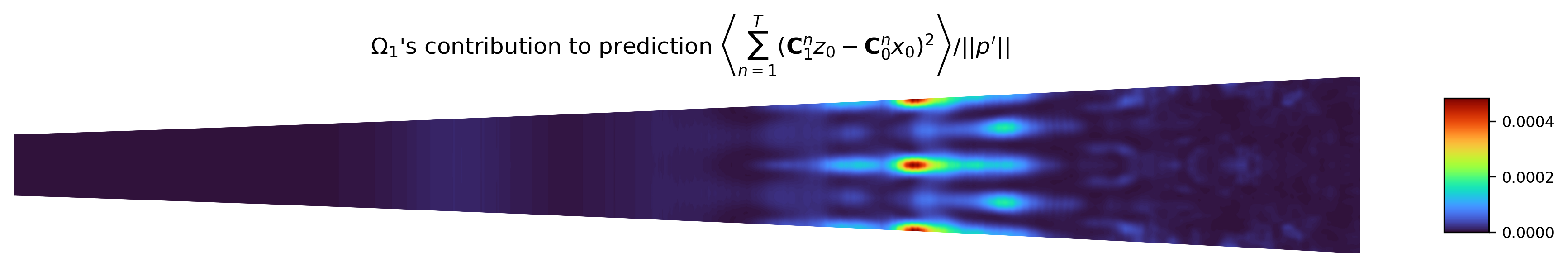}
\end{subfigure}
\begin{subfigure}[]{1\textwidth}
\centering
\includegraphics[width=1\textwidth]{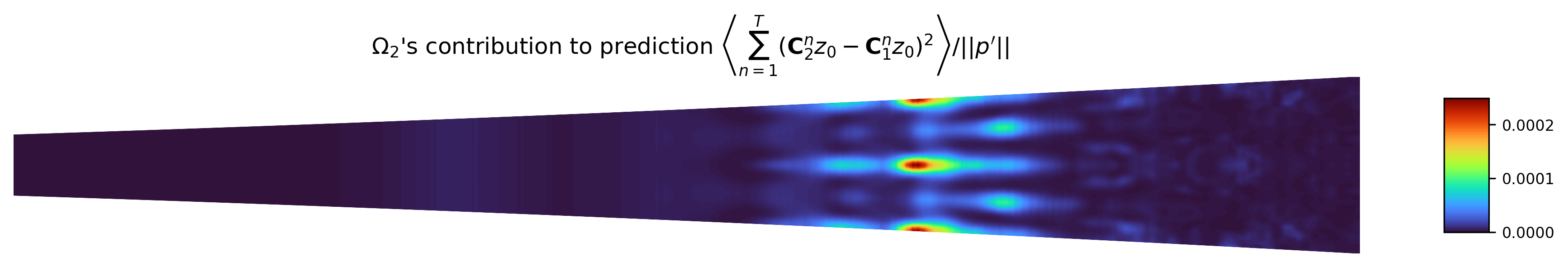}
\end{subfigure}
\begin{subfigure}[]{1\textwidth}
\centering
\includegraphics[width=1\textwidth]{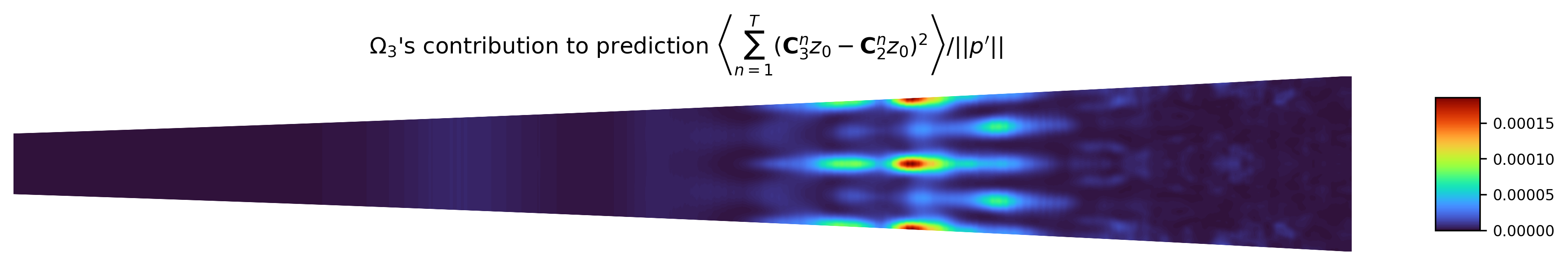}
\end{subfigure}
\begin{subfigure}[]{1\textwidth}
\centering
\includegraphics[width=1\textwidth]{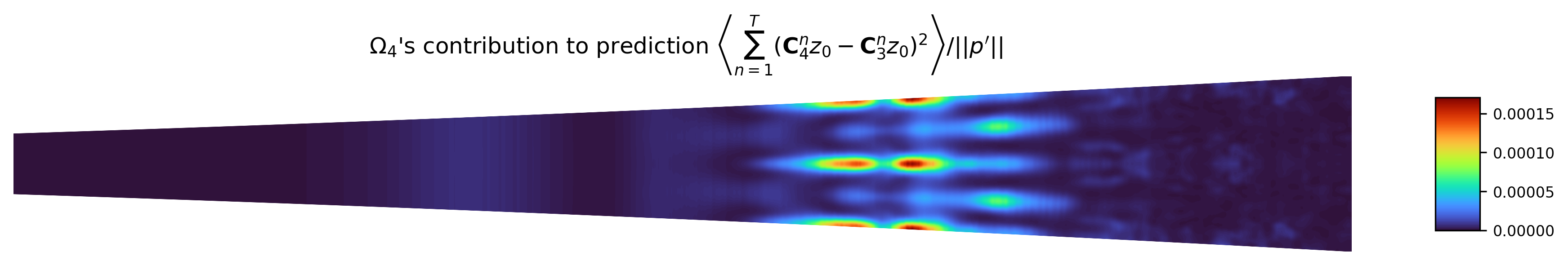}
\end{subfigure}
\begin{subfigure}[]{1\textwidth}
\centering
\includegraphics[width=1\textwidth]{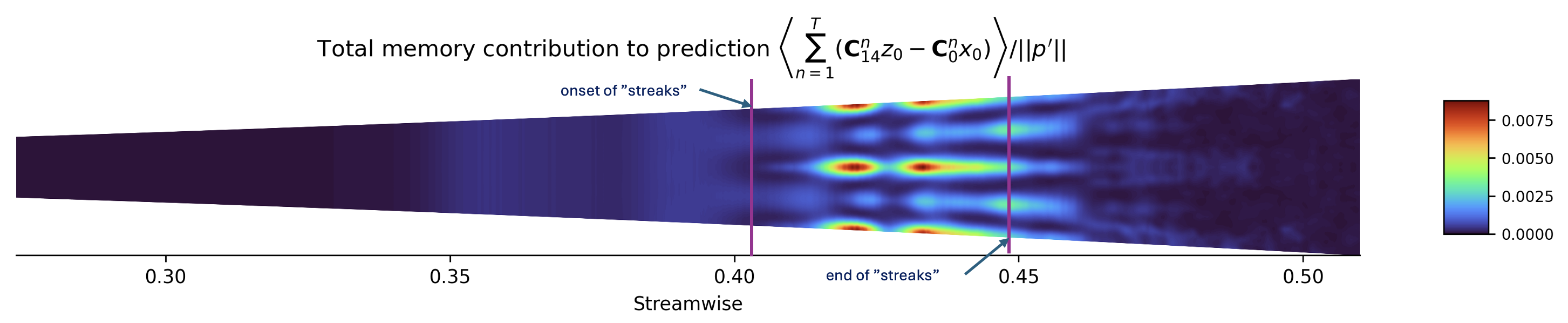}
\end{subfigure}
\caption{The contribution of each MZMD memory term to the future state prediction (top to bottom). $\bm C_k$ is defined as the companion matrix containing $k$ memory terms. Each memory term primarily contributes to the transition region, mainly improving the resolution of the hot streaks. }
\label{fig:mzmd_mem_improvement}
\end{figure}

\subsection{Modal analysis}
\label{sec:modal_analysis_results}

The above analysis is important when considering which technique produces the most accurate modal decomposition (or spectral representation), and, therefore, captures the large-scale coherent structures most accurately. As demonstrated above, MZMD provides the most accurate model for future state predictions of the Mach-6 boundary layer. In the following, we perform the modal and spectral analysis of each method in order to draw further numerical distinctions. Then, we present an analysis of the MZMD modes to investigate the physical mechanisms that generate the hot streaks.

The eigenvalues and amplitudes of the DMD, HODMD, and MZMD modes are compared in Figure \ref{fig:evals_spectrum}. The eigenvalues determine the temporal evolution of the corresponding modes (as shown in equation \ref{eq:sol_companion_modes}); eigenvalues located inside, on, and outside the unit circle correspond to decaying, periodic, and unstable modes, respectively. The largest amplitude mode is associated with the fundamental (or primary) instability wave predicted by Linear Stability Theory (LST), where the dominant mode is amplified and grows exponentially in the downstream direction \citep{hader_fasel_2019, chynoweth_2019}. However, each method considered here produces a slightly different result for this frequency. This dominant frequency is estimated from DNS to be $f_{DNS} \approx 294.38 (kHz)$ on the training set by computing the frequency associated with the largest amplitude of the power spectrum in the linear region (see also Figure \ref{fig:convergence_ops}). As seen in Figure \ref{fig:evals_spectrum}, taking the frequency associated with the largest amplitude, we find that DMD estimates $f_{DMD} = 293.18 (kHz)$,  HODMD $f_{HODMD} = 286.09 (kHz)$, and MZMD $f_{MZMD} = 294.80 (kHz)$. Therefore, of the three methods, MZMD produces the most accurate temporal behavior of the fundamental mode. This also helps to explain why MZMD improves future state predictions. Accurate capture of the dominant mode is also important for predicting the higher harmonics that are integer multiples of $f_{primary}$. These higher harmonics play a crucial role in the laminar-turbulent transition mechanisms and the development of the hot streaks. We see that both HODMD and MZMD increase the spectral complexity over DMD; however, MZMD introduces a broader spectrum as opposed to HODMD, which contains a large spectral gap.

Figure \ref{fig:evals_spectrum} illustrates the influence of memory terms on the overall spectral structure. Both time-delays and MZ memory act to increase the spectral complexity (compare Figures \ref{fig:evals_spectrum} (b) and (c) with \ref{fig:evals_spectrum} (a)), as discussed in Section \autoref{sec:mz_background}. For HODMD and MZMD, the spectrum corresponding to the the leading $r$ modes ranked by amplitude, closely resembles that obtained via standard DMD, with only slight differences: MZMD slightly amplifies the first higher harmonic and and aligns its dominant frequencies more closely with DNS. While MZMD produces a quasi-uniform spectral filling, the new modes generated by HODMD are primarily clustered near the fundamental frequencies and their higher harmonics, presenting a larger spectral gap with respect to the highest frequency components in the data. 

Both MZMD and HODMD produce an identical number of modes with eigenvalues on the unit circle as DMD (Figure \ref{fig:evals_spectrum} (d), (e), (f)); yet HODMD and MZMD supply additional modes with eigenvalues strictly inside the unit circle, representing transient, decaying dynamics. Nevertheless, the modes added by HODMD tend to decay more slowly compared to the transient memory modes produced by MZMD, as their modulus are generally larger than those for the modes produced by MZMD.

MZ memory also improves stability. In Figure \ref{fig:evals_spectrum_stability}, incorporating MZ memory effectively reduces the modulus of the eigenvalues that extend beyond the unit circle relative to those computed by DMD, pulling them closer to the unit circle. Although none of the three data-driven methods formally guarantees stability, MZ memory acts to perturb the eigenvalues associated with the periodic DMD modes in a stabilizing direction. HODMD exerts a similar, but weaker, stabilizing influence.

The spatial structures of the dominant MZMD modes are shown in Figure \ref{fig:phi_mz}.  The fundamental (primary) instability mode $\phi_1$ appears in the early and mid-transition region; its first and second harmonics ($\phi_2$, $\phi_3$ resp.) are activated further downstream and are localized progressively deeper into the core transition zone, with structures resembling the hot streaks. We also identify a dominant memory mode $\phi_{mem}$, which is nearly periodic and thus contributes most to long-time dynamics compared to the rest of the MZ memory modes. The leading memory mode, $\phi_{mem}$, peaks where the streaks originate but also maintains significant amplitude deep into the late-transition and turbulent sectors. Such a distribution can be interpreted in the context of the MZ formalism: the memory term encodes the influence of the unresolved variables on the reduced variables, which are most energetic in the nonlinear transition and turbulence layers. It is challenging to go much further beyond this general interpretation of the memory mode, but another observation can be made: $\phi_{mem}$ is an echo of the fundamental mode that encapsulates some approximate feedback of unresolved turbulent motions, modulating the streak envelope and correcting the phase and amplitude information missed by DMD. 

Next, we perform a quantitative comparison of the dominant modes obtained from MZMD (Figure \ref{fig:phi_mz}), DMD (Figure \ref{fig:phi_dmd}), and HODMD (Figure \ref{fig:phi_hodmd}) by using an inner product-based similarity measure defined as:
\begin{equation}
 s_{ij} :=  \cfrac{|\phi_{i} \cdot \phi_{j}|^2}{||\phi_{i}||_2^2 ||\phi_{j}||_2^2}.  
 \end{equation}
For the dominant mode and the first two harmonics, these similarities are found to be $s_{11} = 0.98$, $s_{22} = 0.92$, $s_{33} = 0.17$ between MZMD and DMD. Thus, MZMD has the largest impact on the higher harmonics. The similarities between MZMD and HODMD modes are $s_{11} = 0.98$, $s_{22} = 0.83$, $s_{33} = 0.64$, while those between HODMD and DMD are $s_{11} = 0.76$, $s_{22} = 0.20$, $s_{33} = 0.55$. This demonstrates that each form of memory not only increases the spectral complexity over DMD, but also acts to perturb the spatial structures of the DMD modes. In combination with the more accurate frequency representation of MZMD, these differences further elucidate the improvement seen in future state predictions.

Having established that MZMD provides the most faithful spectral representation (Section \ref{subsec:analysis_preds}), we now turn to the physical insights that can be gained by its modal decomposition. Returning to Figure \ref{fig:phi_mz}, the location where the hot streaks (see the discussion in \citep{hader_2018} for more details), appear and disappear in the time-averaged Stanton number contours on the surface of the cone are marked with solid black lines. The leading (fundamental) MZMD mode grows axisymmetrically through the early and mid-transition zones with a frequency of $f\approx 300$ kHz, matching linear-stability predictions for the dominant axisymmetric second-mode acoustic wave (or fundamental mode) observed on a flared cone at Mach 6 \citep{hader_2018}. The fundamental mode of MZMD evolves from an initial axisymmetric growth in the early transition stage to a noticeable azimuthal modulation as primary streaks emerge, capturing the initial nonlinear saturation processes.  The azimuthal modulation starts near the onset of the hot-streaks, matching the spacing observed in Stanton number contours (Figure \ref{fig:geometry} (c)) \citep{hader_2018}. This azimuthal modulation indicates that MZMD captures an imprint of the initial stages of the nonlinear mechanism responsible for generating the hot steady streaks via an oblique mode breakdown \citep{hader_2018}.

Further downstream, the first higher harmonic at $f \approx 600$ kHz emerges once the fundamental frequency reaches a sufficiently large amplitude. Initially, this harmonic also exhibits a predominantly axisymmetric signature but then experiences similar azimuthal modulations near the region where the hot streaks begin to form. The activation of this higher harmonic farther downstream compared to the fundamental frequency is consistent with the understanding that this higher harmonic is nonlinearly generated by a self-interaction of the primary frequency once sufficiently large amplitudes are reached \citep{kimmel_1991}. The azimuthal wavelength of this modulation corresponds to the secondary instability wave known to resonate most strongly in hypersonic flows \citep{hader_2018}. The quadratic nonlinearity of the Navier-Stokes equations gives rise to the generation of higher harmonics. Because these quadratic nonlinearities lead to an effective doubling in the frequency of the fundamental mode, the spatial manifestation of the first higher harmonic naturally exhibits a spatial wavelength that is approximately half that of the fundamental mode. This frequency-doubling mechanism leads to the emergence of finer-scale structures downstream, which is confirmed in the MZMD modes seen in Figure \ref{fig:phi_mz}, where the wavelength of the first higher harmonic is approximately half that of the fundamental mode.

The second higher harmonic at the frequency of $f \approx 900$ kHz, which emerges from the interaction between the fundamental mode and the first higher harmonic, is observed even further downstream as a result of continued nonlinear energy transfers and modal interactions in the evolving boundary layer. Additionally, this mode is now almost entirely concentrated in the region corresponding to the structure of the hot streak, and is significantly different then the corresponding DMD and HODMD second higher harmonic. Additionally, this mode’s delayed onset reflects the fact that it's developed later in the nonlinear stage of breakdown; caused by the interaction of the primary and first harmonic modes interacting further into the energy cascade process. These observations substantiate the capability of MZMD to capture not only the primary instability but also the subsequent higher harmonics and their interactions; key factors in understanding the nonlinear stages of transition to turbulence for this Mach-6 boundary layer flow.


\begin{figure}
\centering
\begin{subfigure}[]{0.3\textwidth}
\centering
\includegraphics[width=1\textwidth]{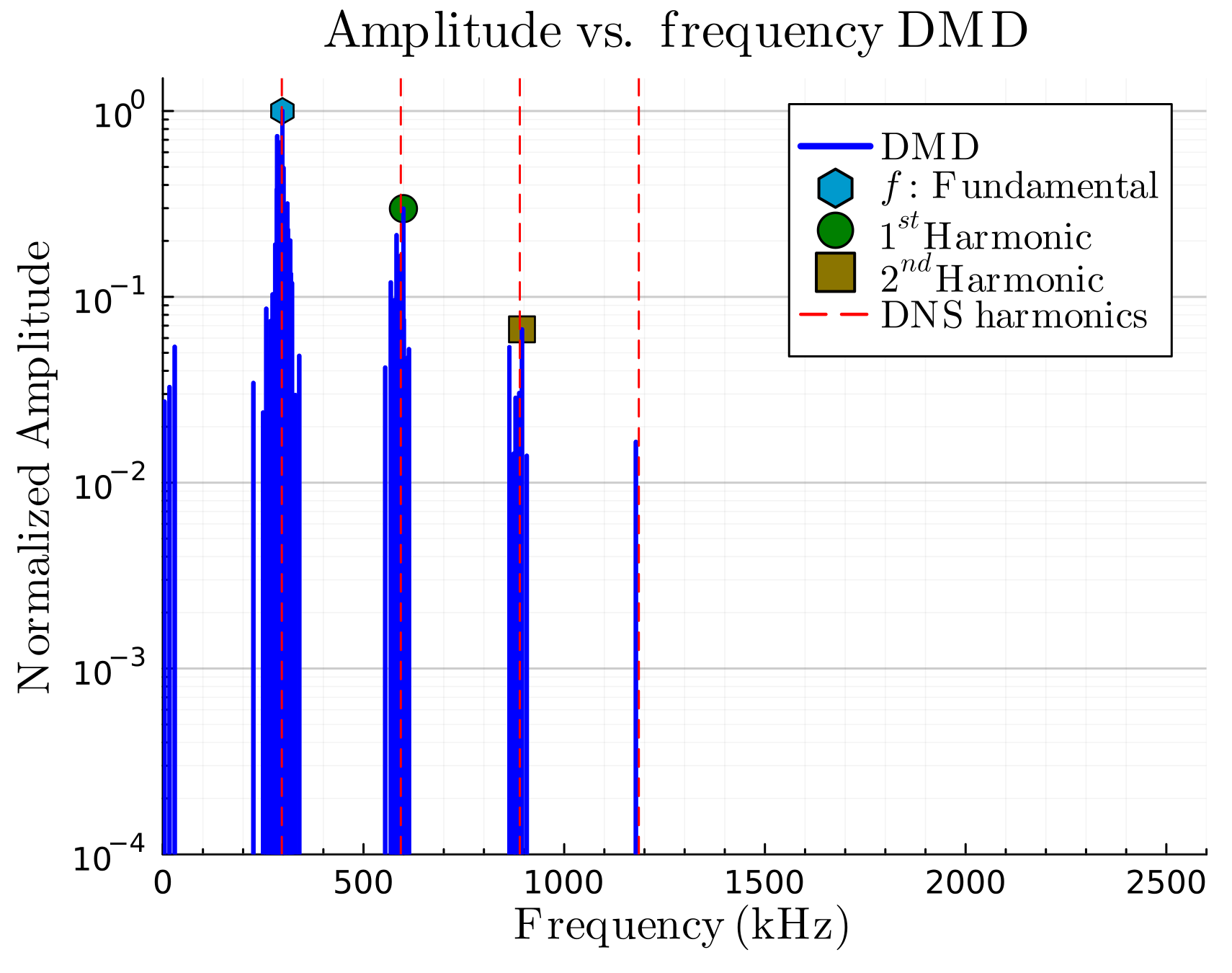}
\caption{}
\label{fig:dmd_amp}
\end{subfigure}
\centering
\begin{subfigure}[]{0.3\textwidth}
\centering
\includegraphics[width=1\textwidth]{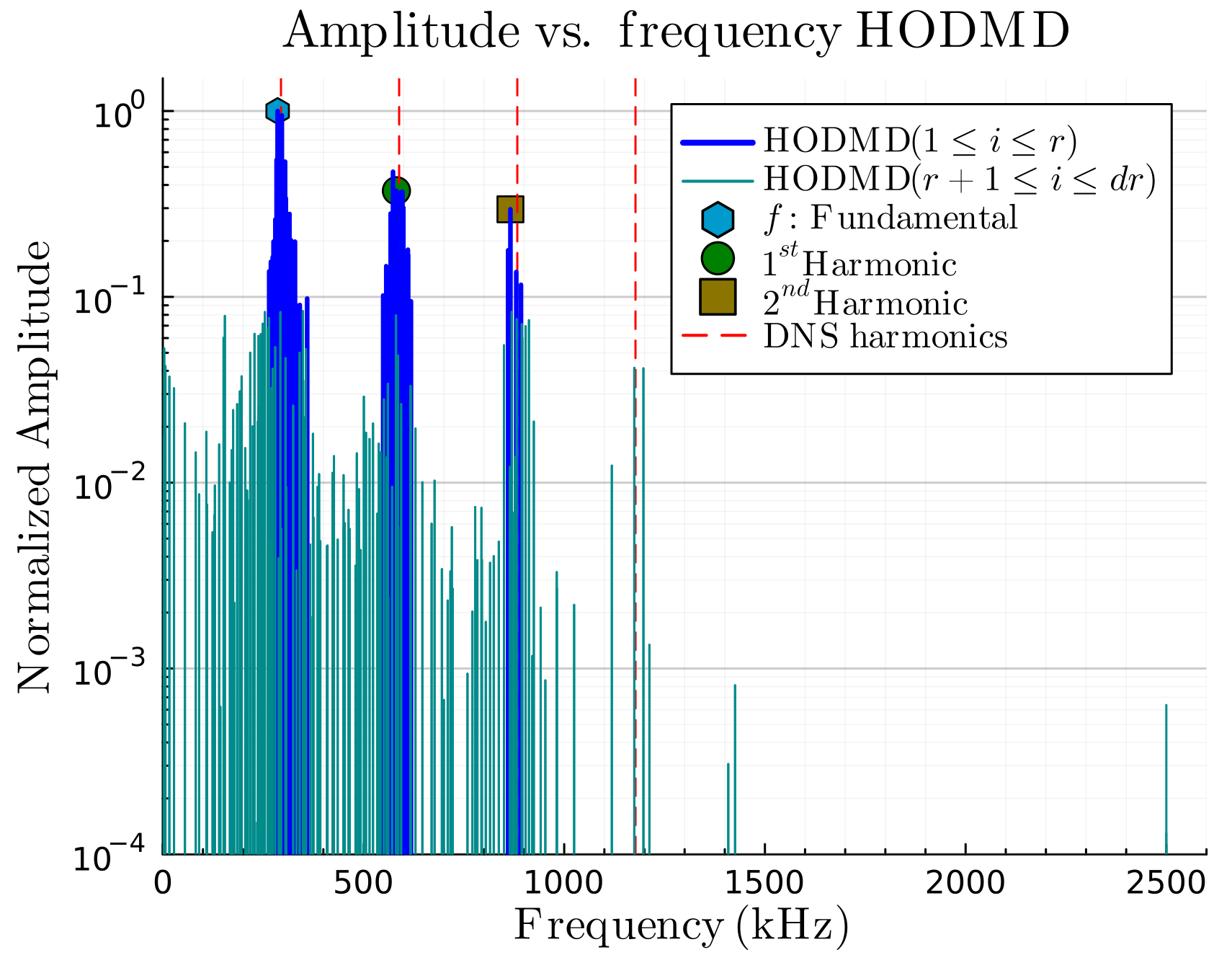}
\caption{}
\label{fig:hodmd_amp}
\end{subfigure}
\begin{subfigure}[]{0.3\textwidth}
\centering
\includegraphics[width=1\textwidth]{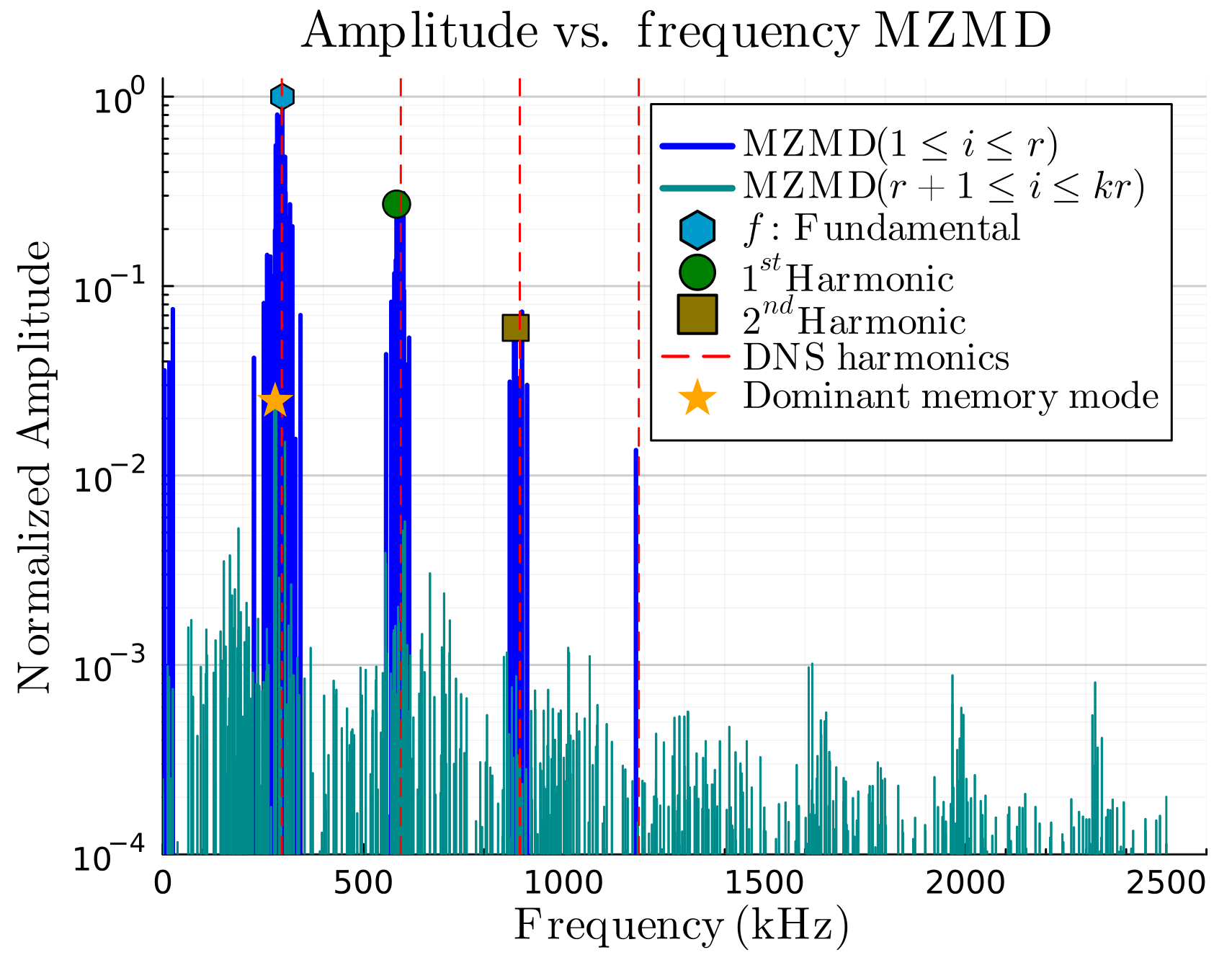}
\caption{}
\label{fig:mzmd_amp}
\end{subfigure}
\begin{subfigure}[]{0.3\textwidth}
\centering
\includegraphics[width=1\textwidth]{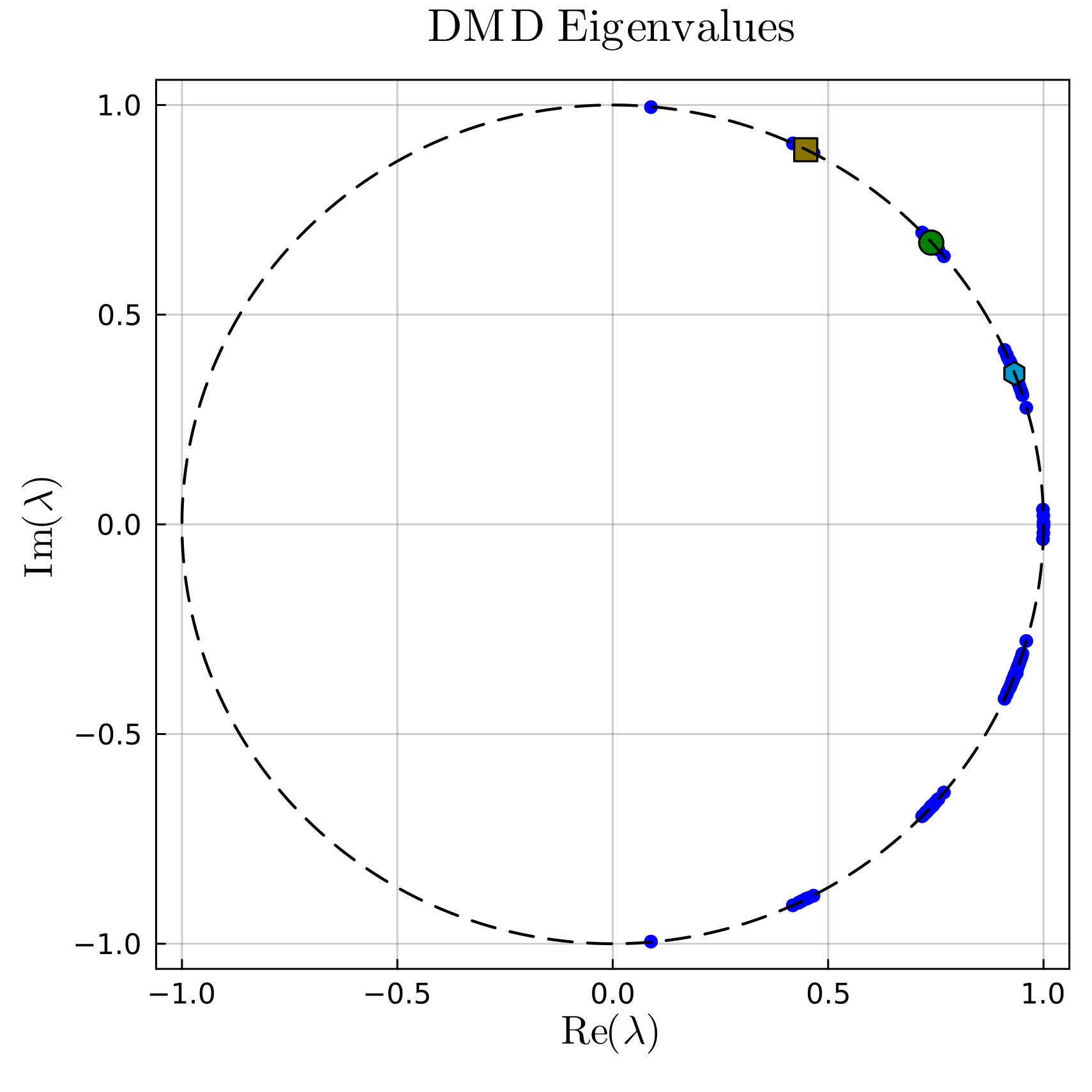}
\caption{}
\label{fig:dmd_ev}
\end{subfigure}
\begin{subfigure}[]{0.3\textwidth}
\centering
\includegraphics[width=1\textwidth]{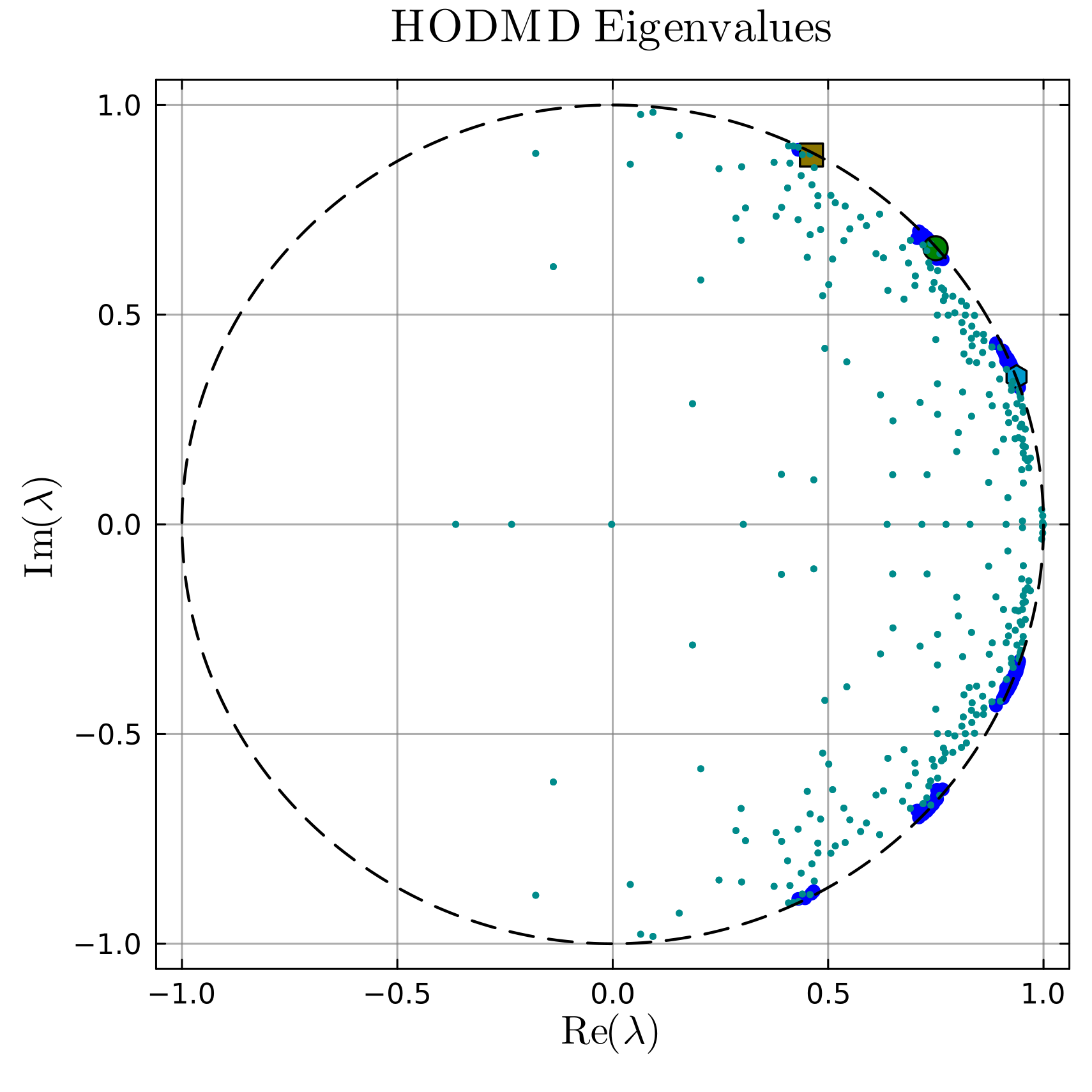}
\caption{}
\label{fig:hodmd_ev}
\end{subfigure}
\begin{subfigure}[]{0.3\textwidth}
\centering
\includegraphics[width=1\textwidth]{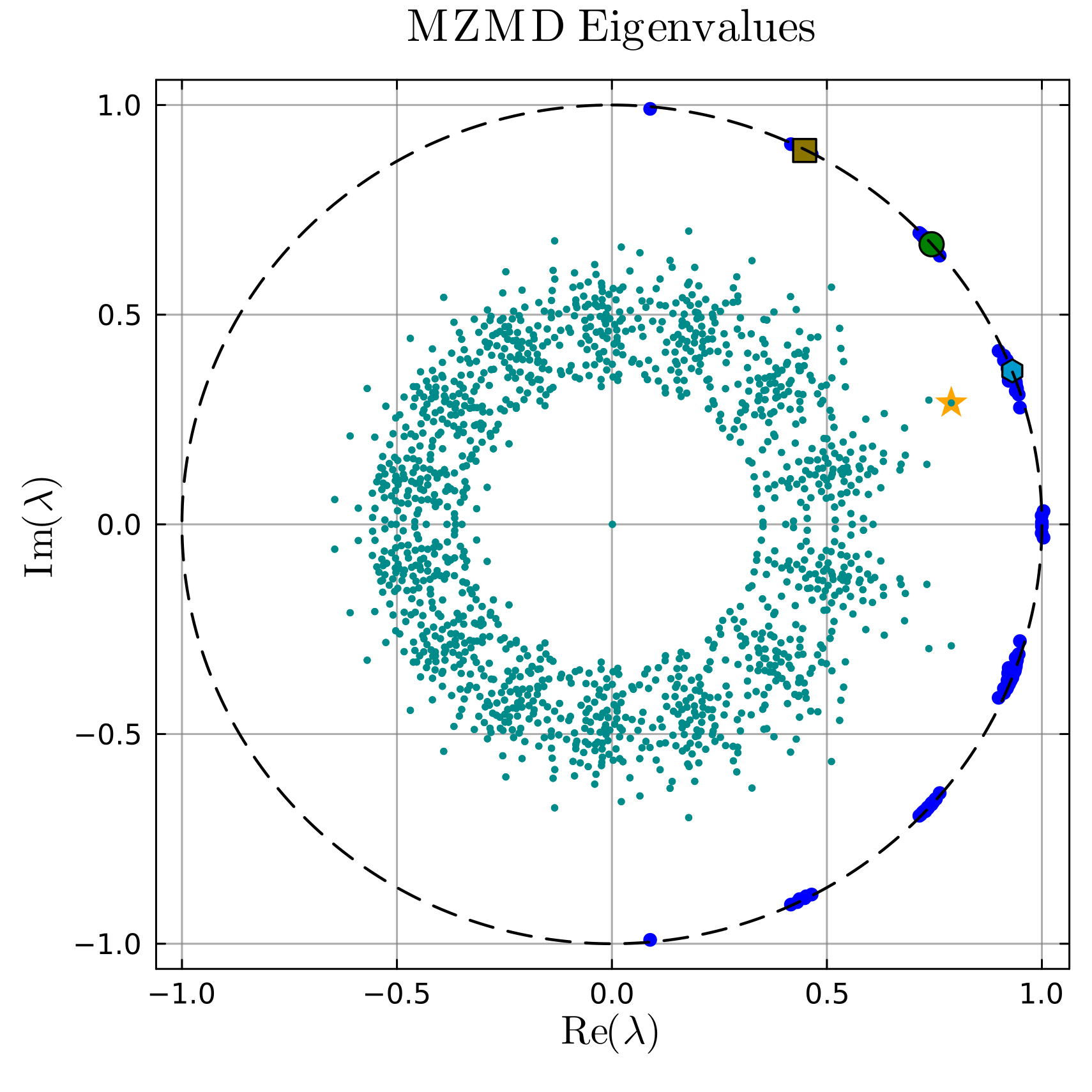}
\caption{}
\label{fig:mzmd_ev}
\end{subfigure}
\caption{Normalized amplitudes $a_i/\max(\bm a)$ (\subref{fig:dmd_amp}, \subref{fig:hodmd_amp}, \subref{fig:mzmd_amp}) and eigenvalues $\lambda$ (\subref{fig:dmd_ev}, \subref{fig:hodmd_ev}, \subref{fig:mzmd_ev}) of DMD, HODMD, and MZMD modes respectively, where the ``dominant'' modes associated with the fundamental frequency (and its higher harmonics) are uniquely labeled.  Parameters are selected as above; namely $r = 100$, $k=15$ and $d=4$. In the top row, the vertical red dashed lines represent the dominant harmonics as computed from DNS. For HODMD and MZMD, the $|\lambda|$ values are sorted largest to smallest, then the first $r$ modes are colored blue, and last $r+1$ to $rd$ ($rk$ respectively) modes are colored green. We isolate the dominant MZ memory mode as the mode which has the largest $|\lambda|$ value that is introduced by memory (highlighted by the gold star) which is later shown in Figure \ref{fig:phi_mz}. The eigenvalues determine the temporal evolution of the corresponding modes.  We observe that there are the same number of periodic eigenvalues between DMD, HODMD and MZMD, thus, in this flow time-delay embeddings and MZ memory serve to introduce transient modes, i.e. modes lying inside the unit circle. We also observe that time-delay embeddings and MZ memory produces small perturbations to the eigenvalues and amplitudes corresponding to DMD modes. }
\label{fig:evals_spectrum}
\end{figure}

\begin{figure}
\centering
\begin{subfigure}[]{0.6\textwidth}
\centering
\includegraphics[width=1\textwidth]{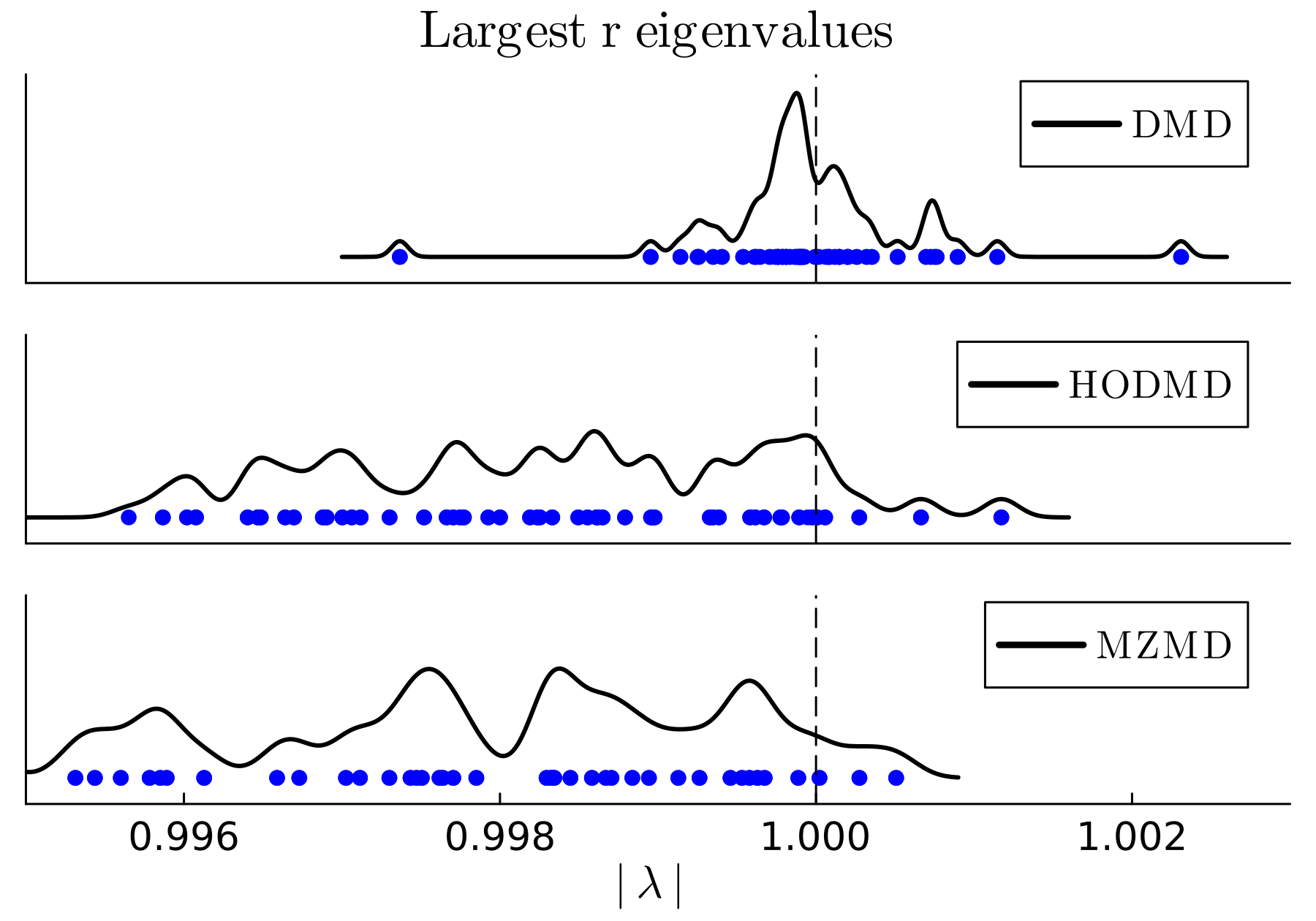}
\end{subfigure}
\caption{The modulus of the largest $r$ eigenvalues from each method, showing that the number of unstable modes decreases by adding time delay embeddings and MZ memory compared to DMD; where MZMD is seen to be the most stable, but also increases the spread of the modulus of eigenvalues.}
\label{fig:evals_spectrum_stability}
\end{figure}

\begin{figure}
\centering
\begin{subfigure}[]{1\textwidth}
\centering
\includegraphics[width=1\textwidth]{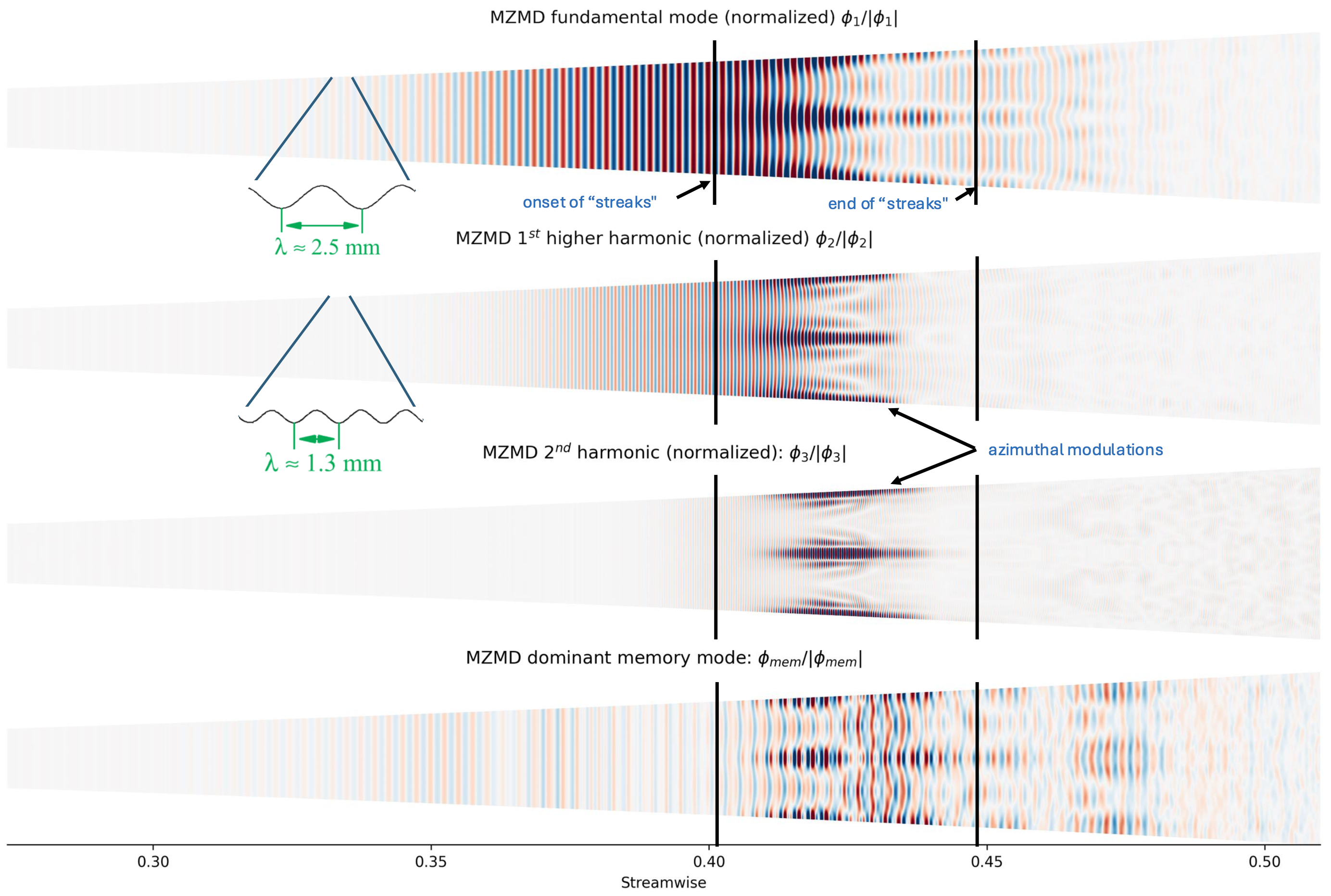}
\end{subfigure}
\caption{Dominant MZMD modes and the first dominant mode that memory introduces. Scale of the contours range from (-1,1)}
\label{fig:phi_mz}
\end{figure}

\begin{figure}[!htb]
\centering
\begin{subfigure}[]{1\textwidth}
\centering
\includegraphics[width=1\textwidth]{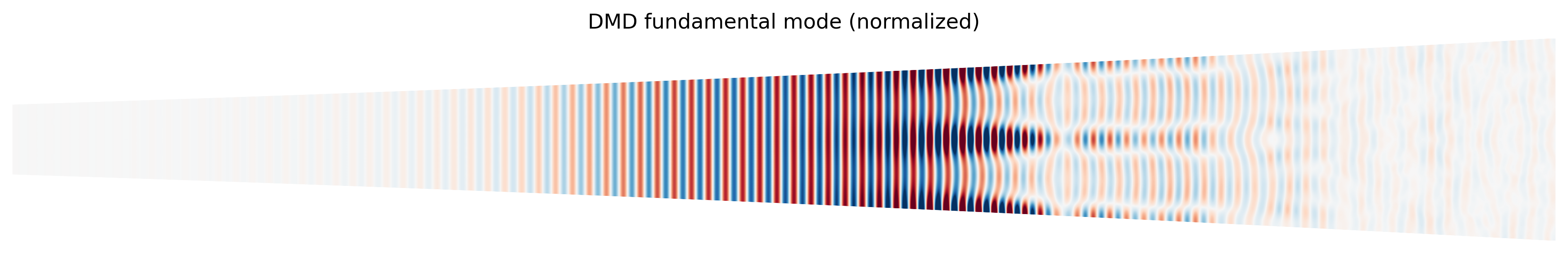}
\end{subfigure}
\centering
\begin{subfigure}[]{1\textwidth}
\centering
\includegraphics[width=1\textwidth]{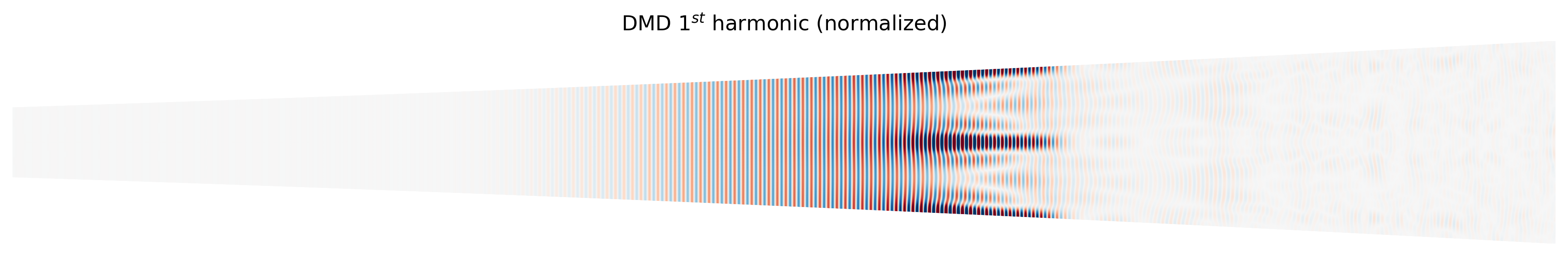}
\end{subfigure}
\begin{subfigure}[]{1\textwidth}
\centering
\includegraphics[width=1\textwidth]{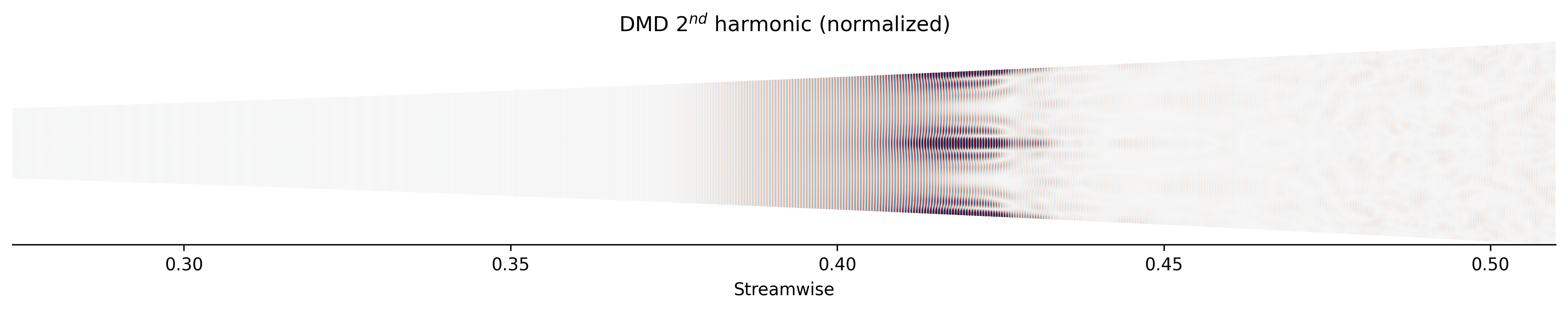}
\end{subfigure}
\caption{Dominant DMD modes; (top) fundamental, (middle) $1^{st}$ harmonic, (bottom) $2^{nd}$ harmonic.}
\label{fig:phi_dmd}
\end{figure}

\begin{figure}[!htb]
\centering
\begin{subfigure}[]{1\textwidth}
\centering
\includegraphics[width=1\textwidth]{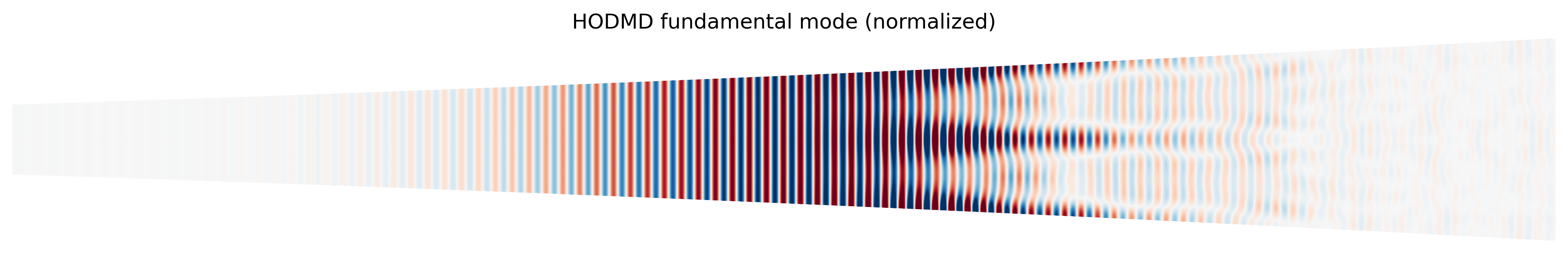}
\end{subfigure}
\centering
\begin{subfigure}[]{1\textwidth}
\centering
\includegraphics[width=1\textwidth]{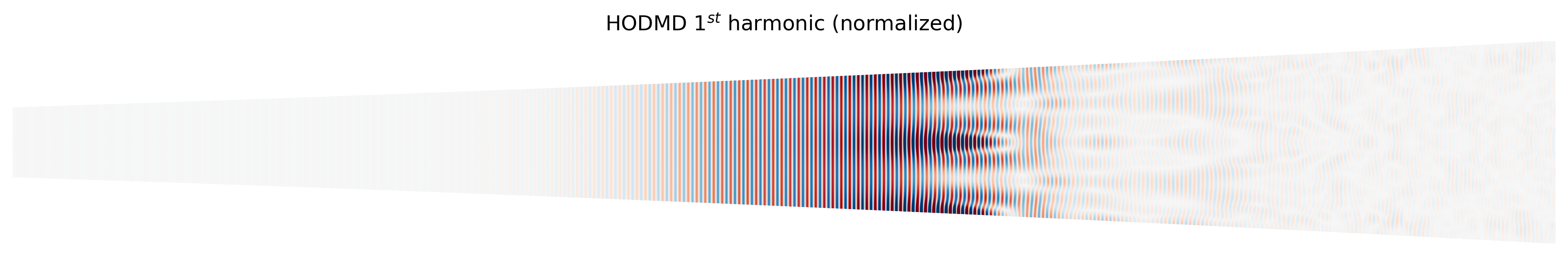}
\end{subfigure}
\begin{subfigure}[]{1\textwidth}
\centering
\includegraphics[width=1\textwidth]{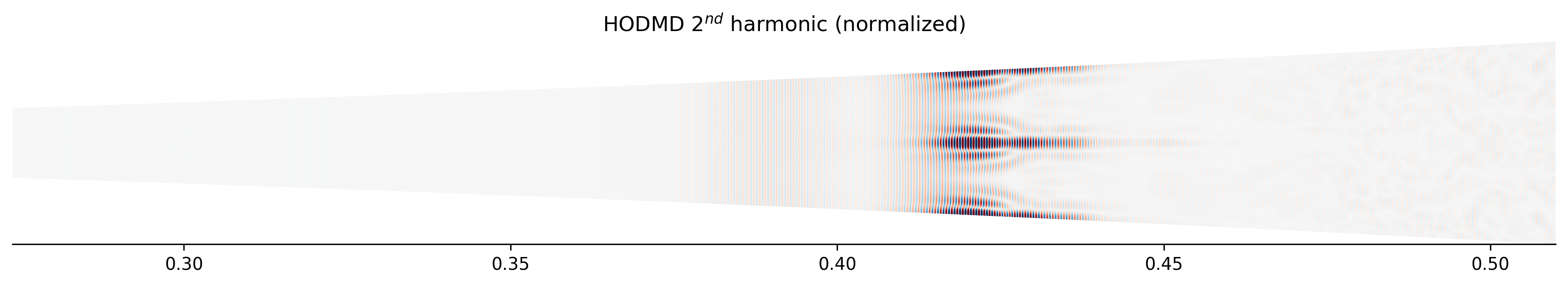}
\end{subfigure}
\caption{Dominant HODMD modes; (top) fundamental, (middle) $1^{st}$ harmonic, (bottom) $2^{nd}$ harmonic.}
\label{fig:phi_hodmd}
\end{figure}


\section{Discussion and conclusion}
\label{sec:conclusion}

In this manuscript, we have introduced the Mori--Zwanzig Mode Decomposition (MZMD); a novel modal analysis technique that leverages the Mori--Zwanzig (MZ) formalism to derive data-driven approximate memory closures based on the residual of Dynamic Mode Decomposition (DMD). The proposed MZMD approach extracts large-scale spatiotemporal structures from high-dimensional nonlinear dynamical systems by explicitly incorporating nonlocal in time MZ memory kernels. These memory kernels quantify the influence of unresolved variables on the resolved modes, addressing a fundamental limitation of DMD when the selected state-space observables do not form a Koopman-invariant subspace. By explicitly modeling how the unresolved variables interact with the resolved variables, which predominantly accounts for the nonlinear dynamics missed by DMD, MZMD significantly enhances the ability to resolve strongly nonlinear dynamics in the flow compared to classical DMD. 

We find that MZMD provides distinct physical insights beyond classical DMD techniques; achieved by explicitly modeling memory effects arising from unresolved nonlinear interactions. Specifically, MZMD identifies transient memory modes that reveal previously hidden nonlinear interactions, aiding in resolving spatially localized events that drive critical flow features such as the hot-streak formation. Additionally, MZMD enhances frequency predictions of primary and harmonic modes, achieving closer agreement with numerical simulations, and improving stabilization of the resulting reduced order model. Furthermore, by approximating these nonlinear interactions through MZ memory, MZMD also improves the representation of energy cascades across a broader range of temporal scales. Finally, isolating the spatial regions influenced by individual memory terms reveals the cumulative impact of transient memory modes. Incrementally incorporating additional memory terms progressively influences the upstream flow features responsible the generation of the hot streaks, which substantially impact the nonlinear transition mechanisms.

At this point, it is helpful to highlight the differences between MZMD and HODMD. While both methods increase the spectral complexity relative to DMD, a key advantage of MZMD is in its computational efficiency relative to time-delay embedding methods such as Higher-Order DMD (HODMD). Unlike HODMD, which depends on embedding state-space histories, constructing large Hankel matrices, and performing two separate SVD truncations, MZMD directly incorporates the Mori--Zwanzig formalism through the Generalized Langevin Equation (GLE). Furthermore, MZMD involves only a single SVD application and enforces the Generalized Fluctuation--Dissipation (GFD) relation to recursively approximate the memory kernel. Consequently, MZMD avoids the computational complexity and potential for overfitting associated with large Hankel matrices employed by HODMD, making it more suitable for high-dimensional, strongly nonlinear systems. Furthermore, MZMD explicitly preserves the original companion matrix structure extending DMD, and ensures the Markovian term remains consistent to DMD while systematically adding memory terms. This recursive addition of memory terms via the GFD relation can not only improve numerical stability but also provides clear physical interpretability of the influence exerted by the unresolved dynamics.

Additional distinctions between HODMD and MZMD arise both in their theoretical foundations and the algorithms employed. The HODMD assumption of the presence of higher-order Koopman representation is primarily developed in the context of providing a modal decomposition technique so that the spectral complexity M can be greater than the spatial complexity N (as defined in Section \ref{sec:hodmd}). Furthermore, HODMD is functionally identical with utilizing time delay embeddings \citep{hodmd_2017}, operating on Hankel matrices formed by time-delayed observables. With MZMD, the higher-order representation is no longer an ansatz, but a consequence of the MZ formalism, naturally achieving a higher-order representation directly from the structure of the Generalized Langevin Equation (GLE). Functionally, HODMD employs time-delay embeddings and a one-shot optimization to approximate the block companion matrix, along with an additional SVD truncation of the Hankel matrix of delay embeddings (which is not required in MZMD). This additional SVD truncation can, in some cases, break the original model assumption as we demonstrated in the Appendix (the companion matrix no longer satisfies the original structure assumed). HODMD can also face challenges of overfitting when applied to high-dimensional complex systems, as demonstrated in the this work. 

Overall, MZMD approximates the modes and eigenvalues of the discrete-time Generalized Langevin equation (GLE) of state measurements by leveraging the MZ formalism and enforces the Generalized Fluctuation Dissipation (GFD) relation to construct an approximate memory closure model for DMD. Memory terms will almost always be non-trivial when applying MZMD to high dimensional nonlinear dynamical systems. Importantly, in MZMD, the first term remains equivalent to the DMD solution, even when memory effects are considered. Each MZ memory term is learned (via one shot optimization) recursively via the GFD relation. This enforces the companion structure by construction, ensuring that the inclusion of MZ memory terms is directly additive to DMD and does not alter the results of the Markovian term ($\bm \Omega_0$) derived from DMD. However, the modes of MZMD and DMD are different when memory is included. This is due to the coupling of the memory terms with the Markovian term as seen in the polynomial eigenvalue problem associated with the block companion matrix, causing perturbative effects of the modes and eigenvalues. Finally, it is important to mention that MZ memory and time-delay embeddings are not mutually exclusive, as demonstrated in \citet{lin22_nn_mz, woodward_aviation}.

To validate the proposed methodology, we applied DMD, HODMD and MZMD to two distinct Direct Numerical Simulation (DNS) datasets: a canonical two-dimensional cylinder flow (validation case presented in the Appendix) and a hypersonic laminar-to-turbulent boundary-layer transition over a flared cone at Mach 6. Our results clearly demonstrate the advantages of MZMD in capturing complex nonlinear structures, notably surpassing both DMD and HODMD in terms of prediction accuracy and numerical stability. Specifically, MZMD exhibited slower error growth rates and superior robustness, monotonically improving predictive performance as additional memory terms were incorporated. In the critical transition region where nonlinear effects (such as hot streak formation on the surface) dominate, MZMD significantly outperformed the other methods, highlighting the practical benefits of incorporating MZ memory kernels.

Despite the demonstrated advantages, none of the methods investigated (DMD, HODMD, MZMD) inherently guarantee long-term stability, as eigenvalues may still appear outside the unit circle in discrete-time dynamics. Nevertheless, both HODMD and MZMD substantially improved the stability of standard DMD, with MZMD providing the greatest stability in the hypersonic boundary-layer flow example.

Several promising research directions naturally follow from this work. The straightforward relationship established between DMD and MZMD facilitates integration of MZ memory kernels into existing DMD frameworks, such as DMD with control or online DMD \citep{dmd_book, online_dmd}. Additionally, future theoretical investigations might explore conditions under which the MZ kernels are guaranteed to stabilize the DMD eigenvalues, as well as more sophisticated modeling of the orthogonal dynamics within the MZ formalism (which is non-Gaussian). Introducing physical constraints and symmetries into the MZ operators through physics-informed machine learning \citep{woodward23_piml_sph, tian2023_lles} and regression-based projection methods \citep{lin22_nn_mz} represents another compelling direction. Overall, we have demonstrated that MZMD provides a robust, efficient, and interpretable generalization of DMD, opening new opportunities for modeling, prediction, and control of complex nonlinear dynamical systems across diverse scientific and engineering disciplines.

\textbf{Acknowledgments.} This work has been authored by employees of Triad National Security, LLC which operates Los Alamos National Laboratory (LANL) under Contract No. 89233218CNA000001 with the U.S. Department of Energy (DOE)/National Nuclear Security Administration. We acknowledge supports from LANL’s Laboratory Directed Research and development (LDRD) program, project number 20220104DR, 
and computational resources from LANL’s Institutional Computing (IC) program. This work was also partially supported by the U.S. Department of Energy, Office of Science, Office of Advanced Scientific Computing Research's Applied Mathematics Competitive Portfolios program.

\textbf{Declaration of interests.} The authors report no conflict of interest.

\bibliographystyle{abbrvnat}
\bibliography{references}

\appendix

\section{Algorithms}\label{sec:algs}

In this appendix, we provide an overview of the MZMD algorithm (Algorithm \ref{alg:mzmd}) and establish connections with the DMD and HODMD algorithms. In MZMD (Algorithm \ref{alg:mzmd}) we first obtain the snapshot data as is done in POD and DMD but we also need to include some past history for which the parameter $k$ is used. In each algorithm, we can interpret the SVD as a linear auto-encoder \citep{GoodBengCour16} for selecting observables from data to avoid the intractable computations of the full two time covariance matrix $\tilde{\bm C}_k = \bm X_k \cdot \bm X_1^T$. This is equivalent to performing a low rank approximation of the full $\tilde{\bm C}_k$ by projecting onto the POD modes (which is what is done in the standard DMD algorithm). This equivalence can be seen by replacing $\bm G_k$ with $\bm U_r^* \bm X_k $ when computing $\bm C_k$, which results in $\bm C_k = \bm U_r^* \bm X_k (\bm U_r^* \bm X_0)^T = \bm U_r^* \bm X_k \bm X_0^T \bm U_r  \approx \tilde{\bm C}_k$. By construction, MZMD adheres to and is formulated with the GFD relation (Eq.~\ref{eq:gfd}), a necessary condition to be consistent with the MZ formalism. 

\begin{algorithm}[H]
    \centering
    \caption{MZMD}
    \begin{algorithmic}[1]
        \State Select the number of memory terms $k$
        \State Given $T+1$ snapshots of data: $\bm X = [\bm x_0, ..., \bm x_{T}]$
        \State Compute truncated SVD $\bm X \approx \bm U_r \bm \Sigma_r \bm V_r^*$ 
        \State Project onto POD modes $\bm G = \bm U_r^* \bm X$ 
        
        \State Collect snapshots over $k$ time delays: $ \bm G_0 = [\bm g_0, \bm g_1, ..., \bm g_{T-k}]$, $ \bm G_1 = [\bm g_1, \bm g_2, ..., \bm g_{T-k+1}]$, ... $\bm G_{k} = [\bm g_{k}, \bm g_{k+1}, ..., \bm g_{T}]$ \\
        
        \vspace{0.3cm}
        $\bm C_0 = \bm G_0 \bm G_0^T$
        \State for $i\in \{1,...,k\}$
        \State \hspace{4mm} $\bm C_i = \bm G_{i} \bm G_{0}^T$
        
        \State $\bm \Omega_0 = \bm C_1 \bm C_0^{-1}$
        \State for $i\in \{1,...,k\}$
        \State \hspace{4mm} $\bm \Omega_i = \left[ \bm C_{i+1} - \sum_{l=0}^{i-1} \bm \Omega_l \bm C_{i-l} \right]  \bm C_0^{-1}$
        \\
            \State Form companion matrix $\bm C_g$, then compute eigendecomposition $\bm C_g \bm W = \bm W \bm \Lambda$\\
            \State \textbf{MZMD modes}: $\bm \phi_{i}^0 = \bm U_r \bm w^{0}_i$, where $\bm w^0_i = \bm P_0 \bm w_i$ and $\bm P_0 = [\bm I \hspace{2mm} \bm 0 \hspace{2mm} ... \hspace{2mm} \bm 0]$. \\
            \State Compute amplitudes: $\bm a  = \bm \Phi^{\dagger} [\bm x_0, \bm x_1, ..., \bm x_k]^T$ \\
            \State $\bm x_{n+1} \approx \sum_i^{rk} a_i \lambda_i^{n+1} \bm \phi_i^0$ (selected modes prediction)
    \end{algorithmic}
    \label{alg:mzmd}
\end{algorithm}

\begin{algorithm}[H]
    \centering
    \caption{DMD}
    \begin{algorithmic}[1]
        \State Given $T+1$ snapshots of data: $\bm X = [\bm x_0, ..., \bm x_{T}]$
        \State $\bm X \approx \bm U_r \bm \Sigma_r \bm V_r^*$ Truncated SVD
        \State Low rank projection $\bm G = \bm U_r^* \bm X$ 
        \State $ \bm G_0 = [\bm g_0, \bm g_1, ..., \bm g_{T-1}]$, $ \bm G_1 = [\bm g_1, \bm g_2, ..., \bm g_{T}]$ \\
        \vspace{0.3cm}
        $\bm C_0 = \bm G_0 \bm G_0^T$\\
        $\bm C_1 = \bm G_1  \bm G_{0}^T$
        \State $\bm K = \bm C_1  \bm C_0^{\dagger}$
            \State Compute eigendecomposition $\bm K \bm W = \bm W \bm \Lambda$
            \State \textbf{DMD modes}: $\bm \phi_{i} = \bm U_r \bm w_i, \quad \bm \Phi = [\bm \phi_1, ..., \bm \phi_r]$;
            \State Given a new initial condition $\bm x_0$, compute amplitudes: $\bm a = \bm \Phi^{\dagger} \bm x_0$
            \State $\bm x_{n+1} \approx \sum_i^{r} a_i \lambda_i^{n+1} \bm \phi_i$
    \end{algorithmic}
    \label{alg:dmd}
\end{algorithm}

\section{Numerical validation: 2D flow over cylinder}
\label{sec:2d_cylinder}

In this section, we validate the DMD, HODMD, and MZMD methodologies on a cononical 2D flow over cylinder at Reynolds number $Re=200$. This flow serves as a simple test-bed for validation, performing analysis and comparisons of DMD, HODMD and MZMD. The flow past a two-dimensional cylinder at low Reynolds numbers results in periodic structures in the wake. The wake forms a pattern of alternating vortices, commonly referred to as a von Kármán vortex street. This vortex street over a cylinder has been extensively studied \citep{Tritton_1959, Dennis_Chang_1970, LINNICK2005157, colonius2008}. In this work, we use results at $Re=200$ to investigate the effects of adding time-delay embeddings vs MZ memory. 

In figure \ref{fig:dns_cylinder} we show a snapshot of the vorticity for the $Re=200$ case, which was obtained by the use of the immersed boundary projection method developed by \citet{colonius2008}. We consider the statistically stationary flow reached after the initial transient that evolves from the unstable equilibrium point to the equilibrium point representing the von Kármán vortex street. The numerical simulations are performed with $dt=0.02$ (s). The simulations are then sampled at every $10^{th}$ time step. The wake shedding frequency is approximately $f_s \approx 0.197$ (HZ) for the $Re=200$ case \citep{LINNICK2005157, colonius2008}. $f_s$ is associated with the dominant coherent structure of the flow. In this study, we consider the first and second higher harmonics (i.e. $2f_s$ and $3f_s$), which are associated with the next two dominant modes (ranked by amplitude) found in each modal decomposition technique.  This flow exhibits a low rank structure as seen in the total variation contained in the first dominant POD modes shown in figure \ref{fig:dns_cylinder}. For the $Re=200$ case, we see that using more than 10 POD modes captures over $99\%$ of the variation in the vorticity field.

We first validate that DMD, HODMD, and MZMD can capture the dominant modes associated with the harmonics $\{f_s, 2f_s, 3f_s\}$ with the appropriate frequencies. This validation can be seen in figure \ref{fig:mzmd_hodmd_evals_cyl_re200}. The modes associated with these higher harmonics are shown in figure \ref{fig:mode_comparions_cyl_re200}. The influence of memory here is very small, and only slight changes can be observed in the modes. 

Next, we investigate what impact memory can have on improving over DMD. For this, we first show that in the case where DMD is truncated (with $r=5$), so that it does not capture the first higher harmonic (see figure \ref{fig:mzmd_hodmd_evals}), adding MZ memory (or time-delay embedding) can recover this missing dominant mode. This is associated with a dominant periodic mode occurring at the first higher harmonic (i.e. $2f_s$). Additionally, in figure \ref{fig:R_tilde_cylinder}, we see that adding memory (both MZ and time-delay embeddings) can improve upon DMD predictions, both for reconstruction and future state prediction (generalization). Each prediction is made over an ensemble of $40$ samples of $\bm z_0$ drawn independently from the test set representing a long statistically stationary flow.

Next, we demonstrate the differences in the algorithms of HODMD and MZMD.  HODMD and MZMD can result in different block companion matrices. For example, the second SVD that occurs in HODMD (recall the procedure described in section \ref{sec:hodmd}) can result in truncation errors that lead to breaking the original ansatz of HODMD, with the non-diagonal terms breaking the "higher-order" Koopman representation of HODMD. Figure \ref{fig:R_tilde_cylinder} illustrates how the assumed companion structure can be broken in HODMD if the second SVD (applied to the Hankel matrix) is truncated. This can happen, for example, when the number of time snapshots for fitting is less than $d*r_1$. In this case, the maximum value $r_2$ can take is $r_2 = T \leq d*r_1$, which forces $\hat{\bm G}$ to be a truncated SVD approximation. Subsequently, the HODMD algorithm could produce a coupled representation of past history which does not satisfy the original higher order Koopman representation \eqref{eq:hodmd_assumption}.  However, by construction, MZMD, which does not perform the second SVD, retains this structure. 

In figure \ref{fig:R_tilde_cylinder}, we also see that the memory terms that are obtained with HODMD do not decay and may even increase their contributions for the higher order terms. This behavior suggests a propensity for overfitting as more delay embeddings are still required. However, MZMD produces monotonic memory decay with the order of the memory term, which is consistent with the expectation that physical processes typically exhibit decaying cross-corelation in time. The relative contribution of the memory terms in MZMD decays to small values and flattens out around $k \approx 10$, which can be used as a memory length selection criterion. Furthermore, MZMD is much less expensive (as seen in figure \ref{fig:costs_convergence}), nearly an order of magnitude less expensive than HODMD, mainly due to the fact that HODMD performs an additional SVD. Furthermore, the learned companion matrices of MZMD converge faster with increasing number of samples compared to HODMD (Figure \ref{fig:convergence_ops}), indicating that less data may be required to build the MZMD model.

\begin{figure}[!htb]
\centering
\begin{subfigure}[]{0.55\textwidth}
\centering
\includegraphics[width=1\textwidth]{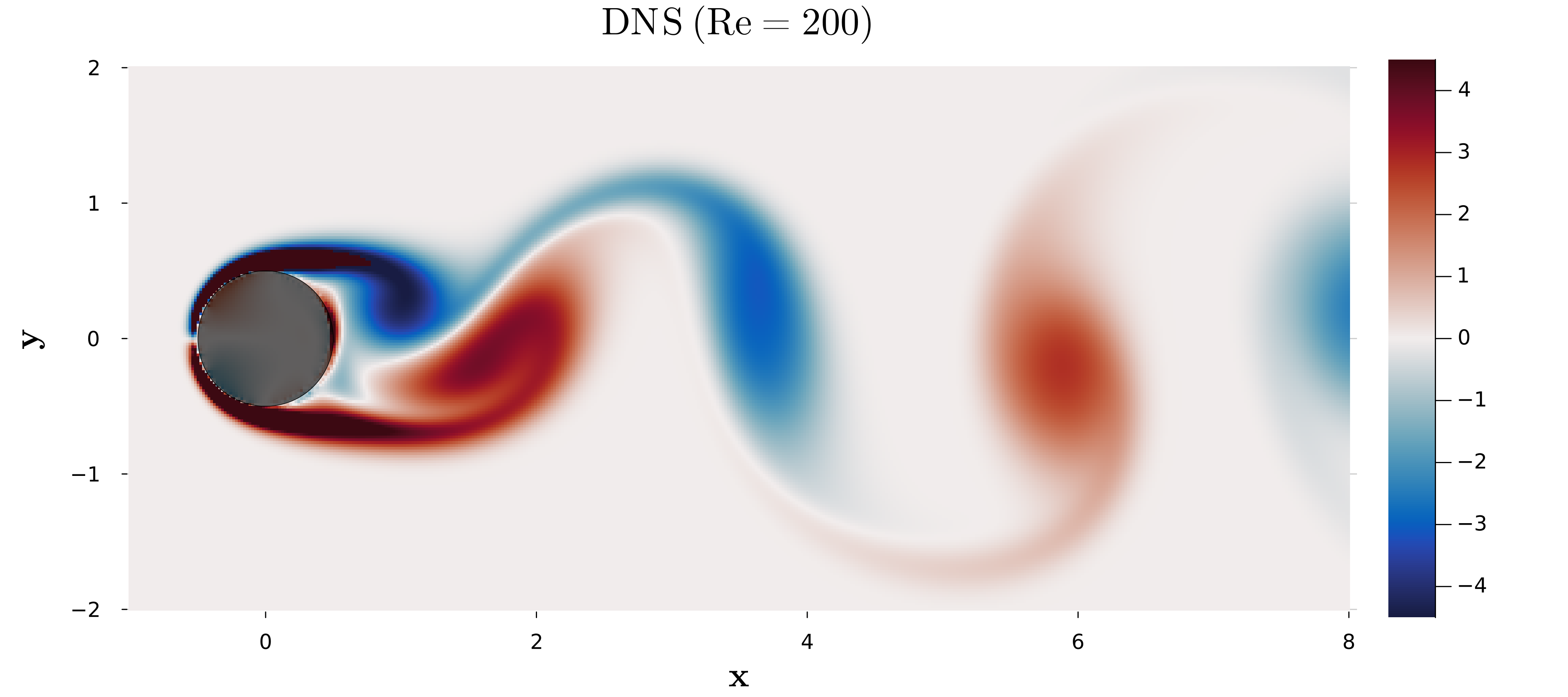}
\caption{}
\label{fig:dns_snap_cyl}
\end{subfigure}
\begin{subfigure}[]{0.44\textwidth}
\centering
\includegraphics[width=1\textwidth]{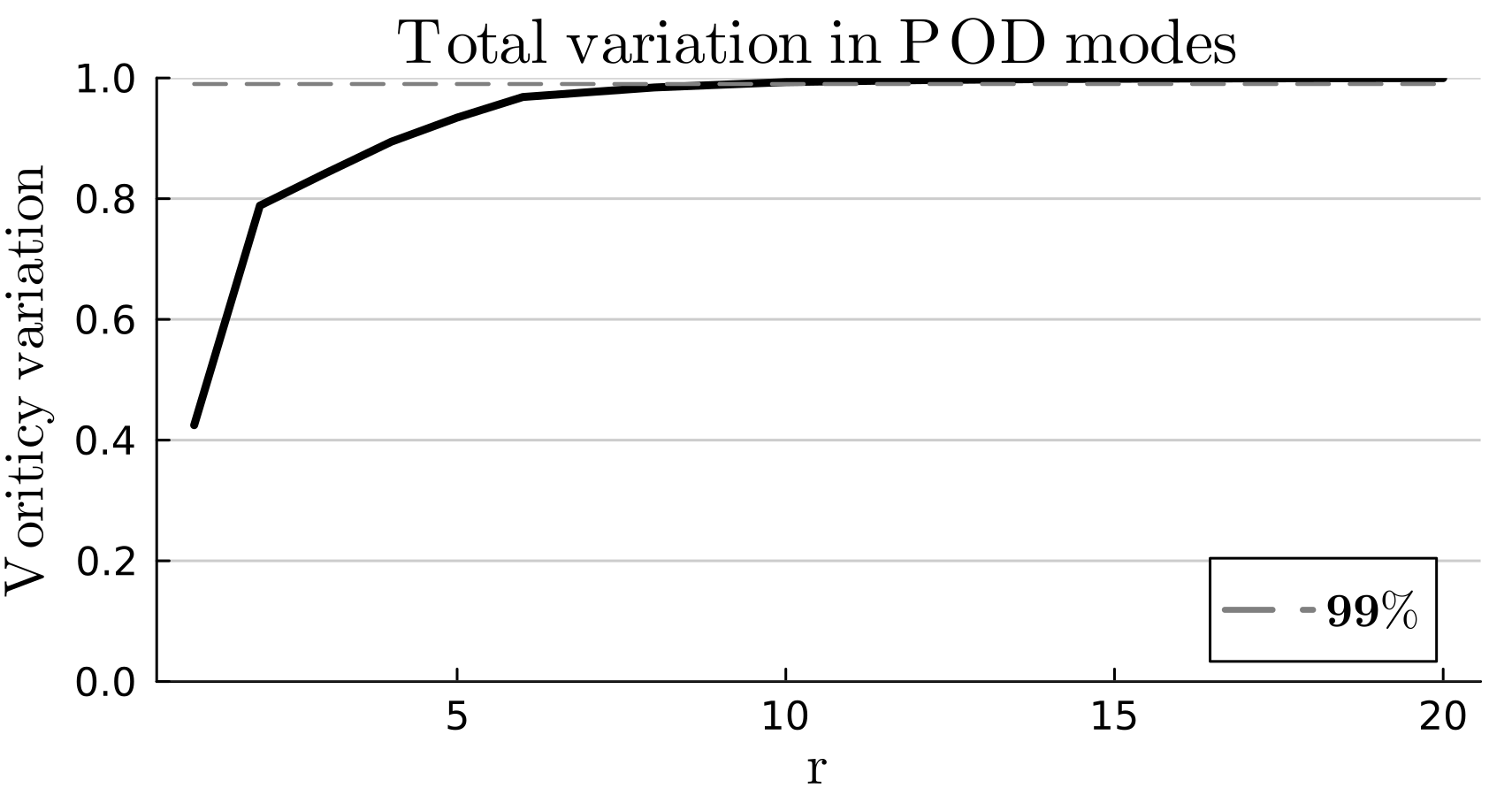}
\caption{}
\label{fig:pod_variation}
\end{subfigure}
\caption{(\subref{fig:dns_snap_cyl}) snapshot of 2D flow over cylinder with $Re=200$. (\subref{fig:pod_variation}) energy contained in the POD modes; demonstrating that the snapshot matrix contains low rank structure.} 
\label{fig:dns_cylinder}
\end{figure}

\begin{figure}[!htb]
\centering
\begin{subfigure}[]{0.3\textwidth}
\centering
\includegraphics[width=1\textwidth]{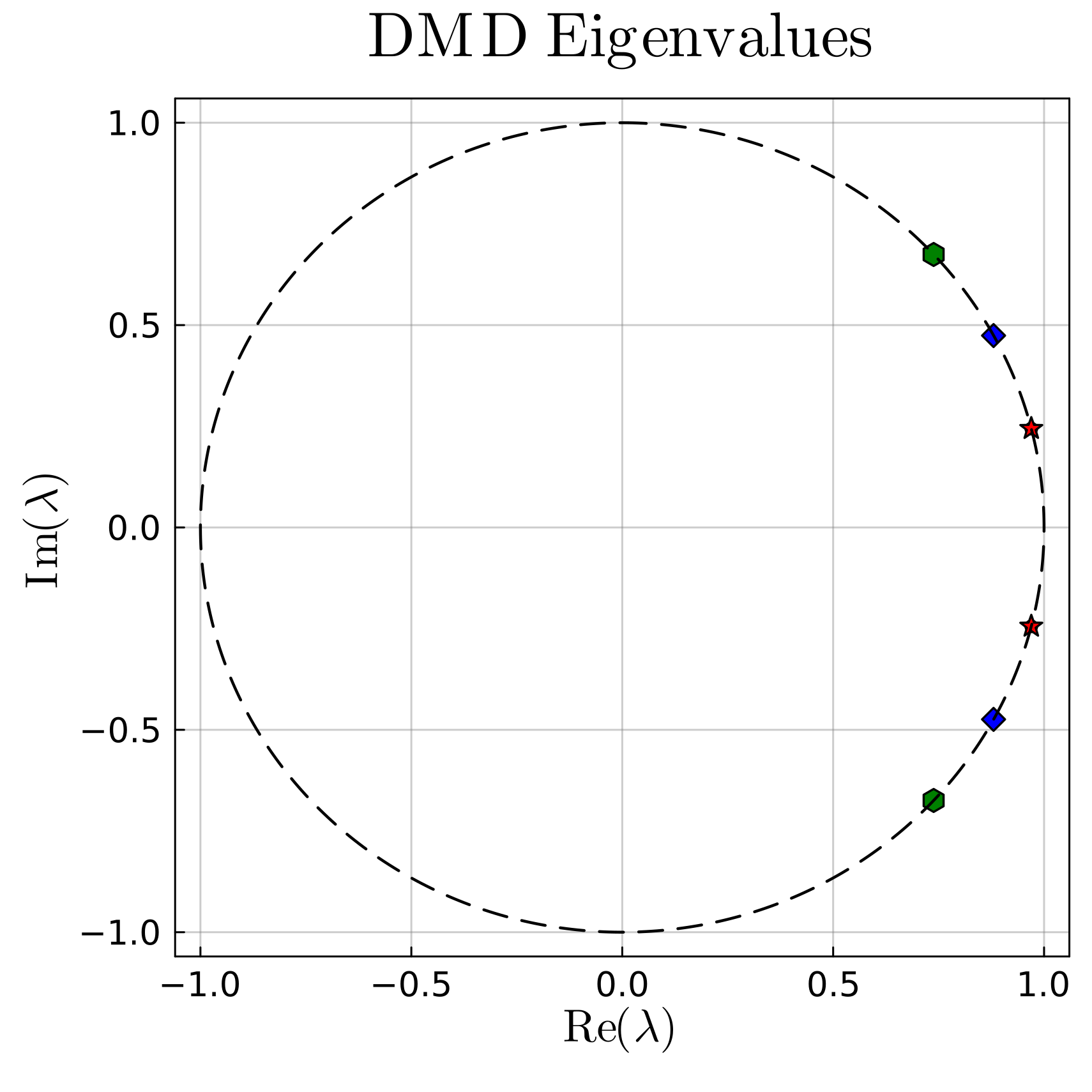}
\caption{}
\label{fig:dmd_eval_cyl}
\end{subfigure}
\begin{subfigure}[]{0.3\textwidth}
\centering
\includegraphics[width=1\textwidth]{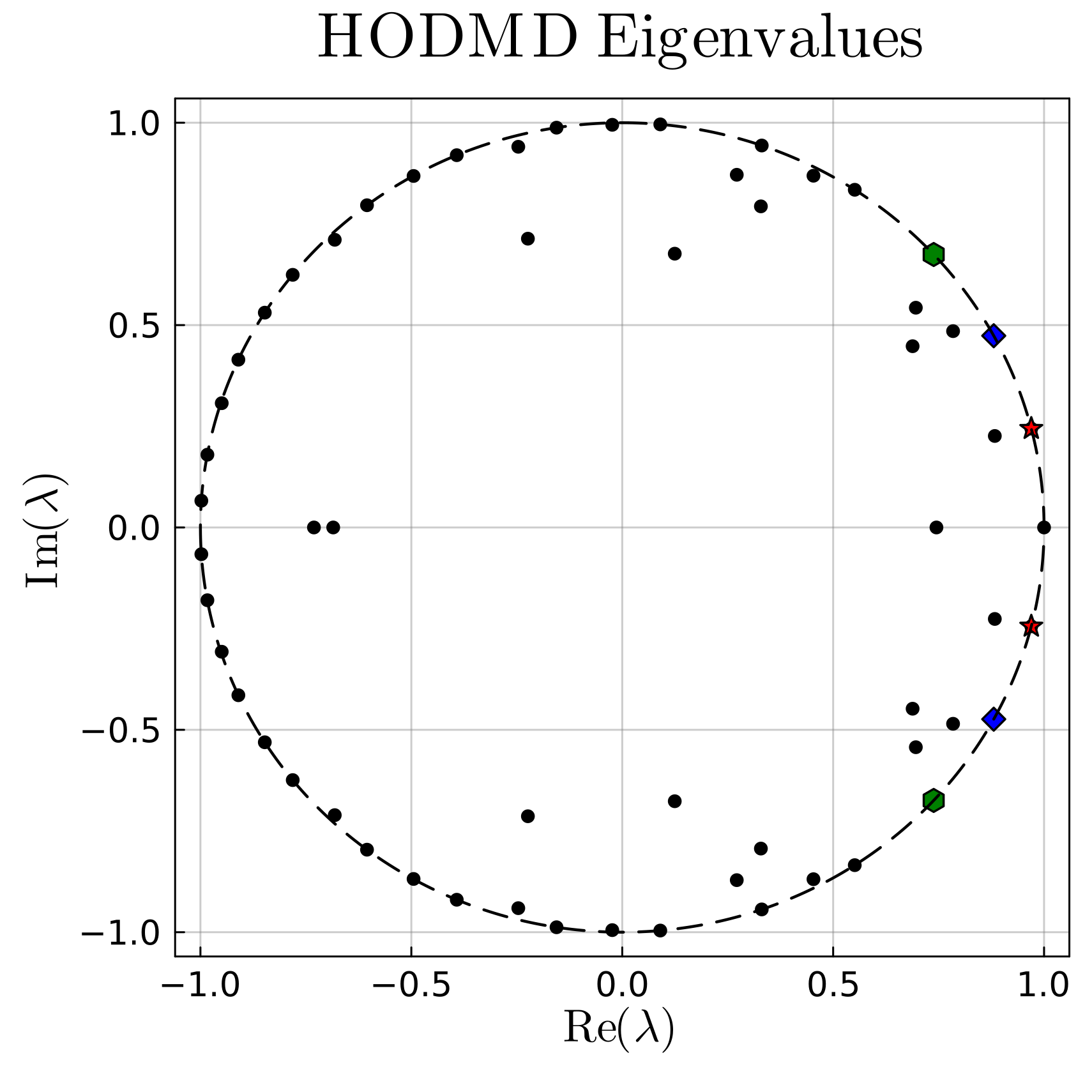}
\caption{}
\label{fig:hodmd_eval_cyl}
\end{subfigure}
\begin{subfigure}[]{0.3\textwidth}
\centering
\includegraphics[width=1\textwidth]{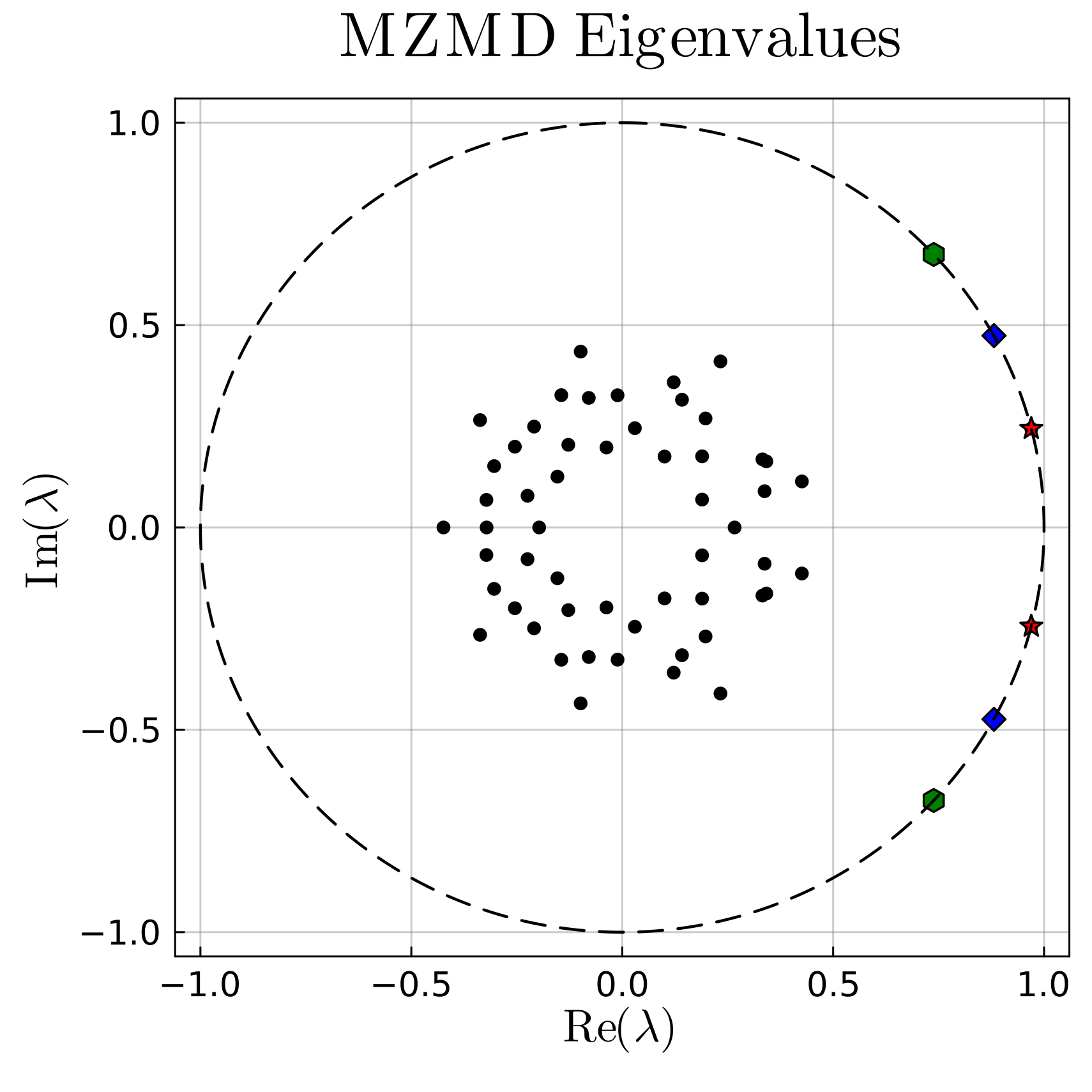}
\caption{}
\label{fig:mzmd_eval_cyl}
\end{subfigure}
\centering
\begin{subfigure}[]{0.3\textwidth}
\centering
\includegraphics[width=1\textwidth]{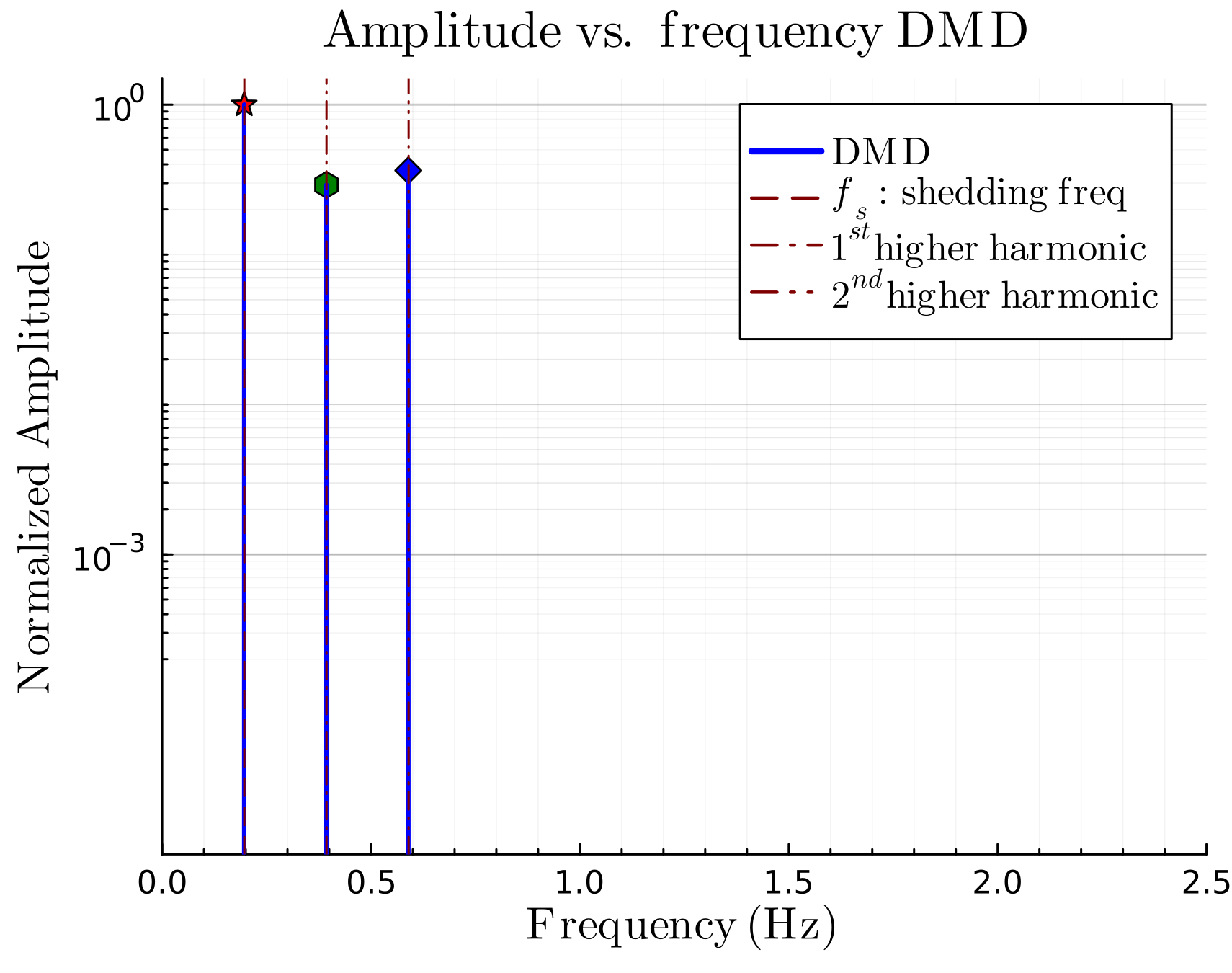}
\caption{}
\label{fig:dmd_amp_cyl}
\end{subfigure}
\centering
\begin{subfigure}[]{0.3\textwidth}
\centering
\includegraphics[width=1\textwidth]{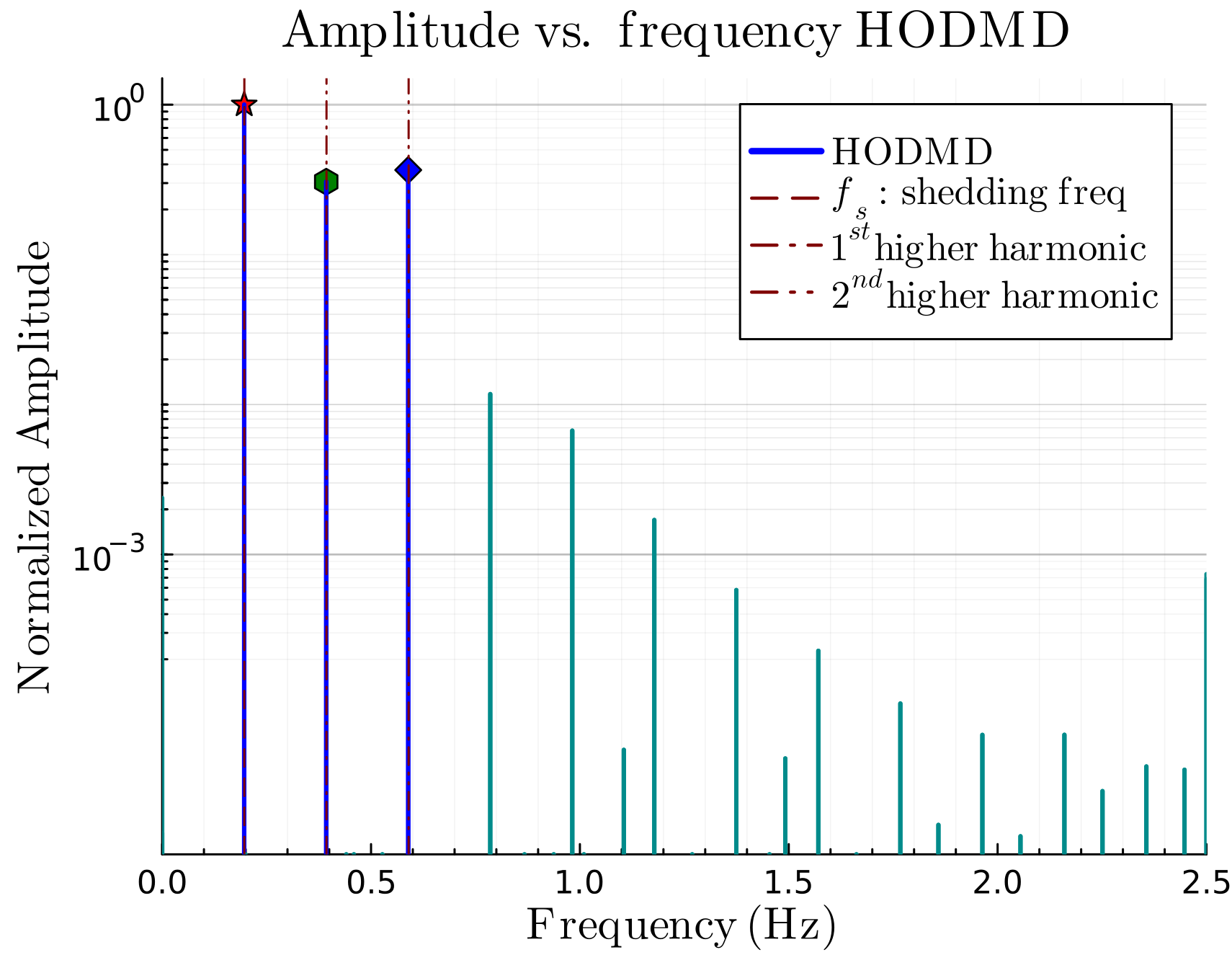}
\caption{}
\label{fig:hodmd_amp_cyl}
\end{subfigure}
\begin{subfigure}[]{0.3\textwidth}
\centering
\includegraphics[width=1\textwidth]{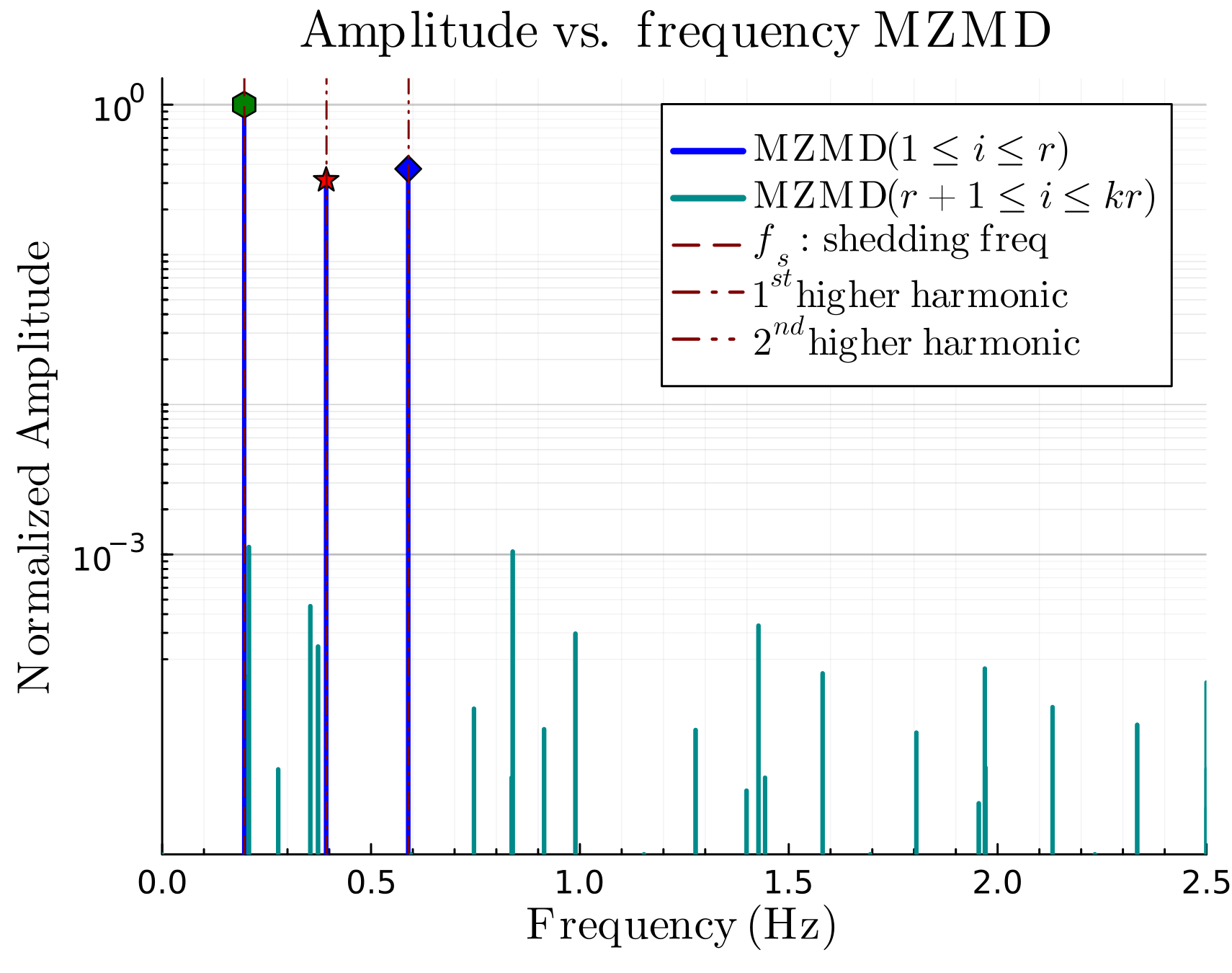}
\caption{}
\label{fig:mzmd_amp_cyl}
\end{subfigure}
\caption{Numerical validation for the $Re=200$ case, fixing $r=6$, with 10 memory terms. Amplitudes (\subref{fig:dmd_amp_cyl}, \subref{fig:hodmd_amp_cyl}, \subref{fig:mzmd_amp_cyl}) and eigenvalues (\subref{fig:dmd_eval_cyl}, \subref{fig:hodmd_eval_cyl}, \subref{fig:mzmd_eval_cyl}) of DMD, HODMD, and MZMD modes respectively, where the mode associated with the fundamental frequency $f_s \approx 0.197$ along with its first 2 higher harmonics are uniquely labeled. Each method captures the dominant mode associated with the shedding frequency $f_s$ and higher harmonics. Both MZMD and HODMD act to increase the spectral complexity over DMD for the same $r$, as seen by the introduction of new higher frequency modes not captured by DMD. Time-delay embedding terms for HODMD can capture the periodic modes associated with the higher harmonics. MZMD captures higher harmonics with low amplitudes decaying modes.} 
\label{fig:mzmd_hodmd_evals_cyl_re200}
\end{figure}

\begin{figure}[!htb]
\centering
\begin{subfigure}[]{0.3\textwidth}
\centering
\includegraphics[width=1\textwidth]{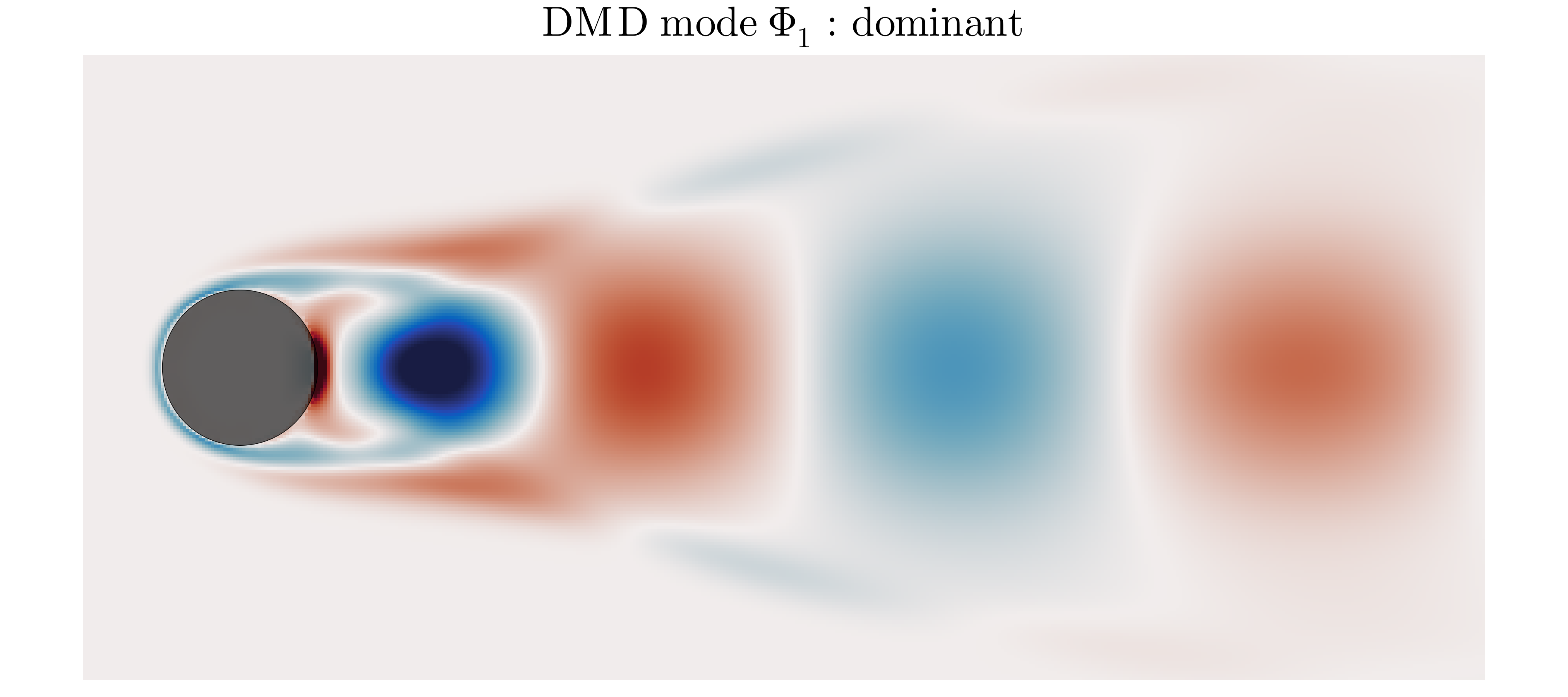}
\end{subfigure}
\begin{subfigure}[]{0.3\textwidth}
\centering
\includegraphics[width=1\textwidth]{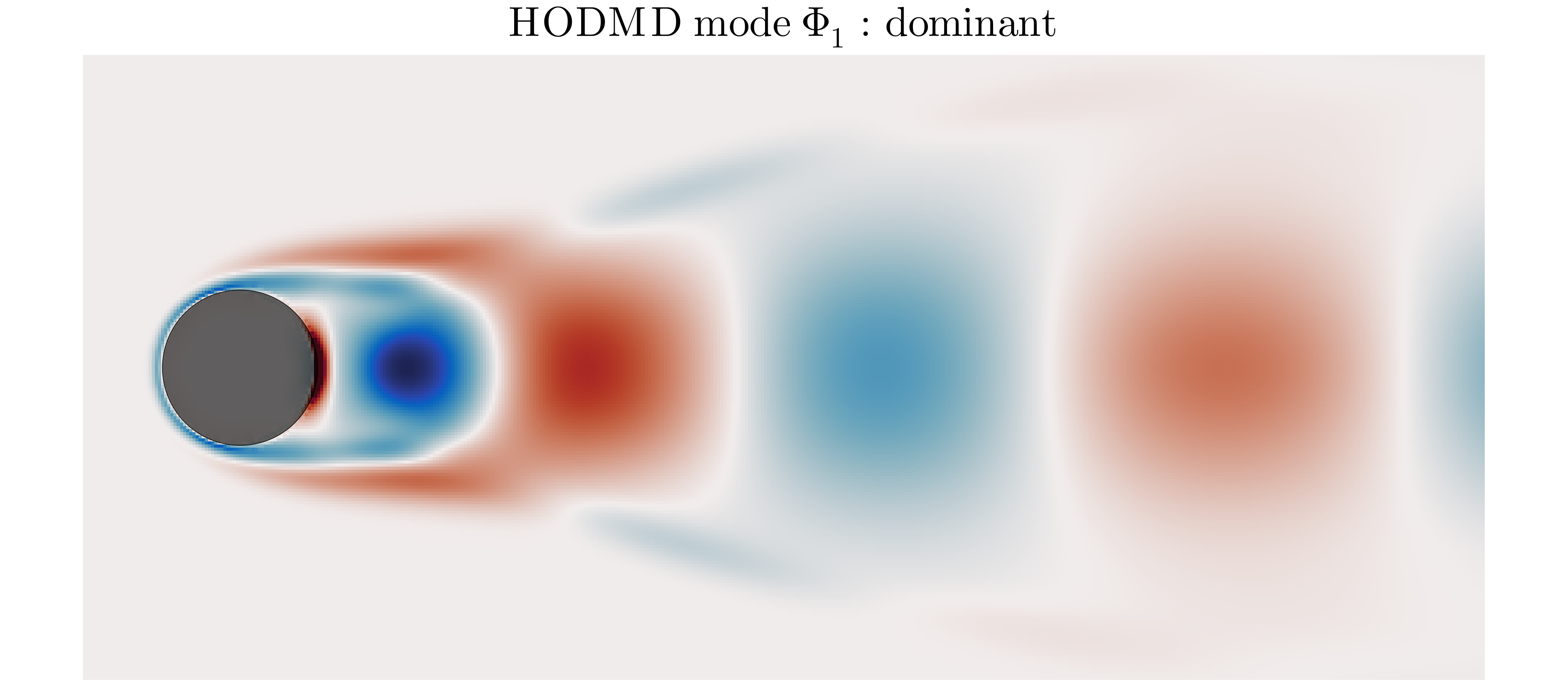}
\end{subfigure}
\begin{subfigure}[]{0.3\textwidth}
\centering
\includegraphics[width=1\textwidth]{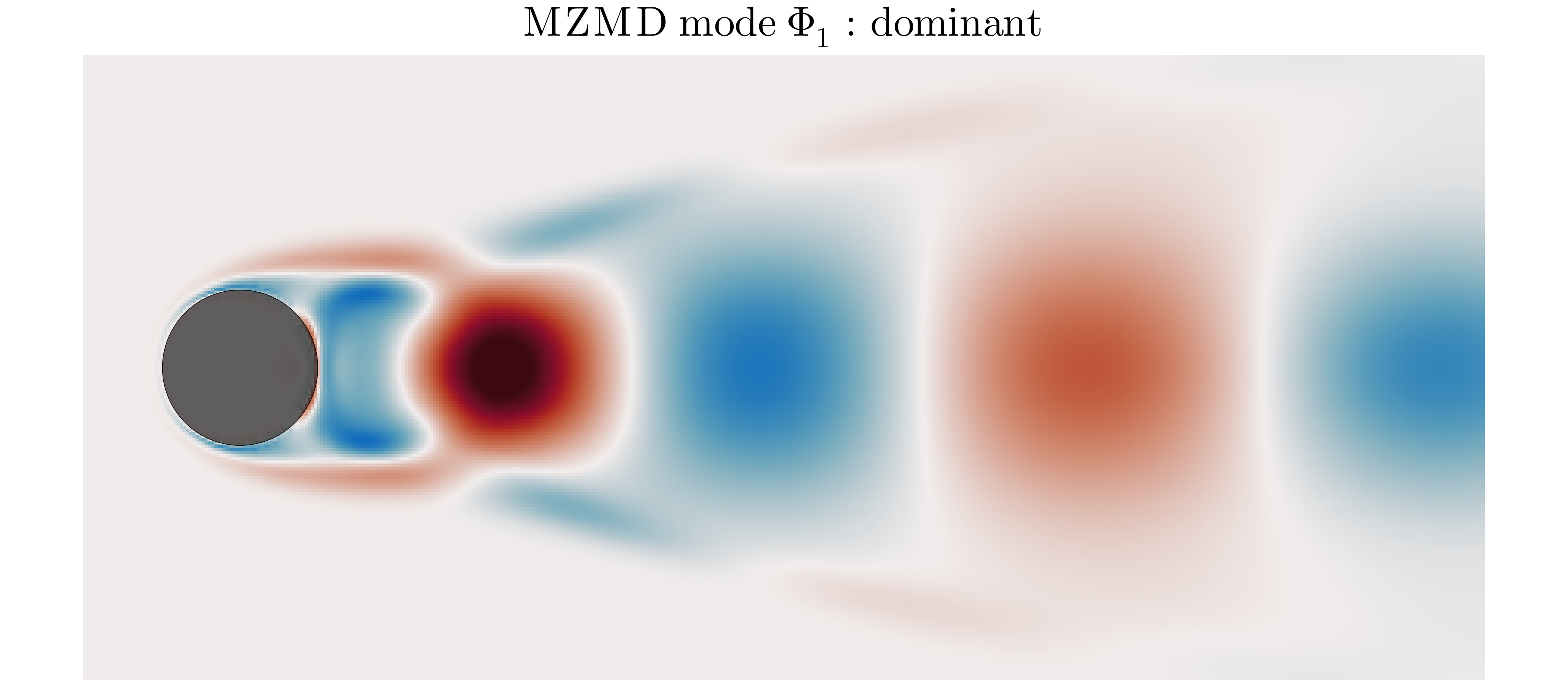}
\end{subfigure}

\begin{subfigure}[]{0.3\textwidth}
\centering
\includegraphics[width=1\textwidth]{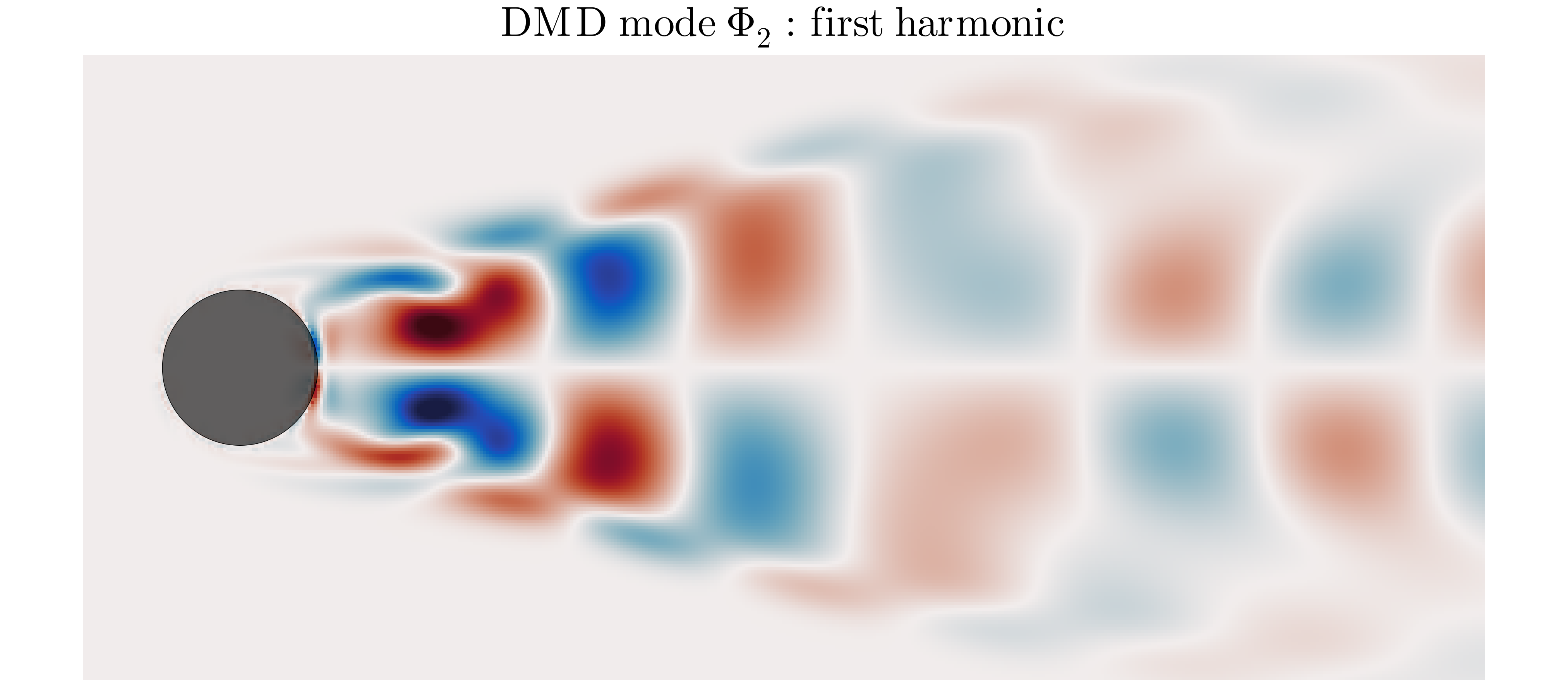}
\end{subfigure}
\begin{subfigure}[]{0.3\textwidth}
\centering
\includegraphics[width=1\textwidth]{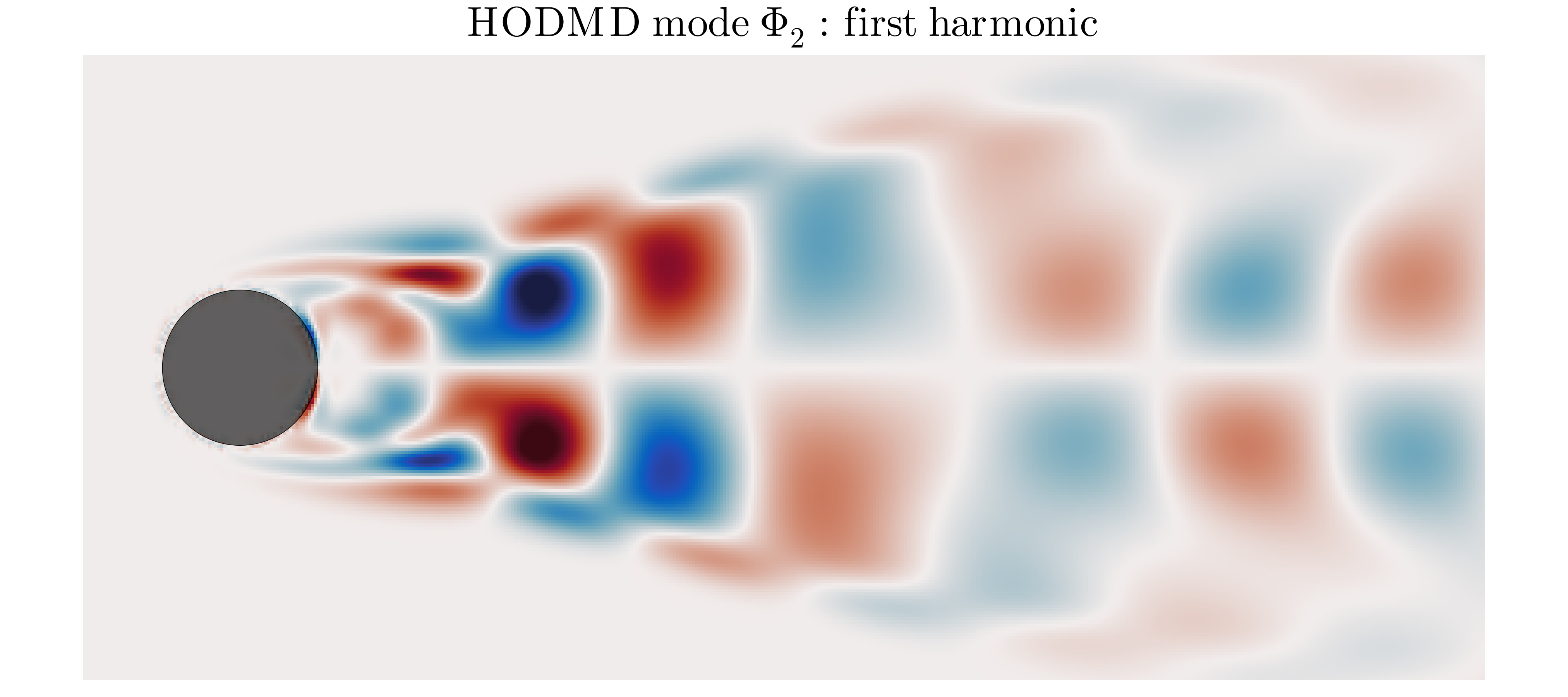}
\end{subfigure}
\begin{subfigure}[]{0.3\textwidth}
\centering
\includegraphics[width=1\textwidth]{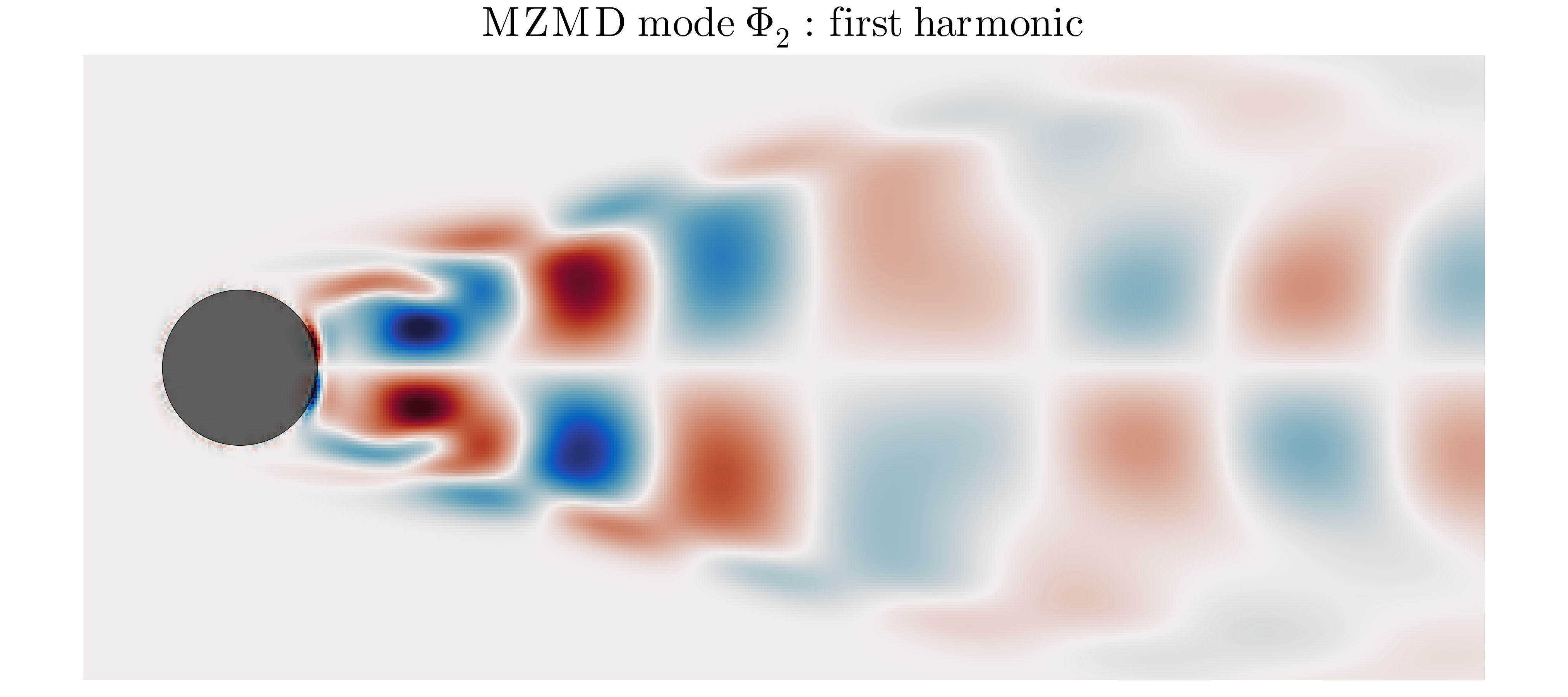}
\end{subfigure}

\begin{subfigure}[]{0.3\textwidth}
\centering
\includegraphics[width=1\textwidth]{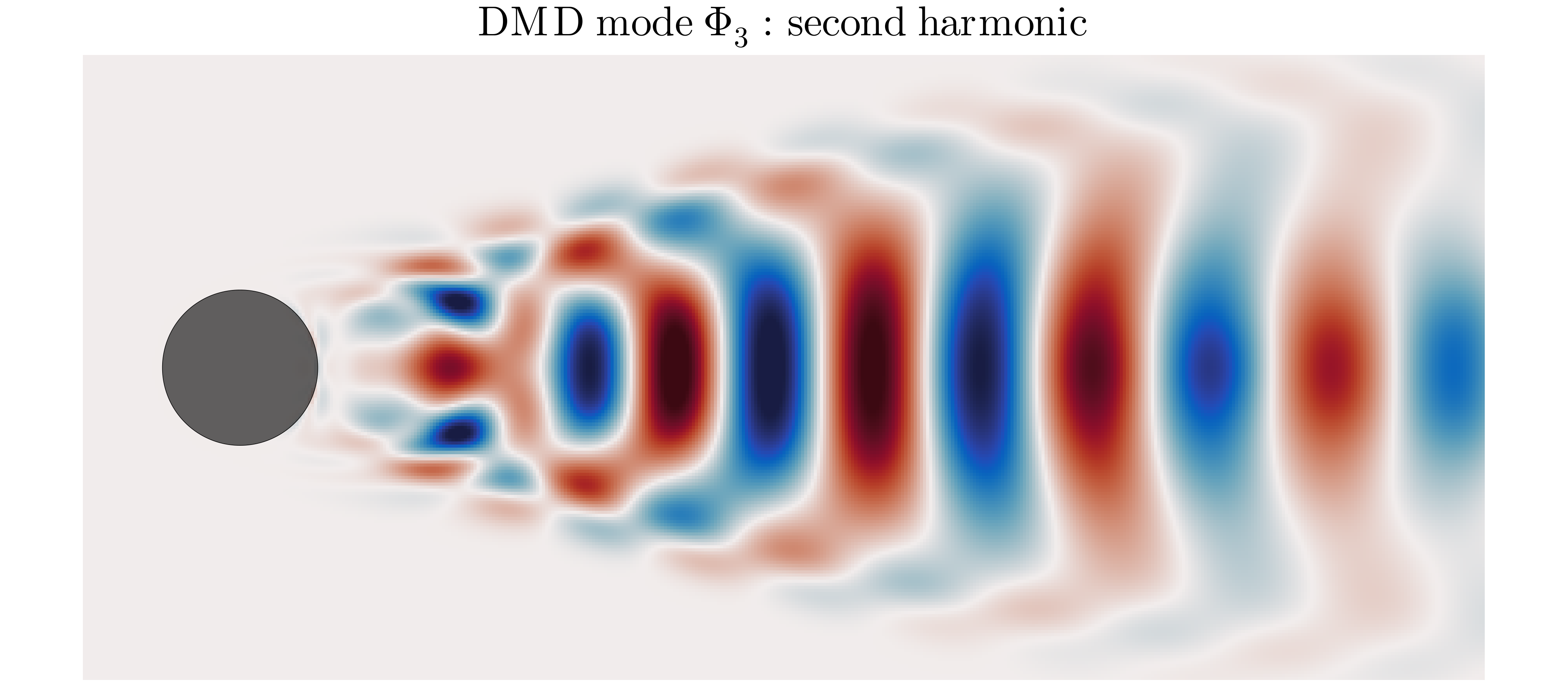}
\caption{DMD modes}
\label{fig:dmd_modes_cyl}
\end{subfigure}
\begin{subfigure}[]{0.3\textwidth}
\centering
\includegraphics[width=1\textwidth]{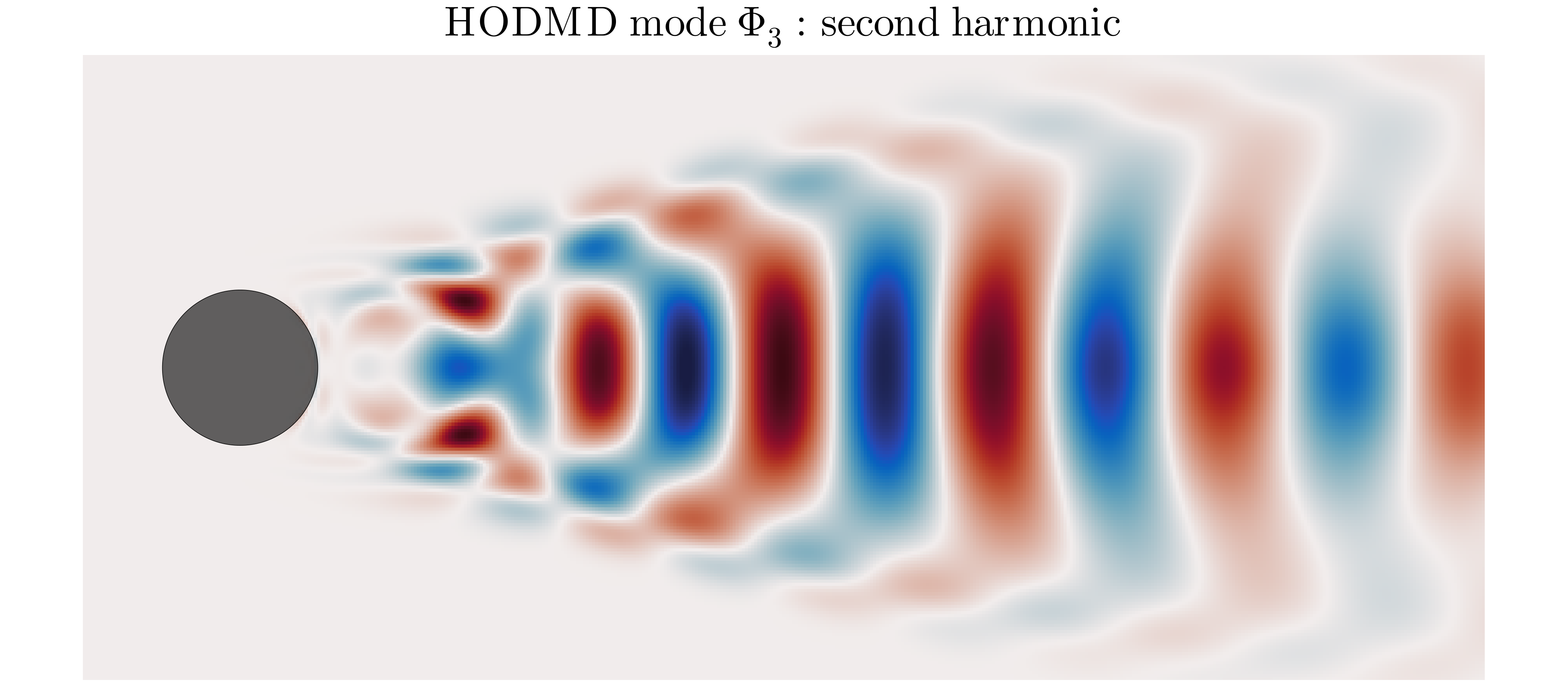}
\caption{HODMD modes}
\label{fig:hodmd_modes_cyl}
\end{subfigure}
\begin{subfigure}[]{0.3\textwidth}
\centering
\includegraphics[width=1\textwidth]{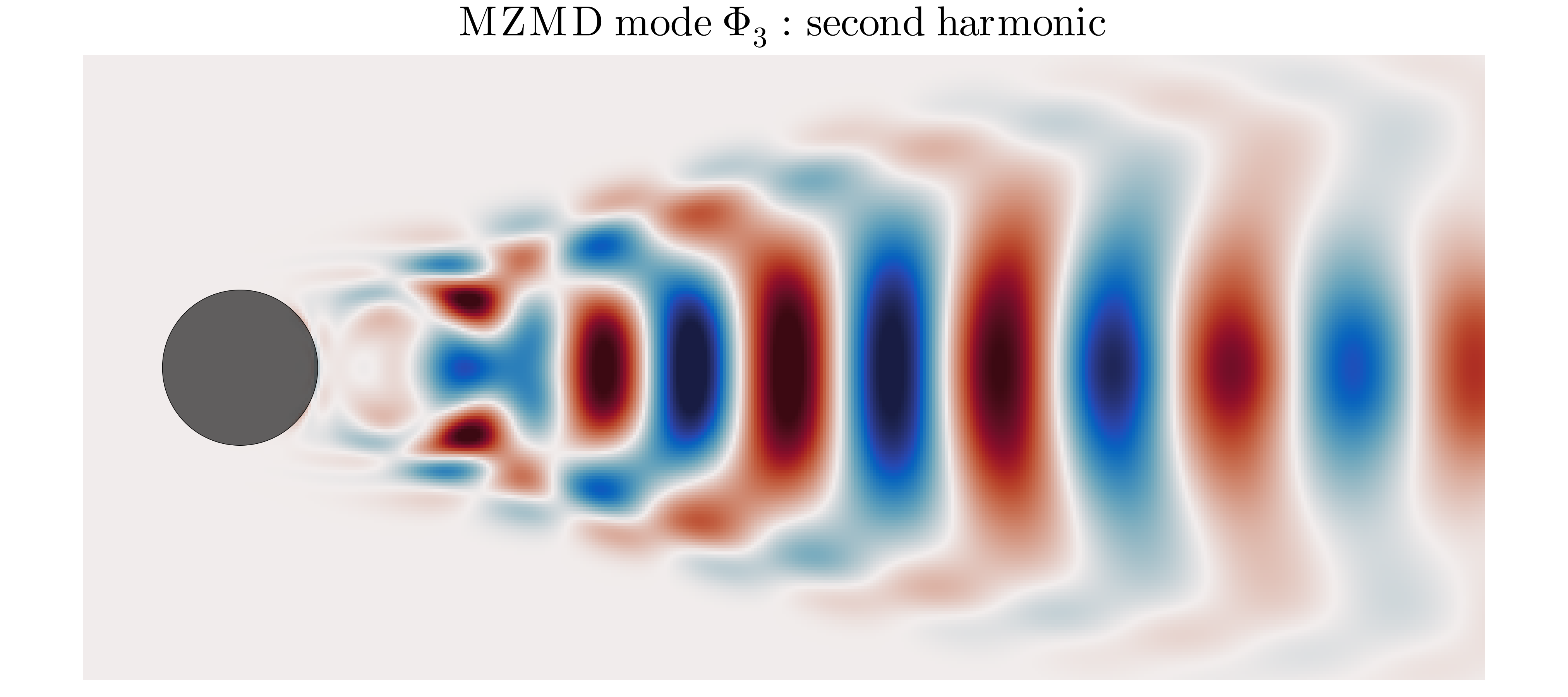}
\caption{MZMD modes}
\label{fig:mzmd_modes_cyl}
\end{subfigure}
\caption{$Re=200$. Comparing modes of (\subref{fig:dmd_modes_cyl}) DMD, (\subref{fig:hodmd_modes_cyl}) HODMD, (\subref{fig:mzmd_modes_cyl}) MZMD, with $r = 6$ and 10 memory terms. Contour levels are normalized between $(-1, 1)$. These dominant modes only differ slightly from DMD, demonstrating that adding memory (either MZ or time-delays) only introduces small perturbation to the modes obtained via DMD for this flow. This is to be expected, since DMD is sufficient to capture the dominant periodic modes seen in this oscillatory flow. } 
\label{fig:mode_comparions_cyl_re200}
\end{figure}

\begin{figure}[!htb]
\centering
\begin{subfigure}[]{0.25\textwidth}
\centering
\includegraphics[width=1\textwidth]{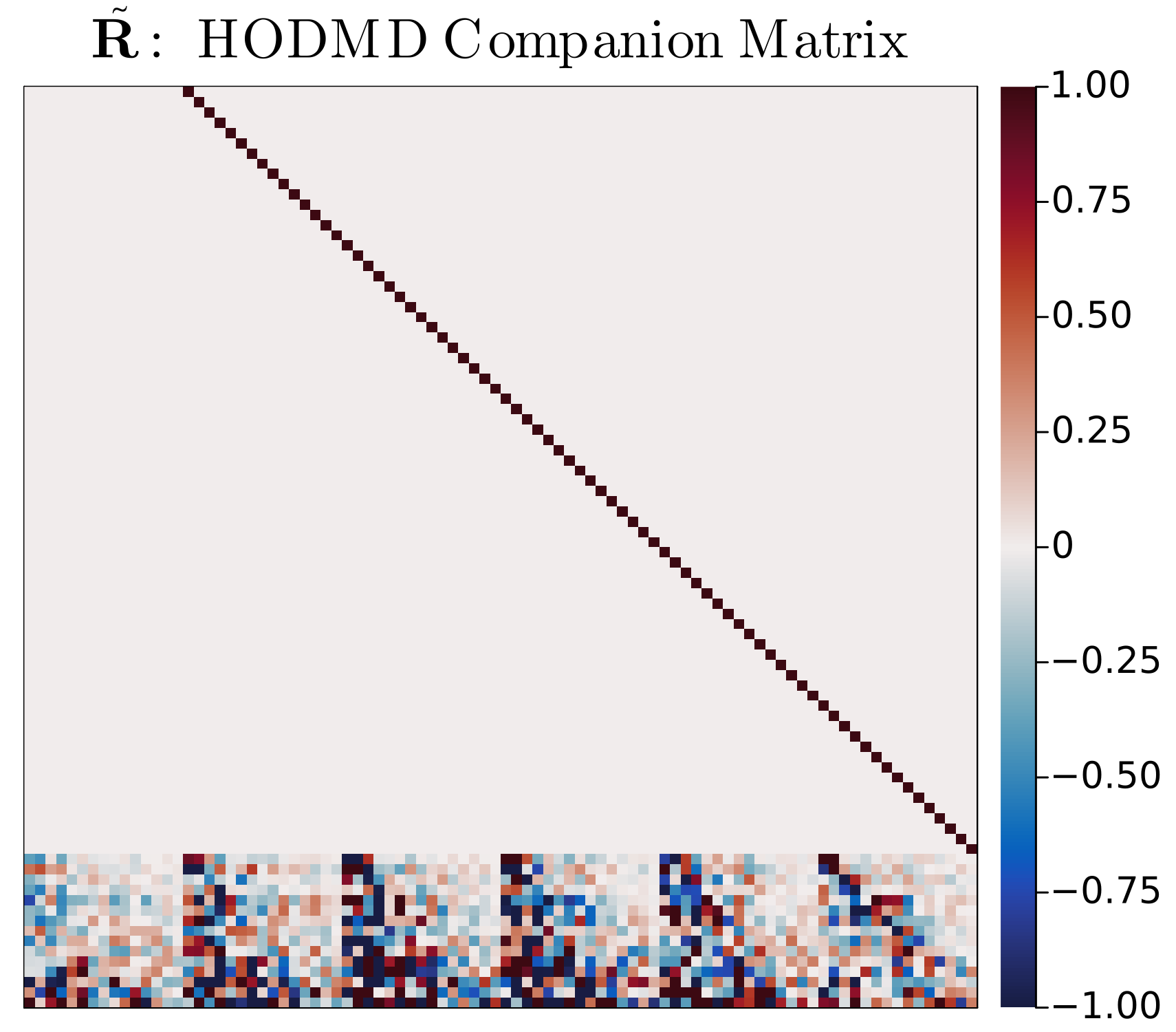}
\caption{Full second $SVD$}
\label{fig:rtil_full}
\end{subfigure}
\centering
\begin{subfigure}[]{0.25\textwidth}
\centering
\includegraphics[width=1\textwidth]{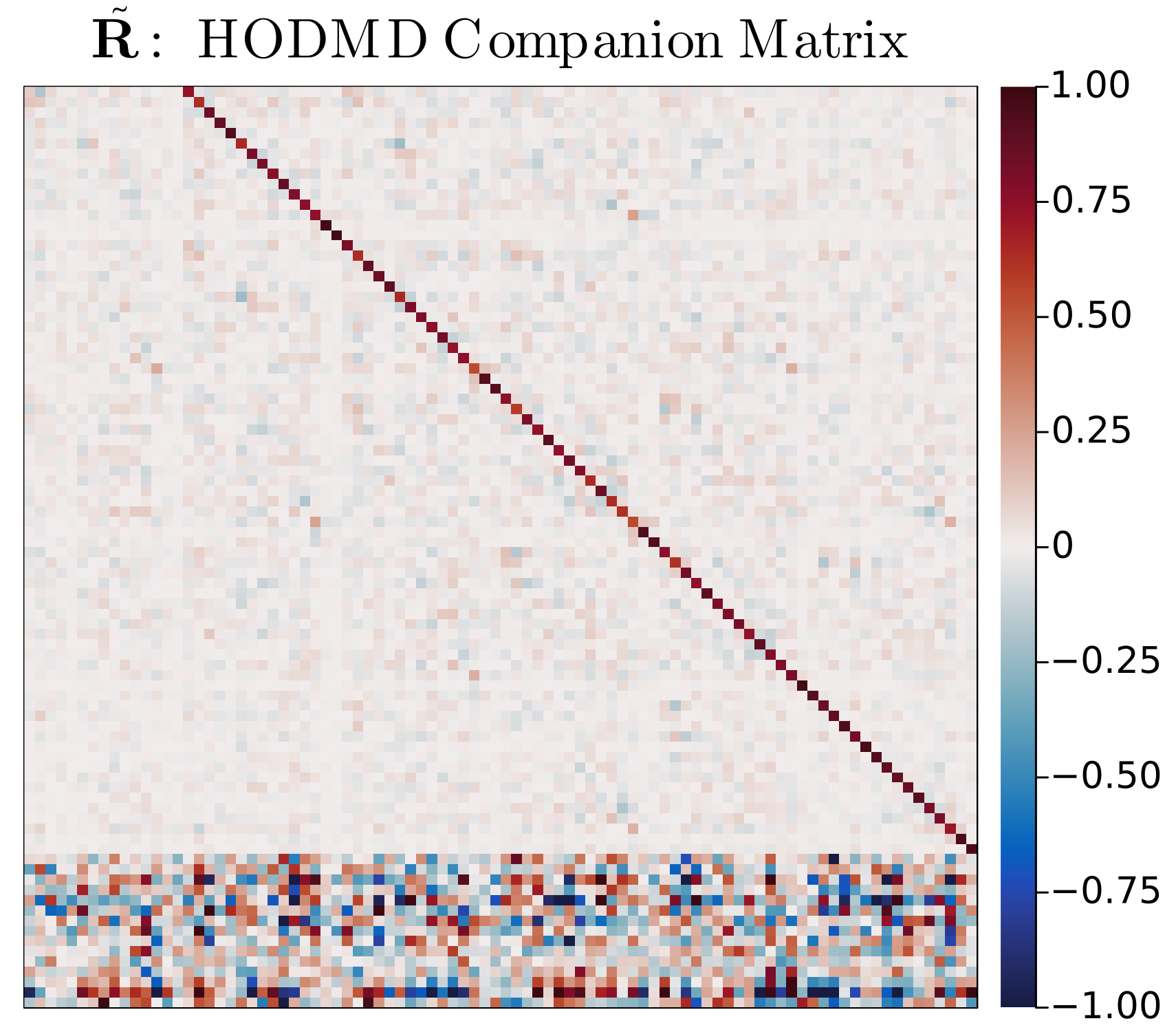}
\caption{Truncated second $SVD$}
\label{fig:rtil_trunc}
\end{subfigure}

\begin{subfigure}[]{0.3\textwidth}
\centering
\includegraphics[width=1\textwidth]{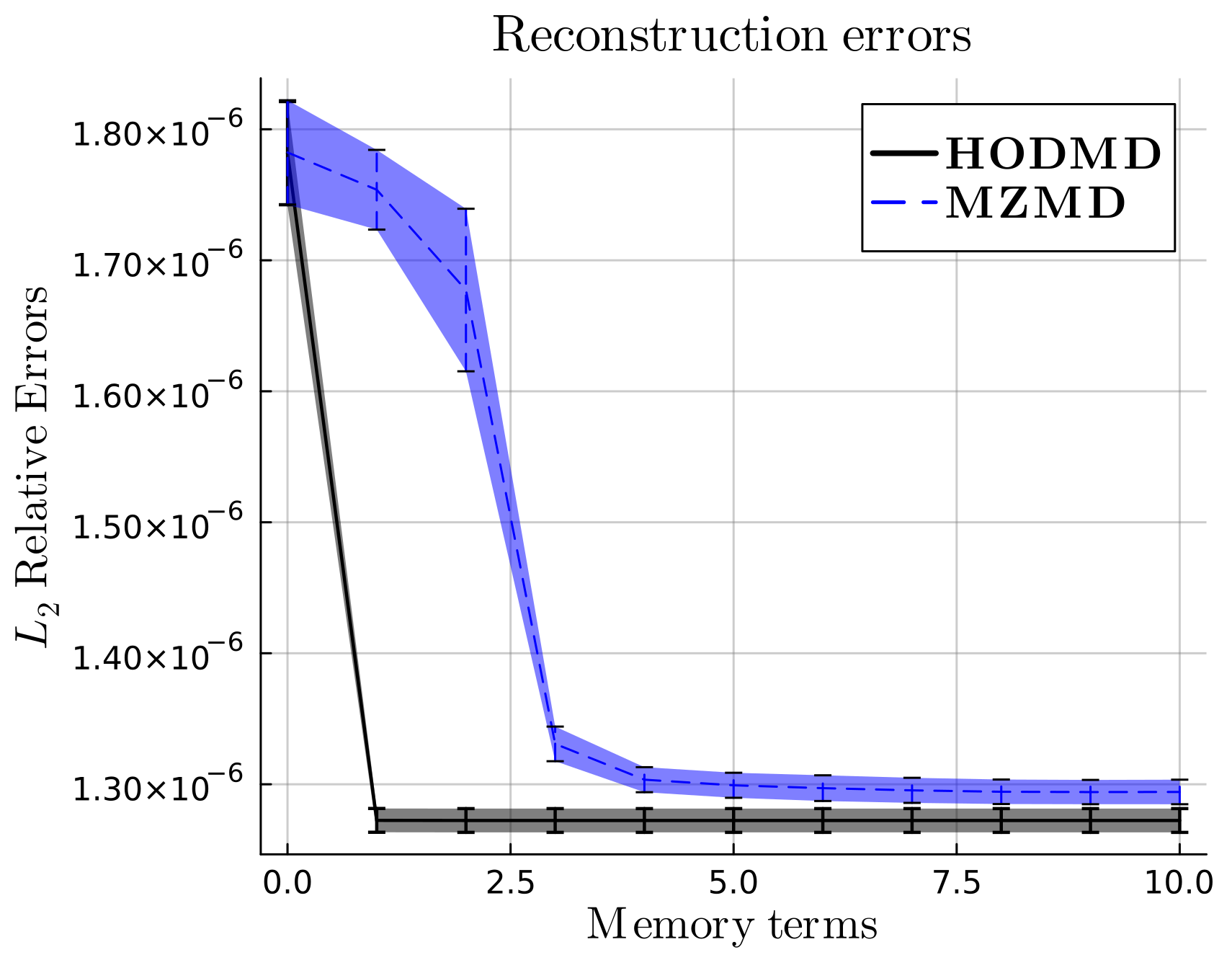}
\caption{}
\label{fig:recon_err_cyl}
\end{subfigure}
\begin{subfigure}[]{0.3\textwidth}
\centering
\includegraphics[width=1\textwidth]{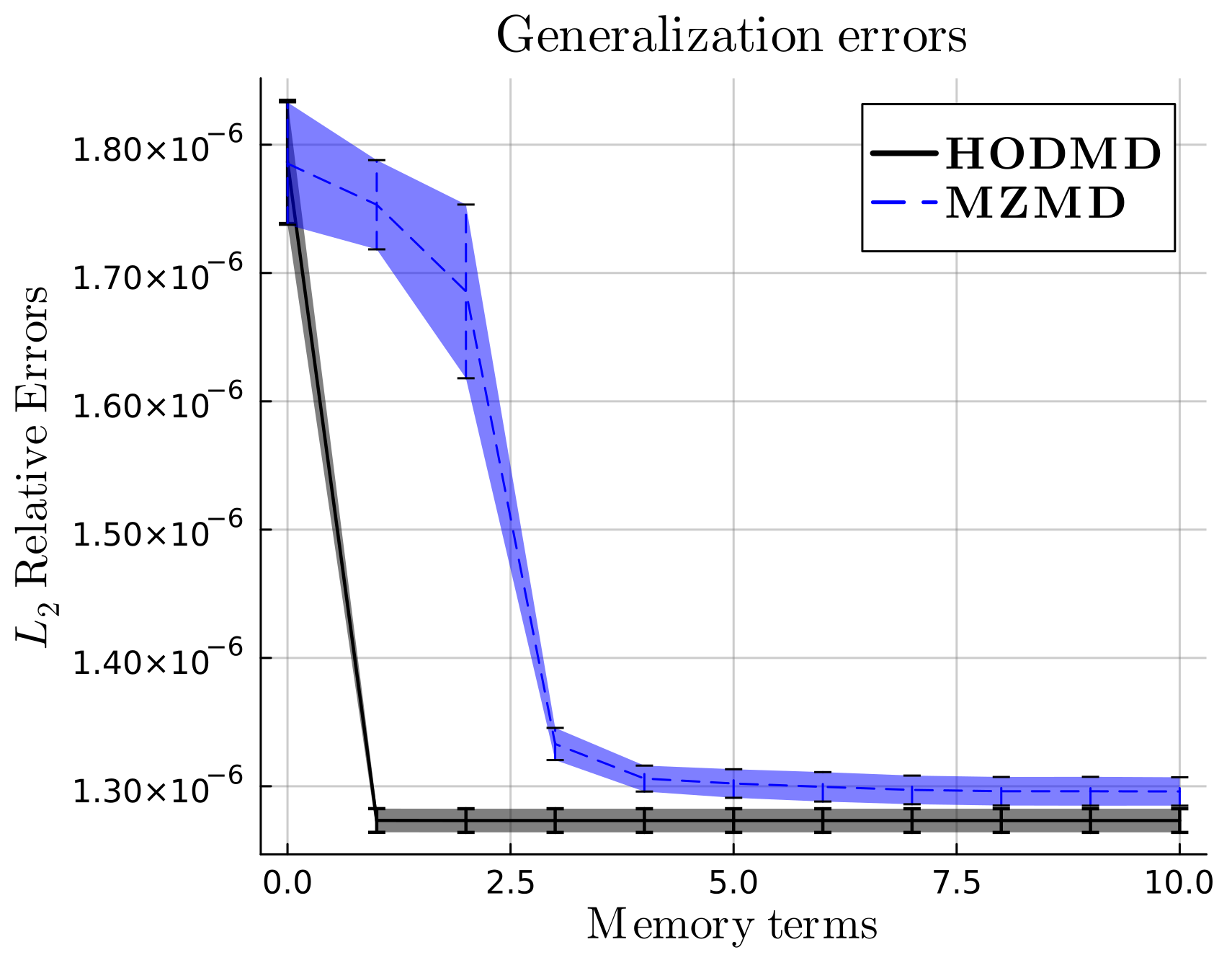}
\caption{}
\label{fig:gen_err_cyl}
\end{subfigure}
\centering
\begin{subfigure}[]{0.3\textwidth}
\centering
\includegraphics[width=1\textwidth]{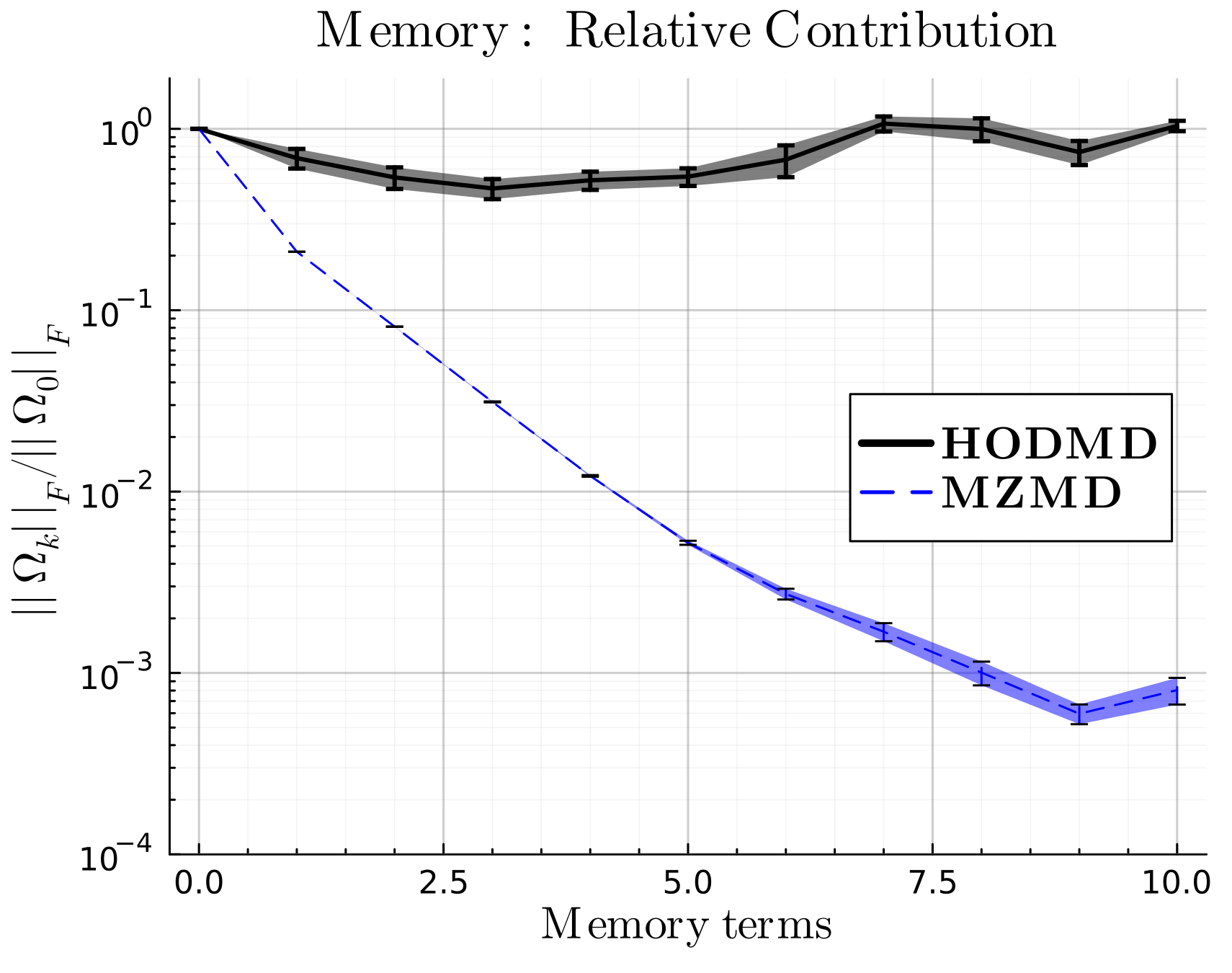}
\caption{}
\label{fig:mem_decay_cyl}
\end{subfigure}
\caption{(\subref{fig:rtil_full}, \subref{fig:rtil_trunc}) compares block companion matrix of HODMD obtained via full vs truncated second $SVD$, demonstrating that HODMD model assumption (\ref{eq:hodmd_assumption}) can be broken if the second SVD is truncated. Comparing reconstruction (\subref{fig:recon_err_cyl}) and generalization (\subref{fig:gen_err_cyl}) errors up to 3 advection time scales using $20$ samples initial conditions from DMD, HODMD, and MZMD with $r=15$. Both memory methods provide a slight improvement over DMD and converge with increased memory.  (\subref{fig:mem_decay_cyl}) compares the contributions from memory terms in HODMD and MZMD using 20 samples. The MZMD memory terms contributions decay with the order of the term, unlike HODMD, where the contributions of the higher order terms remain of order 1.} 
\label{fig:R_tilde_cylinder}
\end{figure}

\begin{figure}[!htb]
\begin{subfigure}[]{0.3\textwidth}
\centering
\includegraphics[width=1\textwidth]{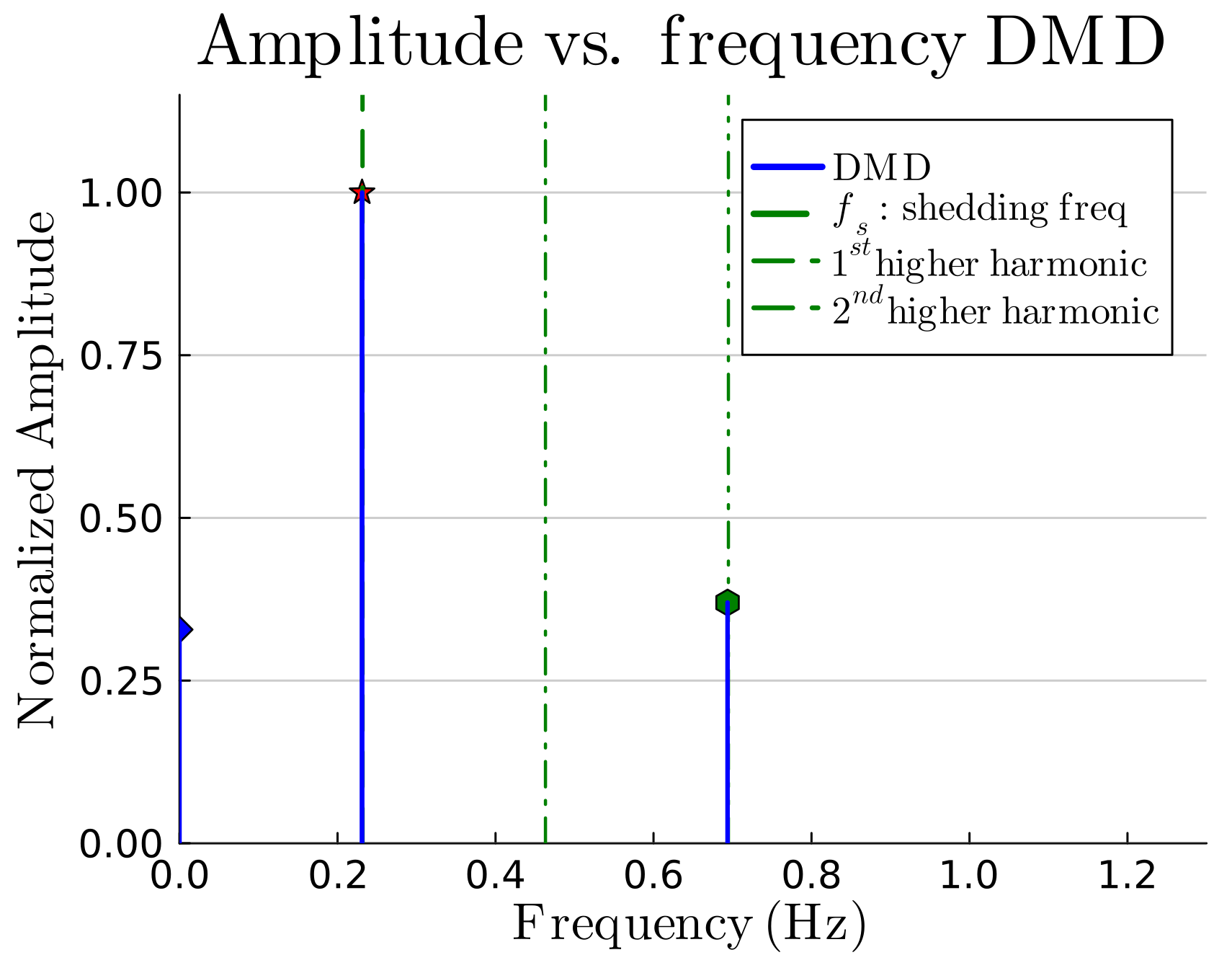}
\caption{}
\label{fig:dmd_amp_reduced}
\end{subfigure}
\centering
\begin{subfigure}[]{0.3\textwidth}
\centering
\includegraphics[width=1\textwidth]{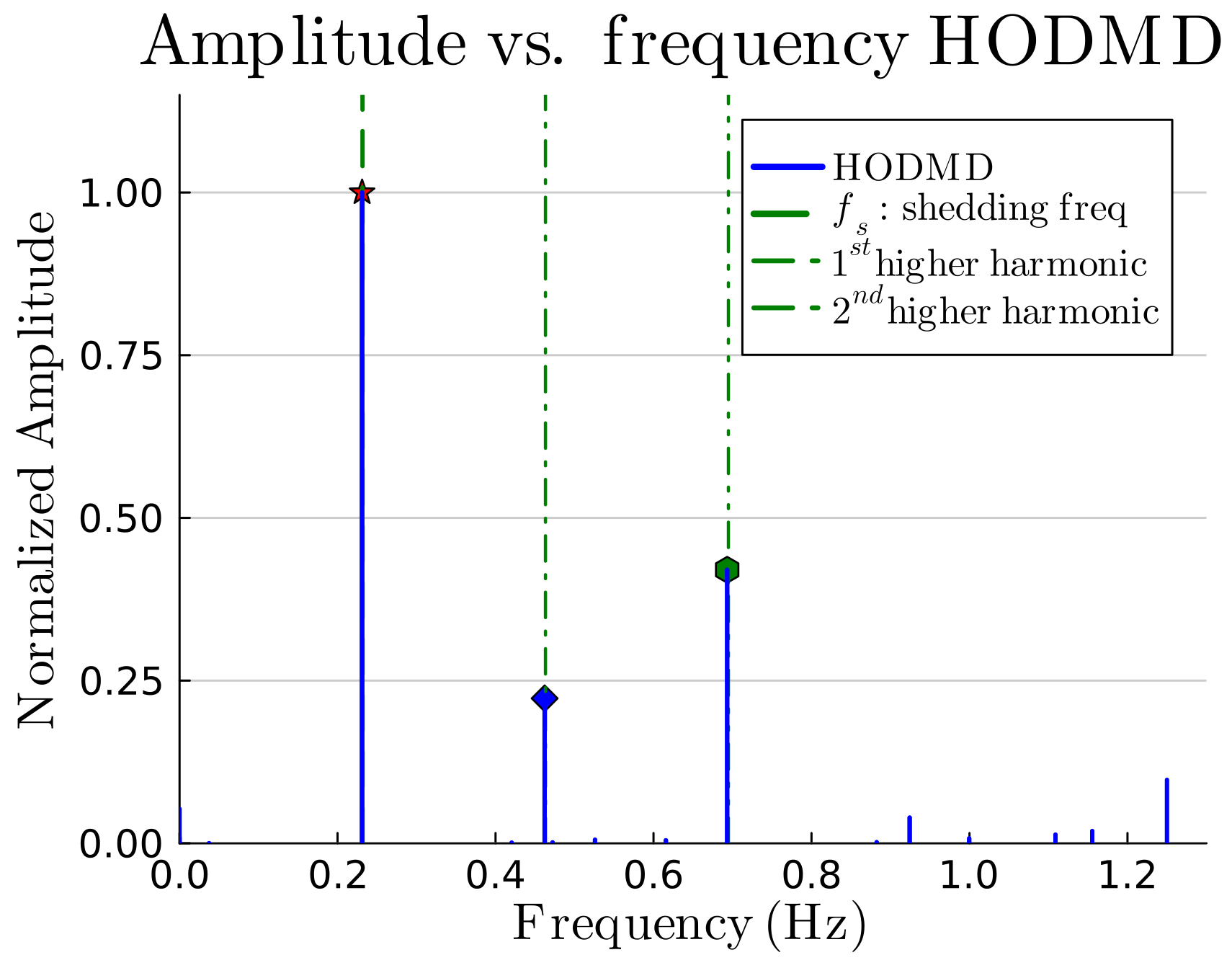}
\caption{}
\label{fig:hodmd_amp_reduced}
\end{subfigure}
\begin{subfigure}[]{0.3\textwidth}
\centering
\includegraphics[width=1\textwidth]{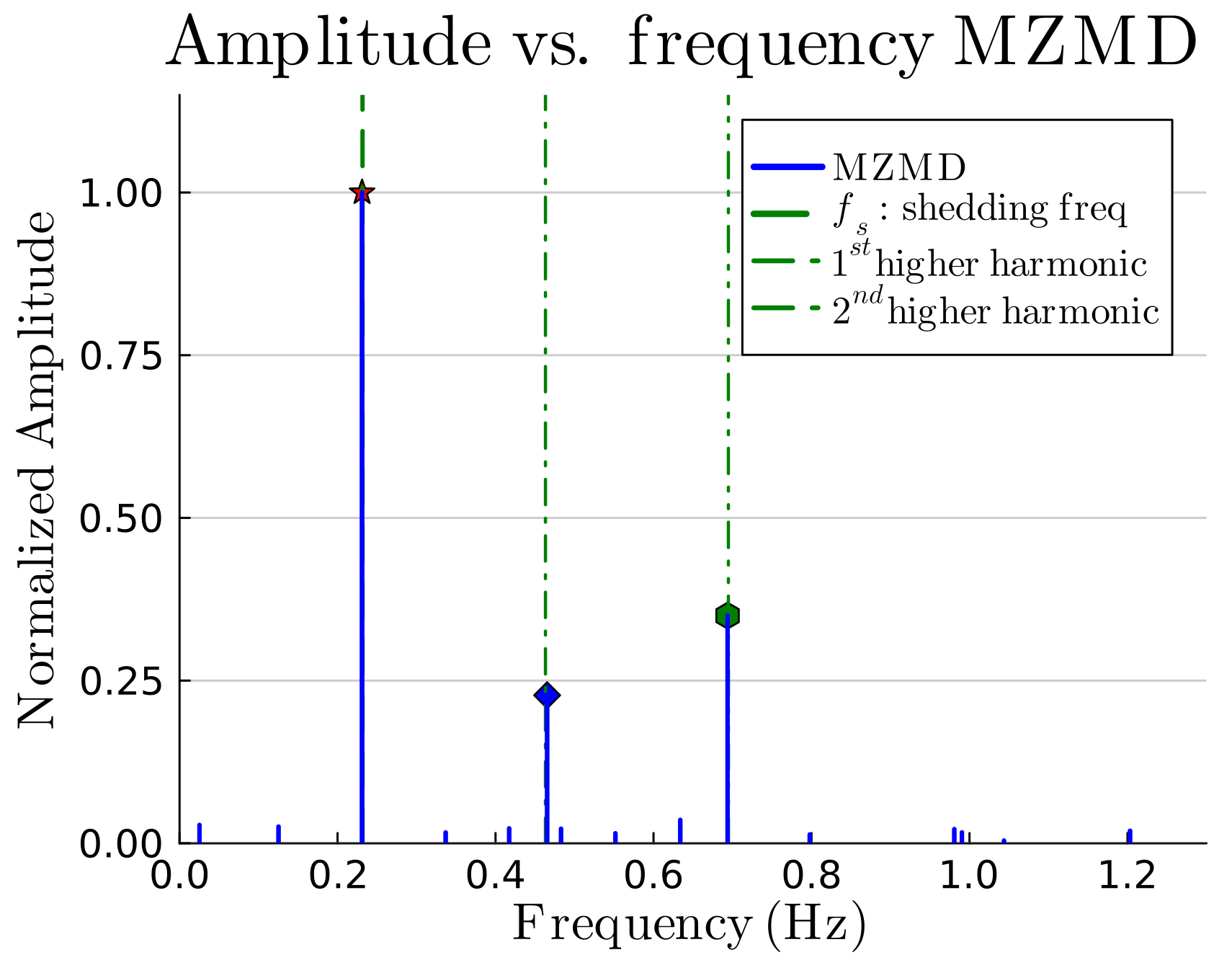}
\caption{}
\label{fig:mzmd_amp_reduced}
\end{subfigure}
\caption{Normalized amplitude vs frequency of DMD (\subref{fig:dmd_amp_reduced}), HODMD (\subref{fig:hodmd_amp_reduced}) and MZMD (\subref{fig:mzmd_amp_reduced}), fixing $r=5$, with 6 memory terms. Each captures the dominant shedding frequency, however, only HODMD and MZMD can capture both higher order harmonics. This demonstrates that adding memory can capture missing dominant higher harmonics which are not obtained from a truncated (reduced) DMD.} 
\label{fig:mzmd_hodmd_evals}
\end{figure}

\begin{figure}[!htb]
\centering
\begin{subfigure}[]{0.4\textwidth}
\centering
\includegraphics[width=1\textwidth]{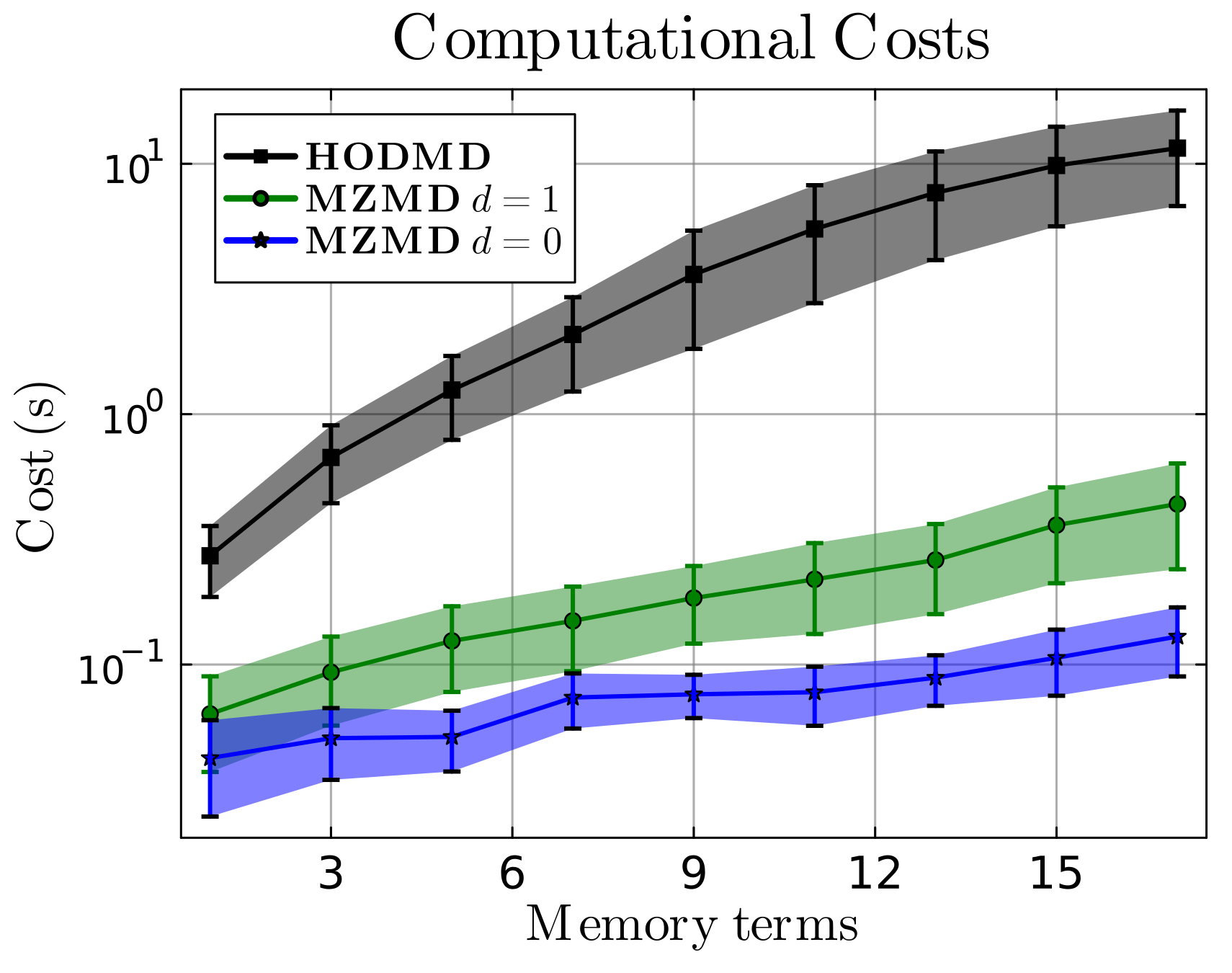}
\caption{Computational Costs}
\label{fig:costs}
\end{subfigure}
\begin{subfigure}[]{0.4\textwidth}
\centering
\includegraphics[width=1\textwidth]{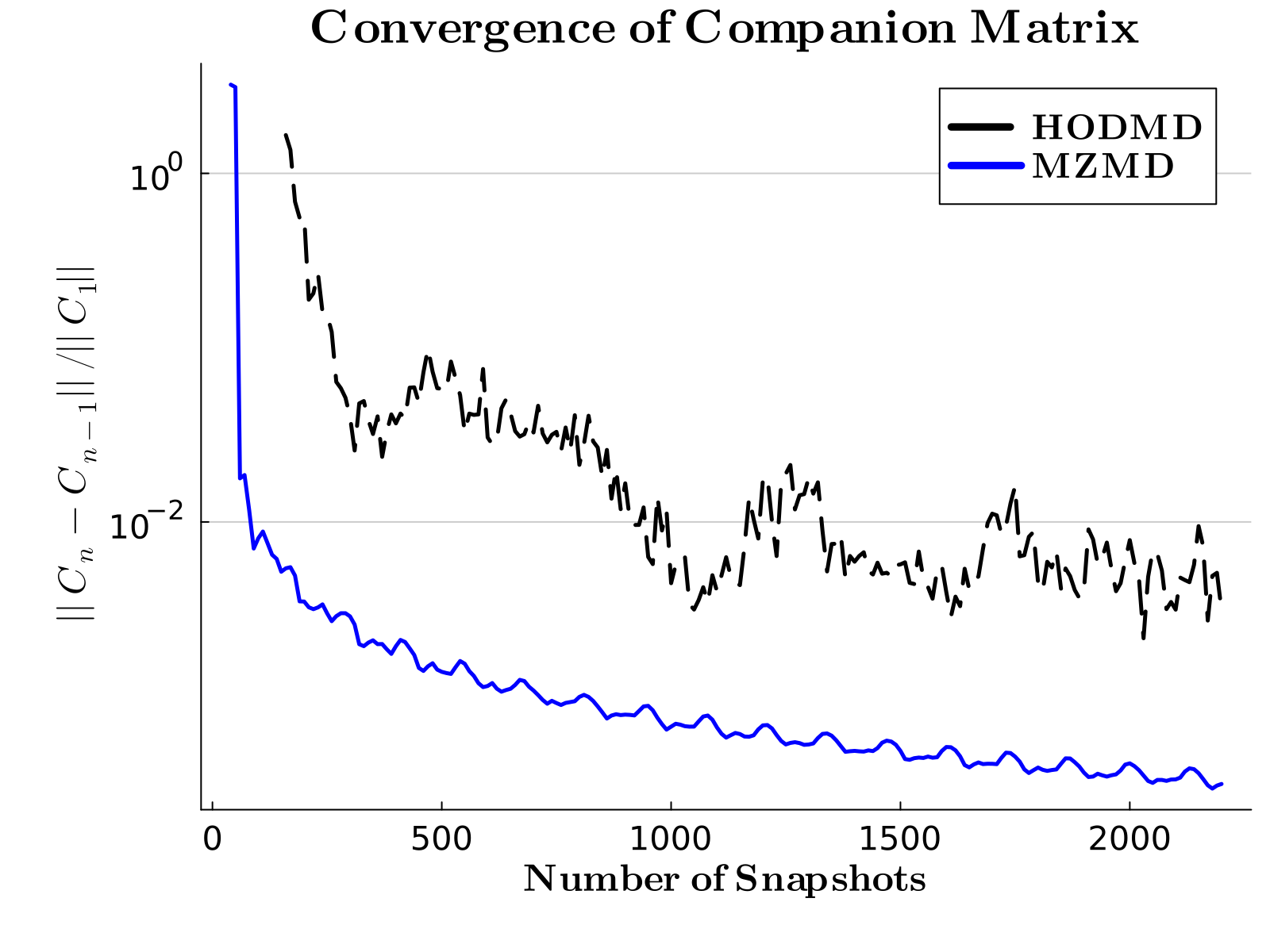}
\caption{Convergence of operators}
\label{fig:conv}
\end{subfigure}
\caption{(\subref{fig:costs}) The computational costs between MZMD (blue $\star$), MZMD with one time delay embedding (green $\circ$), and HODMD (black $\square$) over 30 independent samples of the 2D flow over a cylinder. (\subref{fig:conv}) Measuring the convergence of the learned operators with respect to the amount of training data used, showing that MZMD requires less training data for a similar level of convergence of the respective block companion matrix.} 
\label{fig:costs_convergence}
\end{figure}

\end{document}